\documentclass[11pt,lettersize]{article}

\usepackage{authblk}
\usepackage{fullpage}
\usepackage{amsmath,amsthm,amssymb,color}
\usepackage[usenames,dvipsnames,svgnames,table]{xcolor}
\definecolor{darkgreen}{rgb}{0.0,0,0.9}
\usepackage{todonotes}
\usepackage{xspace} 
\usepackage{complexity} 
\usepackage{optidef}
\usepackage{enumitem}
\usepackage{thm-restate}
\usepackage{amsmath}
\definecolor{darkblue}{rgb}{0.0, 0.0, 0.55}

\usepackage{algpseudocode}
\usepackage{algorithm}
\usepackage{tcolorbox}
\usepackage{hyperref}
\usepackage{cleveref} 
\hypersetup{
  colorlinks   = true,    
  urlcolor     = blue,    
  linkcolor    = blue,    
  citecolor    = OliveGreen      
}

\usepackage{tikz}
\usetikzlibrary{calc,fit,arrows.meta, positioning}
\usetikzlibrary{arrows.meta, positioning, shapes}

\usetikzlibrary{decorations.pathreplacing} 
\usetikzlibrary{decorations.markings}

\newtheorem{theorem}{Theorem}[section]

\newtheorem{lemma}[theorem]{Lemma}
\newtheorem{corollary}[theorem]{Corollary}

\newtheorem{claim}[theorem]{Claim}
\newtheorem{remark}[theorem]{Remark}

\theoremstyle{definition}
\newtheorem{definition}[theorem]{Definition}
\newtheorem{example}[theorem]{Example}

\newtcolorbox{wbox}
{
	colback  = white,
}
\newcommand{\remove}[1]{} 
\newcommand{\opt}{\mbox{\tt opt}} 
\newcommand{\btd}{\bigtriangledown}
\newcommand{\Vr}{V\setminus\set{r}}
\newcommand{\set}[1]{\ensuremath{\left\{#1\right\}}}
\newcommand*\samethanks[1][\value{footnote}]{\footnotemark[#1]}

\title{Equitable Core Imputations for \\
Max-Flow, MST and $b$-Matching Games}

\date{}

\author[1]{Rohith Reddy Gangam\thanks{This work was supported in part by NSF grant CCF-2230414.}}
\author[2]{Naveen Garg\thanks{Work done when the author was visiting University of Warwick on a Royal Society Visiting Fellowship.}}
\author[1]{Parnian Shahkar\samethanks[1]}
\author[1]{Vijay V.~Vazirani\samethanks[1]}

\affil[1]{University of California, Irvine}
\affil[2]{Indian Institute of Technology Delhi}

\begin{document}

\maketitle 

\begin{abstract}

	\bigskip

We study fair allocation of profit (or cost)
for three central problems from combinatorial optimization: Max-Flow, MST and $b$-matching. The essentially unequivocal choice of solution concept for this purpose would be the {\em core}, because of its highly desirable properties
. However, recent work \cite{Vazirani-leximin} observed that for the assignment game, an arbitrary core imputation makes no fairness guarantee at the level of individual agents. To rectify this deficiency, special core imputations, called {\em equitable core imputations}, were defined --- there are two such imputations, {\em leximin and leximax} --- and efficient algorithms were given for finding them. 

For all three games, we start by giving examples to show that an arbitrary core imputation can be excessively unfair to certain agents. This led us to seek equitable core imputations for our three games as well. However, the ubiquitous tractable vs intractable schism separates the assignment game from our three games, making our task different from that of \cite{Vazirani-leximin}. 
As is usual in the presence of intractability, we resorted to defining the Owen set for each game and algorithmically relating it to the set of optimal dual solutions of the underlying combinatorial problem. 
We then give polynomial time algorithms for finding equitable imputations in the Owen set. 

The motivation for this work is two-fold: the emergence of automated decision-making, with a special emphasis on fairness,  and the plethora of industrial applications of our three games.  
\end{abstract}

\newpage

\section{Introduction}
\label{sec.intro}

With automated decision-making emerging as the norm, the design of algorithms which ensure fair allocations to agents has become increasingly important; these are not only ethical, but also lead to customer satisfaction, and hence loyalty, to businesses paying attention to this issue. In this vein, our paper studies the fair allocation of profit or cost for three problems
: Max-Flow, MST and $b$-matching. All three are central to combinatorial optimization, with a plethora of industrial applications
, e.g., see \cite{waissi1994network}, \cite{ahuja1993network} and \cite{Sch-book}. The corresponding games, which distribute profit (or cost), have also received considerable attention in economics and computer science, e.g., see \cite{Moulin2014cooperative}. 

As an example, assume that several agents have built pipes for carrying flow from a source to a sink and the resulting profit needs to be distributed among the agents. What is a disciplined way of doing this?  The gold standard solution concept in this regard is that of the {\em core} which gives each sub-coalition at least as much profit as its inherent worth, hence ensuring that no sub-coalition has incentive to secede, see details in Section \ref{sec.background}. However as pointed out recently, for the assignment game \cite{Vazirani-leximin}, an arbitrary core imputation makes {\em no fairness guarantee at the level of individual agents}. Similarly, for the three games stated above, we give examples in Section \ref{sec:unfair_imputations} to show that an arbitrary core imputation can be excessively unfair to certain agents. This led us to studying {\em equitable core imputations}, proposed in \cite{Vazirani-leximin}, which comprise of two special imputations in the core: {\em leximin and leximax}.

The first work to seek a middle ground in the tension between sub-coalitions demanding profit which is commensurate with their inherent worth and the societal norm of ``equality'', was the  well-known {\em egalitarian solution} of Dutta and Ray \cite{Dutta-Ray}. For a convex game\footnote{The characteristic function of such a game is supermodular.}, this is the unique core imputation which Lorenz dominates all other core imputations. At a high level, our goal, and that of \cite{Vazirani-leximin}, is analogous to that of Dutta and Ray. Section \ref{sec.background} states three shortcomings of the approach of \cite{Dutta-Ray} as well as the way they are rectified via the approach of \cite{Vazirani-leximin} and the current paper. 

The ubiquitous tractable vs intractable schism separates the assignment game from our three  games, making our task  different from that of \cite{Vazirani-leximin}. 
The core of the assignment game consists of all optimal solutions to the dual of the LP-relaxation of the maximum weight matching problem in the underlying bipartite graph \cite{Shapley1971assignment}. Consequently, the problem of determining if a given imputation is in the core is in P. However, the latter problem was shown to be co-NP-hard for the max-flow  \cite{Fang2002computational} and MST games \cite{faigle1997complexity}. Building on these results, we show that finding a leximin or leximax core imputation is also NP-hard for these games. For the $b$-matching game,  we establish a weaker evidence of intractability\footnote{A very recent paper \cite{GTV} has established co-NP-hardness as well.}, see Section \ref{sec:b_matching_game}. 

On the other hand, for all three games, every optimal solution to the dual LP of the corresponding combinatorial  problem leads to a core imputation. Hence, the core of these games is non-empty and can be partitioned into two sets: the well-behaved, tractable part derived from optimal solutions to the dual LP and the rest; the latter is where the phenomenon of NP-hardness shows up. This situation had been encountered in the past and researchers had resorted to studying only the tractable partition\footnote{This is consistent with the parable {\em the streetlight effect}, see \cite{Streetlight}, which sometimes characterizes the way science makes progress.}, naming it the {\em Owen Set} \cite{Owen1975, Owen.Characterization}; we will carry over this name to our paper as well. For all three games, the Owen set is convex, i.e., every convex combination of two of its elements is in the set, and therefore by Lemma 3 in \cite{Vazirani-leximin}, the leximin and leximax imputations in the Owen set are unique. 

{\bf Results and key ideas:}
\begin{enumerate}
	\item For the max-flow game, we show that every optimal solution to the dual LP efficiently leads to an imputation in the Owen set. Furthermore, we give a polynomial time algorithm for determining if a given imputation belongs to the Owen set and if so, finding a corresponding optimal dual solution - the analogous result for the $b$-matching game was given in \cite{Transportation-core}. 
	\item For the MST game, we show that every optimal solution to the dual LP efficiently yields a set of core imputations, hence characterizing its Owen set. Among the various primal-dual formulations for the MST game, we establish that the minimum branching LP serves as the right formulation for defining the Owen set. Despite its exponentially many dual variables on sets, we provide an elegant interpretation that maps them to costs on agents. We also provide a separation oracle that facilitates the use of the ellipsoid method for efficiently verifying core membership. Because the procedure uses ellipsoid algorithm, it is not practical; we leave the open problem of finding efficient combinatorial algorithms.
	\item We give a combinatorial, strongly polynomial algorithm for finding the leximin and leximax core imputations in the Owen set of the max-flow game. Our approach leverages the Picard-Queyranne structure to efficiently compute dual optimal solutions corresponding to these imputations. For the $b$-matching game, we do the same by building on the algorithms of \cite{Vazirani-leximin}. 
	\item For the MST (and more generally, minimum branching) game, we give LP-based polynomial time algorithms, using the ellipsoid method, for finding the leximin and leximax imputations in the Owen set. Again we leave the open problem of finding combinatorial polynomial time algorithms. 
\end{enumerate}


\subsection{Background information}
\label{sec.background}

 
The {\em core} distributes the total worth of a game among the agents in such a way that the profit received by a sub-coalition is at least as large as the profit which the sub-coalition can generate all by itself. This not only ensures stability of the grand coalition, since no sub-coalition has an incentive to secede, but ensure some degree of fairness, namely to each of exponentially many sub-coalitions -- a stringent requirement indeed. Additionally, the core also provides profound insights into the negotiating power of individuals and sub-coalitions, see \cite{Moulin2014cooperative, Va.general} and Remark \ref{flow:negotiating_power}.

However, as pointed out recently in the context of the assignment game \cite{Vazirani-leximin}, which forms a paradigmatic setting for studying the core --- in large part due to the classic work of Shapley and Shubik \cite{Shapley1971assignment} --- an arbitrary core imputation makes no fairness guarantee at the level of individual agents. This is due to the fact that a singleton sub-coalition (or a set of players from the same side of the bipartition), can make zero profit, and therefore its profit under a core imputation can be an arbitrary amount.

It is well know that ``fairness'' can be defined in many ways, depending on the setting. Among these, the use of max-min and min-max fairness is widespread, e.g., in game theory, networking and resource allocation. A leximin allocation maximizes the smallest component and subject to that, it maximizes the second smallest, and so on. It therefore goes much further than a max-min allocation. Similarly a leximax allocation goes much further than a min-max allocation.
\cite{Vazirani-leximin} defined leximin and leximax core imputations as {\em equitable core imputations}. These two imputations achieve equality in different ways: whereas leximin tries to make poor agents more rich, leximax tries to make rich agents less rich, thereby indirectly making poor agents more rich, hence they may be better suited for different applications.

The egalitarian solution of Dutta and Ray \cite{Dutta-Ray} has three shortcomings: it does not apply to several key natural games, including the assignment, MST and max-flow games, since they are not convex; it is not efficiently computable for any non-trivial, natural game; and the Lorenz order is a partial, and not a total, order. As stated above, \cite{Vazirani-leximin} and our work rectify all three shortcomings.

\section{Related Works}
\label{sec:related_work}

The Owen set was introduced in \cite{Owen1975} for linear production games and later formally defined in \cite{Owen.Characterization}. It represents core imputations derived from dual solutions, where each dual variable corresponds to an agent's shadow price. Our definition of Owen set extends beyond linear production games to include games like min-cost spanning tree(\cite{GranotHuberman1981}) and max-flow(\cite{Kalai1982totally}) games, where not all dual variables represent shadow prices, and the prices have to be carefully defined from the dual optimal solutions. Thus, this definition provides a more inclusive framework.

\cite{Samet1987} explored Owen set imputations, referring to them as the set of dual payoffs. While Owen set generally represents a subset of core imputations, \cite{Samet1987} outlined conditions under which the core and Owen set coincide. For instance, in max-flow games with unit capacity edges, one of the conditions from \cite{Samet1987} is met, allowing our algorithms to return the leximin fair imputation among all core sets. However, this is not universally true. As we will show in later sections, finding a leximin imputation among the core set is $NP$-hard in a general max-flow game. However we will provide efficient algorithms for finding a leximin fair imputation among the Owen set. In some special classes of flow and assignment games, \cite{Granot} showed the leximin core can be found efficiently. Our work is closely related, as we extend the analysis to the general class of max-flow games. 

Extensive research has been conducted on finding fair cost-shares in a minimum spanning tree (MST) game that is a special subcase of the min-cost branching game. \cite{GranotHuberman1981} showed that the core of an MST game is never empty. Efficient procedures for computing core imputations and the nucleolus were later developed, see \cite{GranotHuberman1984}. \cite{Bird1976} introduced a cost-sharing rule that ensures each agent's cost share is within the core by distributing costs proportional to their connection costs to the MST. However, this rule may not yield a unique cost share in cases with multiple MSTs. \cite{Angel} proposed two methods to find fair core imputations within the subset governed by Bird's rule. Our work, in contrast, adopts a broader definition of fairness in the class of Owen set imputations, which extends beyond Bird’s rule. 

Besides Bird's rule, several other solutions have been proposed to allocate costs among individuals in a minimum cost network, including methods by \cite{Kar}, Folk \cite{feltkamp1994irreducible} \cite{bergantinos2007fair}, Cycle-complete\cite{trudeau2012new}, a family of strict responsive rules \cite{bogomolnaia2010sharing}, or the class of egalitarian Shapley value solutions \cite{casajus2013null}. Most of these solutions consider all relevant link costs for determining the final imputation of optimal costs. Contrary to conventional approaches, \cite{gimenez2020egalitarian} introduced a model where only a subset of network costs is considered, assuming agents possess localized knowledge with limited relevant information. They proposed a novel egalitarian approach to cost-share in minimum cost spanning games, reinterpreting spanning tree cost-share as a claims problem. However, due to its focus on local perspectives, their method does not always satisfy core properties. In their subsequent work \cite{gimenez2022claims}, the same authors developed new methods for distributing optimal network costs by framing each minimum cost spanning tree scenario as a claims problem and applying corresponding rules. However, similar to their earlier research, these methods do not guarantee core imputations. In contrast, we present an efficient algorithm that computes a leximin/leximax fair Owen set imputation, without imposing restrictions on agent assumptions. A couple of studies have explored some variations of the problem; \cite{zhan2020cost} studied cost sharing in MST games without a source node and proved that the core is non-empty. \cite{le2016generalized} introduced the Generalized Minimum Spanning Tree Game, showing that the core might be empty. \cite{bergantinos2021review} provides an exhaustive review of papers addressing the cost sharing problem in MST games.

In $b$-matching games, \cite{sotomayor1992multiple} showed that the core is non-empty. \cite{Transportation_games} and \cite{vazirani2023lpduality} demonstrated that any optimal dual LP solution is a core imputation, though these imputations do not fully characterize the core for $b$-matching games as they do for assignment games. \cite{Vazirani-leximin} recently proposed a combinatorial algorithm for finding the leximin/leximax core of the assignment game, and we show that with modifications, it can be utilized to find the leximin/leximax Owen set imputation in $b$-matching games. \cite{biro2018stable} proved that core membership testing in edge-constrained $b$-matching is co-NP-complete, and very recently, \cite{GTV} proved this for standard $b$-matching game as well. 

In cooperative game theory, the nucleolus is a solution concept related to the core and fair imputations, designed to balance dissatisfaction among coalitions by minimizing the maximum dissatisfaction lexicographically. Unlike the core, which can sometimes be empty, the nucleolus always exists. \cite{Schmeidler} showed that when the core is non-empty, the nucleolus belongs to the core. \cite{Raghavan} provided an efficient algorithm for determining the nucleolus in assignment games. However, for general cooperative games including the three games studied in this paper, there is no known polynomial-time algorithm for finding the nucleolus - it is in fact NP-hard for both MST games(\cite{MST_nucleolus_np-hard}) and max-flow games(\cite{Flow_nucleolus_np-hard}). This challenge motivates the exploration of alternative fairness concepts, such as the leximin/leximax core employed in this paper. In general, the leximin/leximax core imputation does not coincide with the nucleolus. Nevertheless, \cite{Granot} demonstrated that under specific conditions in max-flow games, the leximin core and the nucleolus are equivalent.

\section{Preliminaries}
\label{sec:prelim}
In this section, we define notions that are used repeatedly in the paper.

\begin{definition}
    Let $N$ be a set of agents. A cooperative game on $N$ is defined by a {\em characteristic function} $v: 2^N \rightarrow \mathbb{R}_+$, where for each $S \subseteq N$, $v(S)$ is the value (or cost) that the sub-coalition $S$ can produce on its own. $N$ is also called the grand coalition.
\end{definition}

In the three games that we consider - the max-flow game, the min-cost branching game and the max-weight $b$-matching game - the set of agents is either the vertex or the edge set of a graph $G$. We will use $(G,v)$ to represent these games, where $v$ is the characteristic function of the specific game.  

\begin{definition}
An \textit{imputation} $p: N \to \mathbb{R}_+$ is a partition of the value(/cost) of the game, $v(N)$, among the agents in $N$. $p(i)$ is the {\em share} of agent $i$ and $\sum_{i \in N} p(i)=v(N)$.
\end{definition}

Depending on the game we might refer to $v(S)$ as the profit or cost of the coalition $S\subseteq N$. The corresponding games are called profit or cost-sharing games. 

\begin{definition}\label{defCore}
An imputation $p$ is in the {\em core} of a profit(/cost)-sharing game $(N,v)$ if and only if for every sub-coalition $S \subseteq N$, the total profit(/cost) shares of the members of $S$ is no less(/more) than the value of $S$, i.e., $\sum_{i \in S} p(i)$ is at least(/at most) $v(S)$.
\end{definition}

\begin{definition}
Let $P$ be a set of imputations of a game $(N,v)$ and $p_1,p_2 \in P$. Let $l_1,l_2$ be the lists formed by arranging the shares of agents in $p_1,p_2$ in ascending order. $l_1$ is {\em lexicographically larger} than $l_2$ if $l_1$ has the larger value at the first index where the two lists differ. The imputation in $P$ which is lexicographically larger than all other imputations in $P$ is the {\em lexicographically minimum} or {\em leximin} imputation in $P$.  
\end{definition}

\begin{definition}
Let $P$ be a set of imputations of a game $(N,v)$ and $p_1,p_2 \in P$. Let $l_1,l_2$ be the lists formed by arranging the shares of agents in $p_1,p_2$ in descending order. $l_1$ is {\em lexicographically smaller} than $l_2$ if $l_1$ has the smaller value at the first index where the two lists differ. The imputation in $P$ which is lexicographically smaller than all other imputations in $P$ is the {\em lexicographically maximum} or {\em leximax} imputation in $P$.  
\end{definition}

If the set $P$ in the definitions above is the set of core imputations in the game, then the imputations are called \emph{leximin} and \emph{leximax core imputation} respectively. Similarly, the respective imputations for the Owen set imputations will be called the \emph{leximin} and the \emph{leximax Owen set imputations}. 

\subsection{The games and their unfair core imputations}

\subsubsection{Max-flow game}
An instance of the $s$-$t$ maximum flow problem is given by a directed graph $G = (V, E)$, a source vertex $s\in V$, a sink vertex $t\in V$ and edge capacities $c: E \to \mathbb{R}_+$. A flow $f: E \to \mathbb{R}_+$ is a function on the edges of $G$ that satisfies the following constraints.
\begin{description}
    \item[Capacity constraint:] The flow $f_e$ through any edge $e$ must not exceed its capacity $c_e$, i.e. $\forall e\in E, f_e\le c_e$.
    \item[Conservation:] For every vertex $v$ except $s$ and $t$, the flow entering $v$ equals the flow leaving $v$, i.e., $\forall v\in V\setminus\set{s,t}, \sum_{u:(u,v)\in E} f_{uv} = \sum_{w:(v,w)\in E} f_{vw}$.
\end{description}
The value of flow $ f $ is the sum of the flow on the edges leaving source $ s $ or entering sink $ t $ (these sums are equal due to flow conservation). The objective is to find a flow $f$ of maximum value. 

The max-flow game is defined over an instance of the maximum $s-t$ flow problem and has an agent for each edge in $G$. The value/profit of the grand coalition, $v(E)$ is the maximum $s$-$t$ flow. The profit of a subset of agents $E'\subseteq E$, is the maximum $s$-$t$ flow in the subgraph induced by the edges $E'$, with the same capacities of edges. An imputation is a distribution of the total profit to the edges/agents and an imputation is in the core if no subset gets a profit less than its value.
\subsubsection{MST game and min-cost branching game}

Let $G=(V, E)$ be a directed graph, $c: E\rightarrow\mathbb{R}_+$ a cost function on the edge and $r\in V$ a root vertex. A {\em branching} is a subset of edges $E'\subseteq E$ such that there is a path from every vertex in $\Vr$ to $r$ using edges in $E'$. The cost of a branching is the sum of the costs of the edges in the branching. The value/worth of a set $S\subseteq \Vr$ , $v(S)$, corresponds to the minimum cost branching in $G(S\cup {r})$ and defines the characteristic function of the game. 

Similarly, if the graph $G$ is undirected and the worth of a set $S$ is given by the minimum cost spanning trees involving the vertices of $S$ and the root, then the game is called min-cost spanning tree(MST) game. This is trivially a special case of the min-cost branching game.
We will use $(G,v)$ to define the MST or min-cost branching game - the vertices of the graph, minus the root, cover the player set and the min-cost branching or the spanning tree formed by the vertices with the root fix the characteristic function of the game.

Let $T$ be a branching(/spanning tree) of the minimum cost in $G$. An assignment $s:V\rightarrow\mathbb{R}_+$ is a {\em cost-share} or a {\em imputation} if $\sum_{v\in\Vr} s(v)=c(T)$. An imputation is in the core of the min-cost branching(/MST) game if $\forall S\subseteq\Vr, \sum_{e\in T'} c(e) \ge \sum_{v\in S} s(v)$ where $T'$ is the minimum cost branching(/MST) in the subgraph induced over $S\cup\set{r}$.

\subsubsection{Max-weight bipartite \texorpdfstring{$b$}{b}-matching game}

Consider a bipartite graph $ G = (U, V, E) $ with an associated edge-weight function $ w: E \rightarrow \mathbb{R}_+ $. The capacity function $ b: U \cup V \rightarrow \mathbb{Z}_+ $ sets a cap on the number of matches a vertex can participate in. While edges may be matched multiple times, the vertex limitations dictated by $ b $ inherently set restrictions on the edges as well. Consequently, edge $ (i, j) $ can be matched up to $ \min \{b(i), b(j)\} $ times. Any choice of edges, with multiplicity, subject to these constraints, is called a $b$-matching. 

In the context of the \textit{max-weight bipartite $ b $-matching game}, ``the $b$-matching game'' for short, the \textit{value} of a coalition $ (S_u \cup S_v) $, where $ S_u \subseteq U $ and $ S_v \subseteq V $, is defined by the maximum weight of a $ b $-matching within the subgraph of $ G $ limited to $ (S_u \cup S_v) $ alone. This value is represented by $ v(S_u \cup S_v) $, which forms the \textit{characteristic function} of the game, with $ v: 2^{U \cup V} \rightarrow \mathbb{R}_+ $. An \textit{imputation} is made up of two mappings $ p_U : U \rightarrow \mathbb{R}_+ $ and $ p_V : V \rightarrow \mathbb{R}_+ $ ensuring that $ \sum_{u \in U} p_U(u) + \sum_{v \in V} p_V(v) = v(U \cup V) $. The definition of the core, as stated in Definition \ref{defCore}, remains the same. 

The particular instance of the bipartite $ b $-matching game where $ b $ is a constant function is called the \textit{uniform bipartite $ b $-matching game}, with the constant represented by $ b_c \in \mathbb{Z}_+ $. If $b(v) = 1, \forall v\in U\cup V$, then it is the special case of \textit{assignment game}, which is the focus of \cite{Shapley1971assignment} and \cite{Vazirani-leximin}.

\subsubsection{Unfair imputations for these games}
\label{sec:unfair_imputations}
In each of these games, there are core imputations that may be considered very unfair. For example, in the max-flow game, consider a path of $n$ unit-capacity edges from the source to the sink. An imputation that allocates all the profit to a single edge while giving nothing to the others is in the core, but a fairer imputation would distribute the profit equally among all edges, giving each a share of $1/n$. 

Similarly, in the MST game, imagine $n$ vertices connected to the root by a path of $n$ unit-cost edges. An imputation that charges all the cost to the farthest vertex is in the core, but a more equitable choice would be to charge each agent one unit of cost. Shapley and Shubik (\cite{Shapley1971assignment}]) demonstrated that in the assignment game, a special case of $b$-matching game, the set of core imputations corresponds precisely to the optimal solutions of the dual LP-relaxation of the maximum weight matching problem. Moreover, they showed that this set of core imputations forms a lattice, where the extreme points tend to disproportionately favor one side. 

Proofs for all lemmas and theorems marked with a $\dagger$ are provided in Appendix~\ref{app:proofs}.

\section{Computing the leximin optimum solution to a linear program}
\label{sec:lp_leximin_solution}

\begin{mini}
		{} {c^T x }
			{\label{lp_general_1}}
		{}
        \addConstraint{Ax}{= b}{}
        \addConstraint{x}{\geq 0}{}
\end{mini} 

Consider the above optimization linear program(LP) on the set of variables $X=\{x_1,x_2, \ldots, x_n\}$. It is easy to see that an LP has a unique leximin optimum solution${}^\dagger$. Finding an optimal solution that maximizes the minimum component of $x$, i.e., the minimum value of any variable in the leximin solution, can be easily found by solving the LP \ref{lp_iterative_1} ($\opt$ is the optimal value of LP \ref{lp_general_1}). 

\begin{maxi}
		{} {\alpha}
			{\label{lp_iterative_1}}
		{}
        \addConstraint{c^T x}{= \opt}{}
        \addConstraint{Ax}{= b}{}
        \addConstraint{x_i}{\geq \alpha, }{\quad \forall i\in \{1,2,\ldots,n\}}
\end{maxi} 

The important point to note here is that, while this gives us a solution with the maximin value, it doesn't give us the exact set of variable(/s) that attain this value in the leximin solution. One way to solve this is by finding a strict complementarity solution - a solution that assigns maximin value to a subset of variables and a higher value to every other variable. Freund et al.(\cite{Freund-StrictComplementarity}) achieve this by solving a different LP. 

Below, we describe a more interesting technique - a generalization from Nace and Orlin(\cite{leximin_load_LP}) - using dual LP and strict complementarity slackness conditions. The idea is that, if $z_i$ represents the dual variable corresponding to the constraint $x_i\geq \alpha$, the dual LP has constraints $z_i\geq 0,\forall x_i\in X$ and $\sum_{x_i\in X } z_i = 1$. Each $z_i$ being zero or positive informs us whether $x_i = \alpha$ in the leximin solution. Since the condition forces at least one positive $z_i$, we can determine a subset of variables having the maximin value in the leximin solution. Thus, we can solve for the leximin solution using an iterative process - set the known variables to their appropriate values and solve the LP that maximizes the next minimum value, using its dual LP. The algorithm is deatiled below.

\noindent Consider an LP formulated as: 

\begin{mini}
		{} {c^T x }
			{\label{lp_general}}
		{}
        \addConstraint{Ax}{= b}{}
        \addConstraint{x}{\geq 0}{}
\end{mini}

\begin{claim}
\label{cl:unique}
The leximin optimum solution of an LP is unique.    
\end{claim}

\begin{proof}
Let $x^1,x^2$ be two distinct optimal solutions to the LP that are both leximin. This implies that the multiset of values assigned to variables in $x^1,x^2$ are identical. By convexity of the set of optimal solutions of an LP, $x=(x^1+x^2)/2$ is also an optimal solution. Suppose $i$ is the smallest index such that the variable with the $i^{\rm th}$ smallest value in $x^1$, say $x_{j_1}$, is different from the variable with the $i^{\rm th}$ smallest value in $x^2$, say $x_{j_2}$. Note $x^1_{j_1}< x^1_{j_2}$, $x^2_{j_2}< x^2_{j_1}$ and $x^1_{j_1}= x^2_{j_2}$. Hence $x_{j_1}=(x^1_{j_1}+x^2_{j_1})/2 > x^1_{j_1}$ and $x_{j_2}=(x^1_{j_2}+x^2_{j_2})/2 > x^2_{j_2}$. Besides the $(i-1)$ smallest variables all variables in $x$ have value strictly larger than $x^1_{j_1}=x^2_{j_2}$. Hence $x$ is lexicographically larger than both $x^1$ and $x^2$ yielding a contradiction.   
\end{proof}

We will find the leximin optimum solution to LP (\ref{lp_general}) by solving a sequence of LPs. Our first step is to solve LP (\ref{lp_general}) to compute its optimal value, say \opt. Next, we modify LP (\ref{lp_general}) by adding the constraint $c^Tx =\opt$ to ensure that any feasible solution to the modified LP is an optimum solution to LP (\ref{lp_general}). Let $x=(x_1,x_2,\ldots,x_n)$ be the variables in LP (\ref{lp_general}). To maximize the minimum value assigned to a variable we include constraints $x_i\ge \alpha, i\in[n]$, and our objective now is to maximize $\alpha$.

Let $\alpha^*$ be the value of the optimum solution of the modified LP. Hence there is no optimum solution to LP (\ref{lp_general}) where all variables take values that are strictly larger than $\alpha^*$. We would like to fix those variables which take value $\alpha^*$ in every optimum solution. Let $V_F$ denote the set of indices of $x$ whose values have been fixed so far. Initially $V_F = \emptyset$. The mapping $m: x \to \mathbb{R}_+$ assigns the fixed values to $x_i$s where $i \in V_F$ and is zero for all others. LP (\ref{lp_iterative}) is the modified LP.
\begin{maxi}
		{} {\alpha}
			{\label{lp_iterative}}
		{}
        \addConstraint{c^T x}{= \opt}{}
        \addConstraint{Ax}{= b}{}
        \addConstraint{x_i}{\geq \alpha, }{\quad \forall i \notin V_F}
        \addConstraint{x_i}{= m(x_i), }{\quad\forall i \in V_F}
\end{maxi} 

and the dual program of the modified LP is the following
\begin{mini}
		{} {b^Ty+\beta\opt-\sum_{i\in F} \alpha_i z_i}
			{\label{lp_iterative_dual}}
		{}
        \addConstraint{A^Ty-z}{\ge \beta c}{}
        \addConstraint{\sum_{i\notin V_F} z_i}{= 1}{}
        \addConstraint{z_i}{\geq 0, }{ \quad\forall i \notin V_F}
\end{mini}

we can solve the dual LP (\ref{lp_iterative_dual}) in polynomial time to obtain an optimal dual solution \( z^* \). By strong duality, the optimum primal objective, denoted by \( \alpha^* \), equates to the optimum dual objective. Let \( x^* \) denote any optimum solution of the primal LP (\ref{lp_iterative}). By the complementary slackness condition, for all \( i \notin V_F \) such that \( z_{i}^* > 0 \), it holds that \( x^*_i = \alpha^* \). This implies that for the set of indices $\{i\notin V_F\}$ with positive duals, denoted as \( V_{\alpha^*} \), the corresponding primal variables assigned the fixed value \( \alpha^* \) in every feasible max-min solution; hence, these variables form a subset of the minimum set of variables that must be fixed at \( \alpha^* \). Given that the dual optimum solutions satisfy the feasibility constraint \( \sum_{i \notin V_F} z^*_i = 1 \), there exists at least one index $i \notin V_F$ with a positive dual, implying that \( V_{\alpha^*} \) is non-empty.

Next, for the primal variables with indices in \( V_{\alpha^*} \), their values are fixed, and the set of fixed indices $V_F$ along with the mapping $m(\cdot)$ are updated accordingly, i.e. $V_F = V_F \cup V_{\alpha^*}$ and $\forall i\in V_{\alpha^*}: m(x_i) = \alpha^*$.
The procedure iterates until all indices are fixed.  Since in every iteration at least one new index is fixed, there are at most \( n \) iterations, ensuring that the problem is solvable in polynomial time.


\begin{lemma}
\label{lem:leximin_optimum}
    The procedure returns a leximin optimum solution to LP(\ref{lp_general}).
\end{lemma}

\begin{proof}
Let $\gamma_j$ be the value at which variables were fixed at step $j$ of our algorithm. Let $x^j$ be an optimum solution to LP (\ref{lp_iterative}) in step $j$ which assigns all variables fixed in step $j$ a value $\gamma_j$ and all other non-fixed variables a value at least $\gamma_j$. $x^j$ is a feasible solution to LP(\ref{lp_iterative}) in step $(j+1)$ and hence the optimum solution at this step has a value at least $\gamma_j$. This implies $\gamma_{j+1}\ge \gamma_j$ at any step $j$ of the algorithm, excluding the last step.

We renumber variables in the order they got fixed through our algorithm. This implies $x_i$ got fixed either in the same or in an earlier step than $x_{i+1}$. Let $\alpha_i = m(x_i)$ for all $i\in [n]$. It follows from the above discussion that $\alpha_1\le\alpha_2\le \cdots\le \alpha_n$. 

We prove the lemma by contradiction and assume that the optimal solution found by our algorithm, say $a$, is not leximin. Let $b$ be the leximin optimal solution and $\beta_1\le\beta_2\le\cdots\le\beta_n$ be the values of variables in solution $b$ in an ascending order. Let $i$ be the first index where $\beta_i > \alpha_i$.

We first observe that variables corresponding to values $\beta_1,\beta_2,\ldots,\beta_{i-1}$ in $b$ and $\alpha_1,\alpha_2,\ldots,\alpha_{i-1}$ in $a$ are identical. If that was not the case, one could use an argument similar to Claim~\ref{cl:unique} to show that the solution $(a+b)/2$ is an optimal solution lexicographically larger than $b$ leading to a contradiction. 

Suppose variable $x_i$ was assigned value $\alpha_i$ at step $j$ of our algorithm and let $x^j$ be an optimal solution to LP(\ref{lp_iterative}) in this step. We first consider the case when variables $x_1,\ldots,x_{i-1}$ were fixed at steps before $j$. The solution $b$ is a feasible solution to LP(\ref{lp_iterative}) at step $j$ and has value $\beta_i > \alpha_i$. This contradicts our assumption that the optimum solution to LP(\ref{lp_iterative}) at step $j$ has value $\alpha_i$. We next consider the case when $x_{i-1},x_i$ were fixed together at step $j$. Hence in the optimum solution to LP(\ref{lp_iterative_dual}) at step $j$, $z^*_{i-1}>0$ and $z^*_i>0$ which implies that in every optimum solution to LP(\ref{lp_iterative}) at this step the variables $x_{i-1},x_i$ have identical values. However, $b$ is also an optimum solution to LP(\ref{lp_iterative}) at this step which violates this property since $\beta_i > \alpha_i=\alpha_{i-1}=\beta_{i-1}$ yielding a contradiction.
\end{proof}

The above method gives us an efficient algorithm to find the leximin LP solution and we can use this method to solve for the leximin Owen set imputation for the Max-flow game and the $b-$matching games. Below, we describe the novel structure of the Owen set in these games(and the MST game) and how we can use them to get efficeint combinatorial algorithms for the same.

\section{Max-flow game}
\label{sec:flow_game}

We would like to obtain the leximin(/leximax) fair imputations in the core of the Max-flow game. But, a corollary of the following theorem is that finding these imputations is NP-hard.

\begin{restatable}{theorem}{flowleximinhardness}
\label{thm:flow-leximin-hardness}
${}^\dagger$
    Finding a core imputation that maximizes the minimum profit-share of any edge in a max-flow game is NP-hard. Similarly, it is NP-hard to find a core imputation that minimizes the maximum profit-share of any edge.
\end{restatable}

And so, we will focus on the Owen set of the max-flow game. We will first look at a linear programming(LP) formulation of the max-flow problem.

\subsection{LP formulation of the max-flow problem}

Consider a max-flow instance on a directed graph $G = (V, E)$, a source vertex $s\in V$, a sink vertex $t\in V$ and edge capacities $c: E \to \mathbb{R}_+$. The maximum $s$-$t$ flow problem can be formulated as a circulation by introducing an edge of infinite capacity from $ t $ to $ s $ and ensuring flow conservation at all vertices including $s,t$. The objective of finding a flow of maximum value can be restated as finding a circulation which maximizes the flow in the edge $(t,s)$. A linear program(LP) for the maximum flow problem is formulated by associating a variable $f_{ij}$ with edge $(i,j)\in E$ whose value equals the flow through this edge.  

\begin{maxi}
		{} {f_{ts}}
			{\label{eq.flow-primal}}
		{}
		\addConstraint{\sum_{j:(j,i)\in E}{f_{ji}} ~-
            \sum_{j:(i,j)\in E}{f_{ij}}}{\leq 0 \quad}{\forall i \in V}
		\addConstraint{0\leq f_{ij}}{\leq c_{ij} }{\forall (i,j) \in E}
\end{maxi} 

The first set of inequalities, taken together for all vertices ensures the conservation of flow while the second set of inequalities imposes the capacity constraints.

The dual of the maxflow LP is obtained by associating a potential $\pi_i$ with vertex $i\in V$ and a length $\delta_{ij}$ with edge $(i,j)\in E$. 

\begin{mini}
    {} {\sum_{(i, j) \in E}  {c_{ij} \delta_{ij}}}  
        {\label{eq.flow-dual}}
    {}
    \addConstraint{\delta_{ij} - \pi_i + \pi_j}{ \geq 0 \quad }{\forall (i, j) \in E}
    \addConstraint{\pi_s - \pi_t}{\geq 1}{}
    \addConstraint{\delta_{ij}}{\geq 0}{\forall (i, j) \in E}
    \addConstraint{\pi_i}{\geq 0}{\forall i \in V}
\end{mini}

The first set of constraints requires that the potential drop across any edge be at most the length of the edge. The second constraint requires that the potential difference between $s,t$ is at least 1 and these two constraints together imply that the length of any path from $s$ to $t$ under the length function $\delta$ is at least 1. Thus the dual LP can be viewed as an assignment of lengths to edges, $\delta: E \rightarrow \mathbb{R}_+$, that minimizes $\sum_{(i,j)\in E} \delta_{ij}c_{ij}$ subject to the condition that the length of any path from $s$ to $t$ is at least 1. If $\delta$ is 0/1, then every path from $s$ to $t$ has an edge of length 1 and hence edges of length 1 form a cut separating $s$ and $t$; this set of edges is an $s$-$t$ cut. For this reason, the assignment $\delta$ is also called a fractional $s$-$t$ cut and its capacity is $\sum_{(i,j)\in E} \delta_{ij}c_{ij}$. By the strong duality theorem, the capacity of the minimum fractional $s$-$t$ cut equals the maximum $s$-$t$ flow.

\subsection{Owen set for the max-flow game}

The central idea is to identify the subset of the core that corresponds to optimal solutions of the dual linear program (LP), referred to as the \textbf{Owen set}. 

Let $(\delta, \pi)$ represent an optimal dual solution. Define the \textit{profit} of an edge $(i, j)$ as: 
$$ p_{ij} := c_{ij} \cdot \delta_{ij}$$
This defines an imputation because, by the LP duality theorem, the worth of the game which is the maximum flow, equals the objective function value of an optimal dual solution. Specifically, $\text{worth}(E) = \sum_{(i, j) \in E} c_{ij} \cdot \delta_{ij}$. An imputation is classified as an \textbf{Owen set imputation} if and only if there exists an optimal dual solution $(\delta, \pi)$ such that: $p_{ij} := c_{ij} \cdot \delta_{ij}.$

\begin{theorem}
	\label{thm.flow}
	For the max-flow game, the Owen set is a subset of the core. 
\end{theorem}

\begin{proof}

For any sub-coalition $S \subseteq V$, define $ \text{profit}(S) := \sum_{{ij} \in S} {p({i, j})} .$ To prove that $p$ is in the core of the max-flow game, we need to show that for any sub-coalition $S \subseteq V$, $\text{worth}(S) \leq \text{profit}(S)$. 
	 
Consider a sub-coalition $S \subseteq V$. Let $f$ be a max-flow in the graph $G(S)$; its value is $\text{worth}(S)$. Using standard methods, $f$ can be decomposed into at most $|S|$ flow paths. Let $\mathcal{P}$ denote the set of all such paths and for path $p \in \mathcal{P}$, let $f_p$ denote the flow sent on $p$. 

We will use the following two facts. First, as stated above, the sum of distance labels on $p$, $\sum_{{ij} \in p} {\delta_{ij}} \geq 1$. Second, by complementary slackness, if $\delta_{ij} > 0$, edge ${i, j}$ is fully saturated, and therefore for such an edge, $\sum_{p \ni {i j}} {f_p} \ = \ c_{ij}$. Now,
$$ \text{worth}(S) \ = \ \sum_{p \in \mathcal{P}} {f_p} \ \leq \ \sum_{p \in \mathcal{P}} {f_p \cdot \left(\sum_{(i, j) \in p} {\delta_{ij}}\right)} \ = \ \sum_{(i, j) \in S} {\delta_{ij} \cdot \left(\sum_{p \ni (i, j)} {f_p} \right)} \ = \ \sum_{(i, j) \in S} {\delta_{ij} \cdot c_{ij}} \ = \text{profit}(S)$$ 
where the inequality follows from the first fact and the second-last equality follows from the second. 

\begin{corollary}
\label{cor.flow}
	The core of the max-flow game is always non-empty. 
\end{corollary}

Example \ref{ex.flow-no-core} gives a max-flow game and a core imputation not in its Owen set. Therefore, Owen set does not fully characterize the core.
\end{proof}


\begin{figure}[ht]
\begin{center}
\includegraphics[width=4in]{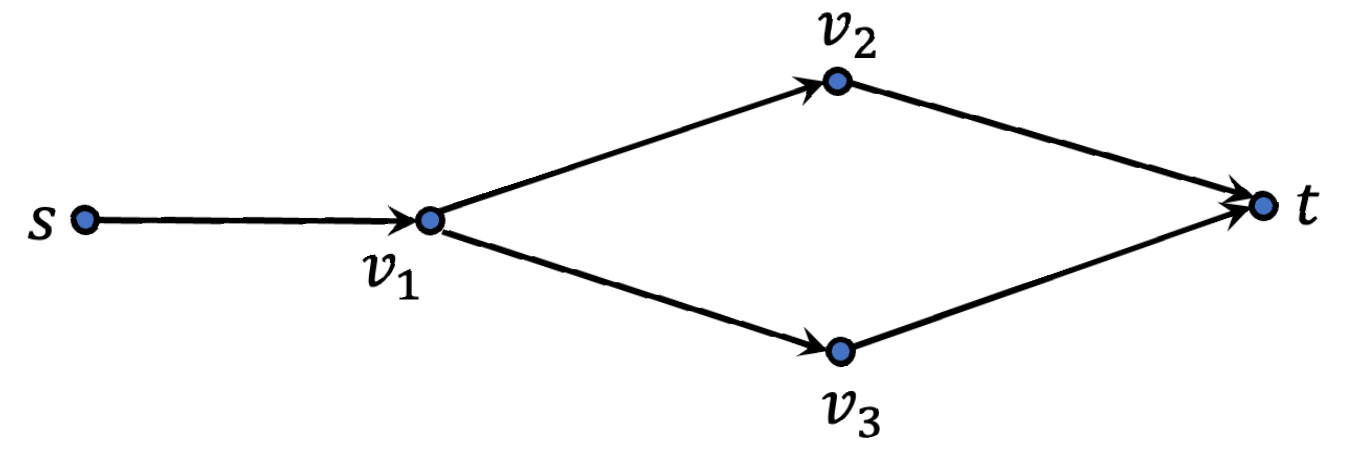}
\caption{The graph for Example \ref{ex.flow-no-core}.}
\label{fig.flow}
\end{center}
\end{figure}


\begin{example}
	\label{ex.flow-no-core}
Consider the graph given in Figure \ref{fig.flow}. Assume that the capacities of edges $(v_1, v_2)$ and $(v_1, v_3)$ are 1 each. An easy way of showing that there is a core imputation which does not correspond to an optimal dual solution is the following: Assume that the rest of the edges have very high capacities, e.g., 10 each. Assign a profit of 2 to edge $(s, v_1)$ and zero to the rest; observe that the worth of the game is 2. Clearly, the dual will not assign edge $(s, v_1)$ a positive distance label, since a max-flow does not saturate it. 

\end{example}

\begin{definition}
    An edge $e\in E$ is {\em essential} if it is saturated in every $s$-$t$ max-flow and it is called {\em inessential} if there is an $s$-$t$ max-flow that does not saturate $e$.
\end{definition}

\begin{remark}
\label{flow:negotiating_power}
If the capacity of an essential edge is dropped by (a small)$\epsilon$, the worth of the game also drops by $\epsilon$.
However, dropping the capacity of a non-essential edge by $\epsilon$ will not lead to any change in the worth of the game.
Therefore, only essential edges have negotiating power and, by complementarity slackness conditions, profits assigned by a Owen set imputations are consistent with this. 

\end{remark}

\begin{theorem}
\label{thm:flow_efficient_dual}
    Deciding if an imputation of the max-flow game is in the Owen set can be done in polynomial time. 
\end{theorem}

\begin{proof}
    Trivially, we can modify LP~\ref{eq.flow-dual} to check for the existence of a feasible solution $(\pi, \delta)$, where the distance labels for each edge are defined based on their profit as $\delta_{ij} = \frac{p_{ij}}{c_{ij}}$. Algorithm~\ref{alg:max_flow_Owen_set_check} provides a combinatorial procedure to determine whether a given imputation belongs to the Owen set. The algorithm first constructs a solution by enforcing strict complementarity conditions on a set of essential edges, then verifies whether the resulting solution satisfies all dual LP constraints. The algorithm builds upon ideas from the section~\ref{subsec:flow-leximin} and the proof of the algorithm is provided at the end of the section.

\end{proof}
\begin{figure}[h]
	\begin{wbox}
		\textbf{Deciding if an imputation is in the Owen set of a max-flow game}: \\
		\textbf{Input:} $G=(V,E)$, $s,t\in V$, $c : E \to \mathbb{R}_+$ and an imputation $p= \{p_{ij}\}_{(i,j)\in E}$. 
		\begin{enumerate}[noitemsep,nosep]
            \item If any non-essential edge got paid, return ``No''.
		  \item \textbf{Initialization:}
		  \begin{enumerate}[noitemsep,nosep]
                \item Set $\delta_{ij} = \frac{p_{ij}}{c_{ij}},\forall (i,j)\in E$. 
				\item $G'=(V',E')\leftarrow $ Picard-Queyranne structure of $G$.
                \item $\pi(t') \leftarrow 0$; $\textsc{Fixed} \leftarrow \{t'\}$; $\textsc{Free} \leftarrow V'\setminus \textsc{Fixed}$.
		  \end{enumerate}  
            \item \textbf{While} \textsc{Fixed} $\neq V'$ \textbf{do:}
                \begin{enumerate}[noitemsep,nosep]
                    \item Find an essential edge $e=(i,j)$ going from \textsc{Fixed} to \textsc{Free}.
                    \item Set $\pi_j = \pi_i + \delta_{ij}$.
                    \item Update \textsc{Fixed} and \textsc{Free}.
            \end{enumerate}
            \item Set potentials on $V$ to the potentials on corresponding SCC nodes in $V'$.
            \item \textbf{Output:} ``Yes'' if $(\pi,\delta)$ is a feasible dual solution to LP~\ref{eq.flow-dual}, otherwise ``No''.

		\end{enumerate}
	\end{wbox}
        \caption{Algorithm to decide if an imputation is in the Owen set of a max-flow game}
	\label{alg:max_flow_Owen_set_check} 
\end{figure}

\subsection{Computing the leximin Owen set imputation}
\label{subsec:flow-leximin}

\subsubsection{High level ideas}

Let $G=(V,E),s,t\in V, c: E\rightarrow \mathbb{R}_+$ be an instance of the max-flow game. Computing the leximin Owen set imputation involves three key ideas:
\begin{enumerate}
    \item \textbf{Dual LP and Potentials:} Using dual LP~\ref{eq.flow-dual}, we work with vertex potentials instead of edge distance labels. Optimal potentials corresponding to the leximin imputation are computed iteratively, building distance labels and profits from these potentials.
    \item \textbf{Iterative Optimization:} In each iteration, vertices and their optimal potentials that maximize the next minimum profit share are obtained. The key idea is that the vertices chosen in each iteration will correspond to certain longest paths in a related graph.
    \item \textbf{Efficient Longest Path Computation:} Computing longest paths, like the Hamiltonian path, might take exponential time. To avoid this, we leverage the ``Picard-Queyranne structure'' of $G$. This is a Directed Acyclic Graph(DAG) which retains all essential edges and the max-flow of $G$. This enables us to efficiently compute the required paths.
\end{enumerate}

\subsubsection{Picard-Queyranne Structure}

Our algorithm first starts by creating a Picard Queyranne structure of the graph, defined as follows. 

\begin{definition}
    Let $G=(V,E), s,t\in V$ be a max-flow instance. Let $f$ be the maximum $s-t$ flow in $G$ and $G|f$ be the residual graph of $G$ under $f$\footnote{Recall that the \textit{residual graph} $G|f$ is obtained from $G$ by reducing the capacity of every edge $(i,j)\in E$ from $c_{ij}$ to $c_{ij}-f_{ij}$ and introducing edge $(j,i)$ with capacity $f_{ij}$ for every edge $(i,j)\in E$ with $f_{ij}>0$. Edges with zero capacity are then deleted.}. Construct $G'=(V',E')$ from $G|f$ by shrinking each maximal strongly connected component in $G|f$ into a vertex and removing any self-loops. The resulting graph $G'=(V',E')$ is a \textit{Picard-Queyranne structure} of $G$.
\end{definition}

The resulting graph $G'$ will be an acyclic multigraph. Each vertex $i\in V'$ is a strongly connected component in $G|f$ and let $\phi(i)\subseteq V$ be the vertices in this strongly connected component. There is no path from $s$ to $t$ in $G|f$; hence, $s,t$ are in distinct strongly connected components of $G|f$. Therefore $s'\neq t'$ where $s'=\phi^{-1}(s)$ and $t'=\phi^{-1}(t)$.

\begin{remark}
    $G'$ is exactly the graph described by Picard and Queyranne(\cite{Picard_Queyranne_Structure}), and so, we call it the Picard-Queyranne structure. They prove that every $s'-t'$ cut of the graph $G'$ corresponds to a min $s-t$ cut of the graph $G$. 
\end{remark}

Define $F'$ and $Z'$ as the sets of edges corresponding to the edges that carry the full flow($f_{ij}=c_{ij}$) and the zero flow($f_{ij}=0$) respectively in $G$ under $f$. Formally, $$F'=\set{(j,i)\in E', f_{ij}=c_{ij}}$$ $$Z'=\set{(i,j)\in E', f_{ij}=0}$$
The construction of the graph ensures only these edges are preserved in $G'$ - all other edges will removed when the strongly connected components are shrinked, i.e., 

\begin{restatable}{lemma}{flowsplit}
\label{lem:flow_split_E'}
    $E'=F'\cup Z'$.
\end{restatable}

\begin{proof}

    Consider an edge $e=(i,j)\in E$. We have three cases: 
    \begin{itemize}
        \item If $0 <f_e<c_e$ then both $(i,j)$ and $(j,i)$ are in $G|f$ and hence $i$ and $j$ are in the same strongly connected component of $G|f$. Hence \textbf{neither} $(i,j)$, nor $(j,i)$ are in $E'$. 
        \item If $f_e=c_e$ then edge $(j,i)$ is in $G|f$ and if it is not on a cycle, it belongs to $E'$. Note that only the essential edges are not part of any cycle, and hence preserved in $E'$.
        \item If $f_e=0$ then edge $(i,j)$ is in $G|f$ and if it is not on a cycle it belongs to $E'$. Note that only the inessential edges that carry zero flow in every max flow are preserved in $E'$. 
    \end{itemize}

\end{proof}

Note that $F'$ is defined as the set of edges that carry full flow in $G$ i.e., they are the set of essential edges in $G$, although pointed in reverse direction.

\begin{corollary}
    All essential edges of $G$ are preserved in $G'$.
\end{corollary}

\subsubsection{Algorithm}

As mentioned above, the algorithm works in iterations and in each iteration, it fixes the potentials of vertices along some path. Let us first define a few terms in relation to this.

\begin{definition}
    \textsc{Fixed} is the set of all vertices with assigned potentials. \textsc{Free} is the set of remaining vertices in $V'$.
\end{definition}

Initially, the source and sink vertices of $G'$, $s'$ and $t'$ are given a potential of 1 and 0 respectively and are in \textsc{Fixed}. 

\begin{definition}
    A \textit{Free Path} is sequence of vertices $a=v_0,v_1,v_2, \ldots, v_{n-1},v_n=b$ with $(v_i,v_{i+1})\in E', \forall i\in \{0,1,2,\ldots, n-1\}$ such that $a,b\in \textsc{Fixed}$ and $v_1,v_2, \ldots, v_{n-1} \in \textsc{Free}$.
\end{definition}

We first give new length function $l$ to edges in $G'$ such that $l_e =  \frac{1}{c_e}$ if $e\in F'$, and $l_e = 0$ otherwise. Our algorithm, see Figure~\ref{alg:max_flow_leximin}, will run a sequence of iterations and, in each iteration, it finds the longest free path(under lengths $l_e$) between every pair of fixed vertices. It then assigns potentials to vertices on this path such that every essential edge gets equal profit share, under the rule $\delta_{ij}=\max\set{\pi_i-\pi_j,0}$, and the potential difference across every non-essential edge is zero - giving them zero profit.

\textbf{Intuition:} To understand why this algorithm gives us the leximin imputation, consider a free path $P_{uv}$ in $\mathcal{P}_{uv}$ - the set of all free paths between $u$ and $v$ - after $i$ iterations of the algorithm. If the free vertices on this path were the only free vertices left, we would maximize the next minimum profit share by choosing potentials in a way that gives equal profits, say $\alpha_{uv}$ to all the essential edges and zero to the rest. But, the potential difference across this path is fixed at $\pi_v-\pi_u$, which when split across edges gives -$$\pi_v - \pi_u  =  \sum_{(i,j)\in P_{uv}\cap F'} (\pi_j - \pi_i) + \underbrace{\sum_{(i,j)\in P_{uv}\cap Z'} (\pi_j - \pi_i)}_{non-negative} \geq \sum_{(i,j)\in P_{uv}\cap F'} \frac{\alpha_{uv}}{c_{ij}} = \sum_{(i,j)\in P_{uv}} \{\alpha_{uv} \cdot l_{i,j}\}$$For the inequality, We use $\delta_{ij} \geq \set{\pi_i-\pi_j}$ for the essential edges and $\delta_{ij}\geq 0$ for the non-essential edges. So, to maximize the profit, potential difference across the non-essential edges should be zero and potential difference across the essential edges should be split proportional to $1/c_e$. Our length function, $l_e$, was chosen to combine these comstraints, and so, we get $ \alpha_{uv} = \frac{\pi_v - \pi_u}{{\sum_{e\in P_{uv}} l_e}} $, to be the maximum profit to the edges this path can provide. And so, the longest free path in $G'$ is the free path with worst potential constraints that ensures minimum profit, $\alpha_{uv}$, to the edges.

Since $G'$ is acyclic, such a path can be computed efficiently using topological sort. Let $a,b\in \textsc{Fixed}$ after $(i-1)$ iterations and $P \in \mathcal{P}_{ab}$ a path for which the minimum profit is achieved. For vertices $j,k\in P$ let $P[j,k]$ be the subpath of $P$ from $j$ to $k$. For all vertices $j$ on $P$ we assign $\pi_j=\pi_a + \alpha_i\sum_{e\in P[a,j]} w_e$. Iteration $i$ ends with adding vertices of $P$ to $\textsc{Fixed}$ and the algorithm ends when all vertices are in $\textsc{Fixed}$.

\begin{figure}[ht]
	\begin{wbox}
		\textbf{Finding leximin Owen set imputation in max-flow game}: \\
		\textbf{Input:} $G=(V,E)$, $s,t\in V$, $c : E \to \mathbb{R}_+$.  
		\begin{enumerate}[noitemsep,nosep]
		  \item \textbf{Initialization:}
		  \begin{enumerate}[noitemsep,nosep]
				\item $G'=(V',E')\leftarrow $ Picard-Queyranne structure of $G$.
                \item $\pi(s') \leftarrow 1; \pi(t') \leftarrow 0$.
                \item $\textsc{Fixed} \leftarrow \{s',t'\}$; $\textsc{Free} \leftarrow V'\setminus \textsc{Fixed}$.
                \item $\forall e \in E', l_e = \frac{1}{c_e}$ if $e$ is essential and $0$ otherwise.
		  \end{enumerate}  
            \item \textbf{While} \textsc{Fixed} $\neq V'$ \textbf{do:}
                \begin{enumerate}[noitemsep,nosep]
                    \item \textbf{For} every pair $u,v\in \textsc{Fixed}$ \textbf{do:} 
                    \begin{enumerate}
                        \item Compute longest free path, $P_{uv}$, between them.
                        \item Compute profits of essential edges, $\alpha_{uv} = \frac{\pi_v-\pi_u}{{\sum_{e\in P_{uv}} l_e}}$, on $P_{uv}$.
                    \end{enumerate}
                    \item Set potentials of vertices on free path $P_{ab}$ with the minimum profit $\alpha_{ab}$.
                    \item Update \textsc{Fixed} and \textsc{Free}.
            \end{enumerate}
            \item \textbf{Output:} The imputation $p$, where $\forall (i,j)\in E, p_{ij} = c_{ij} \cdot \max\set{\pi_i-\pi_j,0}$.            
		\end{enumerate}
	\end{wbox}
        \caption{Algorithm to find the leximin Owen set imputation in a max-flow game.}
	\label{alg:max_flow_leximin} 
\end{figure}

The proofs of the below results will require some notation. Let $V_i'\in V'$ be the set \textsc{Fixed} after $i$ iterations, i.e., the set of vertices with fixed potentials after $i$ iterations. For $a,b\in V'_{i-1}$, $\pi_a <\pi_b$, let $\mathcal{P}_{ab}$ be the set of free paths from $a$ to $b$ in iteration $i$. Then the profit assigned to edges during iteration $i$ can also be written as 
$$ \alpha_i=\min_{a,b\in V'_{i-1}} \min_{P\in\mathcal{P}_{ab}} 
\frac{\pi_b-\pi_a}{\sum_{e\in P} l_e}.$$

It is straightforward to see that the profits assigned, $\alpha_i$'s, increase with each iteration.

First, we show that the potentials conform to the dual by proving a useful property of $\pi$. 
\begin{restatable}{claim}{clorient}\label{cl:orient}
    $\forall (i,j)\in E'$, $\pi_i\le\pi_j$.
\end{restatable}

\begin{proof}
\begin{figure}[ht]
    \centering
    \begin{tikzpicture}

    \coordinate (a) at (0, 0);
    \coordinate (j) at (4, 0);
    \coordinate (b) at (8, 0);
    \coordinate (i) at (3, -1.5);
    
    \filldraw[black] (a) circle (2pt) node[above] {$a$($\pi_a$)};
    \filldraw[black] (b) circle (2pt) node[above] {$b$($\pi_b$)};
    \filldraw[black] (j) circle (2pt) node[above] {$j$($\pi_j$)};
    \filldraw[black] (i) circle (2pt) node[below right] {$i$($\pi_i$)};
    
    \draw[->] (a) to[bend left=20, dotted] node[midway, above] {$P[a,j]$} (j);
    \draw[->] (j) to[bend right=20, dotted] node[midway, above] {$P[j,b]$} (b);
    \draw[->] (i) to node[midway, right] {$\gamma$} (j);
    
    \end{tikzpicture}
    \caption{Illustration of claim~\ref{cl:orient} }
    \label{fig:dual_are_decreasing}
    \end{figure}
    
Consider an edge $(i,j)\in E'$. If vertices $i,j$ were assigned potentials in the same iteration of the algorithm then vertices $i,j$ lie on the path $P$ which achieves the minimum in this iteration and the claim is trivially true. 
We first consider the case when vertex $i$ is assigned a potential before vertex $j$. 
Suppose $j$ was assigned a potential in iteration $k$ and $a,b\in V'_{k-1}$ were the vertices and $P$ the path from $a$ to $b$ for which the minimum ratio $\alpha_k$ was achieved. Vertex $j$ lies on $P$ and let $P'$ be the path from $i$ to $b$ consisting of the edge $(i,j)$ followed by $P[j,b]$. Since $P$ is a free path in iteration $k$, $P'$ is also free in this iteration and hence $P'\in\mathcal{P}_{ib}$.

In iteration $k$ the algorithm chose $P$ over $P'$ and hence
$$\pi_b-\pi_i \geq \alpha_k\left( \gamma +\sum_{e\in P[j,b]\cap F'} c_e^{-1}\right)$$
where $\gamma=1/c_{ij}$ if $(i,j)\in F'$ and is 0 otherwise. This implies 
$$ \pi_i \le \pi_b-\alpha_k \left(\sum_{e\in P[j,b]\cap F'} c_e^{-1}\right) - \gamma\alpha_k = \pi_j - \gamma\alpha_k$$
Since $\alpha_k,\gamma\ge 0$, $\pi_i\le \pi_j$, proving the claim. 
\end{proof}

The potential assigned to vertices in $V'$ is extended to vertices in $V$ by assigning all vertices in $\phi(i)$ the potential $\pi_i$. 
The potentials on vertices in $V$ are in turn used to assign lengths to edges in $E$; for edge $(i,j)\in E$, let $\delta_{ij}=\max\set{0,\pi_i-\pi_j}$. The following claim~\ref{cl:flow_opt} proves that the profit share is an Owen set imputation and then lemma~\ref{lem:flow_leximin} shows that it is also the leximin imputation. Theorem~\ref{thm:flow_DCC_theorem} shows that the algorithm~\ref{alg:max_flow_leximin} is efficient.  

\begin{restatable}{claim}{clflowopt}
\label{cl:flow_opt}
    $(\pi,\delta)$ is an optimum solution to LP \ref{eq.flow-dual}.
\end{restatable}

\begin{proof}
From our definition of $\delta$ and the fact that $\pi_s-\pi_t=1$ it follows that assignments $\pi:V\to\mathbb{R}_+$ and $\delta:E\to\mathbb{R}_+$ form a feasible solution to LP \ref{eq.flow-dual}. Hence we only need to argue that this solution is an optimum solution and to do so we show that the value of this solution, $\sum_{(i,j)\in E} \delta_{ij}c_{ij}$, equals the maximum $s$-$t$ flow.

For $x\in(0,1]$ let $S_x=\set{x\in V, \pi_v\ge x}$ and $T_x=\set{x\in V, \pi_v <x}$. Since $\forall x\in(0,1]$, $s\in S_x$ and $t\in T_x$, $(S_x,T_x)$ is an $s$-$t$ cut. Let $C(x)$ be the capacity of the cut $(S_x, T_x)$; recall that the capacity of an $s$-$t$ cut is defined as the total capacity of edges from the $s$-side to the $t$-side of the cut. 

Consider the $s$-$t$ maxflow $f$, edge $(i,j)\in E$ and $x\in(0,1]$. 
\begin{enumerate}
\item If $i\in S_x, j\in T_x$ then $\pi_i\ge x > \pi_j$ and hence by Claim~\ref{cl:orient} $(j,i)\in E'$. $(i,j)\in E$ and $(j,i)\in E'$ implies $(j,i)\in F'$ and hence $f_{ij}=c_{ij}$. Thus all edges from $S_x$ to $T_x$ are saturated by $f$. 
\item If $i\in T_x, j\in S_x$ then $\pi_j\ge x > \pi_i$ and hence by Claim~\ref{cl:orient} $(i,j)\in E'$. $(i,j)\in E$ and $(i,j)\in E'$ implies $(i,j)\in Z'$ and hence $f_{ij}=0$. Thus all edges from $T_x$ to $S_x$ carry zero flow in $f$.
\end{enumerate}
Thus the net flow from $S_x$ to $T_x$ which is the value of the flow $f$ and denoted by $|f|$ equals, $C(x)$, the capacity of the cut $(S_x,T_x)$. Since this is true for all $x\in(0,1]$ it follows that $\int_0^1 C(x) dx=|f|$. 

Recall that for $(i,j)\in E$, $\delta_{ij}=\pi_i-\pi_j$ if $\pi_i > \pi_j$ and is 0 if $\pi_i\le \pi_j$. If $\pi_i>\pi_j$ then edge $(i,j)$ contributes $(\pi_i-\pi_j)c_{ij}$ to $\int_0^1 C(x) dx$ while if $\pi_i \le \pi_j$ it contributes 0 to $\int_0^1 C(x) dx$. Therefore the contribution of $(i,j)\in E$ to $\int_0^1 C(x) dx$ equals $\delta_{ij}c_{ij}$ and hence $|f|=\int_0^1 C(x) dx = \sum_{(i,j)\in E} \delta_{ij}c_{ij}$ which implies that the solution $(\pi,\delta)$ is optimum.
\end{proof}

The following establishes the correctness of our algorithm. 
\begin{restatable}{lemma}{lemflowleximin}\label{lem:flow_leximin}
The imputation defined by the solution $(\pi,\delta)$ is the leximin Owen set imputation for the $s$-$t$ max-flow game.  
\end{restatable}

\begin{proof}
For contradiction assume $(\pi^*,\delta^*)$ is a dual solution to LP \ref{eq.flow-dual}, different from the solution $(\pi,\delta)$ such that the imputation, $p$, defined as $p(e)=\delta^*_ec_e$ is the leximin Owen set imputation. Since $\delta^*_{ij}\ge \pi^*_i-\pi^*_j$, and the dual objective is minimized when $\delta^*_{ij}$ is as small as possible, it must be the case that $\delta^*_{ij}=\max\set{0,\pi^*_i-\pi^*_j}$. Hence the dual solution is completely specified by $\pi^*$.

Recall that $\phi(i)\subset V$ are the vertices corresponding to the vertex $i\in V'$ and all vertices in $\phi(i)$ have identical potential in the assignment $\pi$. We first argue that $\forall i\in V'$, $\forall j,k\in\phi(i)$, $\pi^*_j=\pi^*_k$. Since $j,k\in \phi(i)$ these vertices lie on a cycle, say $C$, in the graph $G|f$. Since all edges on $C$ have non-zero residual capacity in $G|f$, we can increase flow on these edges by an $\epsilon >0$ thus introducing reverse edges (if they did not already exist) for every edge of $C$. Let $f'$ be this $s$-$t$ max-flow and consider the primal-dual pair of solutions $f',(\pi^*,\delta^*)$ and edge $(i,j)\in C$.
\begin{enumerate}
\item If $(i,j)\in E$ then $0 < f'_{ij}<c_{ij}$ and hence by complementary slackness $\delta^*_{ij}=0$ and $\delta^*_{ij}=\pi^*_i-\pi^*_j$ which implies $\pi^*_i=\pi^*_j$.
\item If $(j,i)\in E$ then $0 < f'_{ji}<c_{ji}$ and hence by complementary slackness $\delta^*_{ji}=0$ and $\delta^*_{ji}=\pi^*_j-\pi^*_i$ which implies $\pi^*_i=\pi^*_j$.
\end{enumerate}
Hence all vertices on $C$ and therefore all vertices in $\phi(i)$ have the same potential in the assignment $\pi^*$.

The above argument allows us to "project" the assignment $\pi^*:V\to\mathbb{R}_+$ to assign potentials to vertices in $V'$. This assignment, also denoted by $\pi^*$, is defined as: for $i\in V'$, let $\pi^*_i=\pi^*_j$ where $j$ is an arbitrary vertex in $\phi(i)$. 

We will now argue that the assignments $\pi^*$ and $\pi$ are identical. Note that $\pi^*_{s'}=1=\pi_{s'}$ and $\pi^*_{t'}=0=\pi_{t'}$. Let $k$ be the earliest iteration of our algorithm in which we assign a potential to a vertex $v$ that is different from $\pi^*_v$. Let $a,b\in V'_{k-1}$ be the end-vertices and $P\in\mathcal{P}_{ab}$ the path  which achieves the ratio $\alpha_k$ in iteration $k$. If $(i,j)\in P\cap Z'$ then $f_{ij}=0 < c_{ij}$ and by complementary slackness it follows that $\delta^*_{ij}=0$ and hence $\pi^*_i\le\pi^*_j$.

The potential drop between the endpoints of $P$ is $\pi_b-\pi_a = \pi^*_b-\pi^*_a$ and our algorithm distributes this across the edges of $P\setminus Z'$, in such a manner that all these edges get an equal profit-share. If $\pi^*_v < \pi_v$ then the potential drop across $P[a,v]$ for assignment $\pi^*$ is less than that for assignment $\pi$ and hence some edge in $P[a,v]$, will have a lower profit-share under $\pi^*$ than what it had under assignment $\pi$. Similarly, if $\pi^*_v >  \pi_v$ then some edge in $P[v,b]$, will have a lower profit-share under $\pi^*$ than what it had under assignment $\pi$. Let $e$ be this edge. Edges with a lower profit-share than $e$ got their profit-share in iterations 1 through $k-1$ of our algorithm and these were identical for both the solutions $\pi,\pi^*$. Further in subsequent iterations, our algorithm only assigns a larger profit-share to edges than $e$. Hence the imputation defined by the solution $\pi$ is lexicographically larger than the imputation defined by $\pi^*$, contradicting our assumption.

\end{proof}

\begin{restatable}{theorem}{flowDCC}
\label{thm:flow_DCC_theorem}
    There exist $O(mn^2)$ run time algorithms to find the leximin and leximax Owen set imputations of the max-flow game on a graph with $n$ nodes and $m$ edges.
\end{restatable}

\begin{proof}

Lemma~\ref{lem:flow_leximin} establishes the correctness of the algorithm~\ref{alg:max_flow_leximin}. To verify the time complexity, let $n$ denote the number of vertices, and $m$ the number of edges in a graph $G(V,E)$. The algorithm discussed involves, in each iteration, finding an optimal free path originating from a fixed vertex. Note that, the Picard-Queyranne structure is a Directed Acyclic Graph and so, the longest path from a source to all vertices can be found in linear time. 

To find the longest path from vertex $a$ to all other vertices, first set the longest distance to $a$ as 0 and every other vertex as $-\infty$. Now, create a topological ordering of the vertices and for every vertex $b$ in the topological ordering, update the distances to its neighbours by comparing their current distances and the distances using edge from $b$. This gives us the length of longest path from $a$ to every reachable vertex. Repeating this procedure from each vertex, we get the longest paths between every pair of vertices. 

Since the procedure to find single-source longest paths is linear, i.e., $O(m)$, an iteration to find the optimal free path takes $O(mn)$ time. As the algorithm guarantees the fixation of the potential of at least one vertex per iteration, it performs at most $O(n)$ iterations. Consequently, the total time complexity of the algorithm is $O(mn^2)$. 
Analogous analysis for the leximax algorithm completes the proof. 
    
\end{proof}

We can use the above results to also prove the correctness of Algorithm~\ref{alg:max_flow_Owen_set_check}.
    
\begin{restatable}{lemma}{flowdual}
\label{lem:flow_efficient_dual}
The imputation $p$ belongs to the Owen set of the max-flow game on $G$ if and only if Algorithm~\ref{alg:max_flow_Owen_set_check} generates a feasible dual solution.
\end{restatable}

\begin{proof}

    Firstly, the algorithm~\ref{alg:max_flow_Owen_set_check} directly returns ``No'' if any non-essential edge gets paid owing to Remark~\ref{flow:negotiating_power}. The algorithm then constructs distance labels of edges from profits as $\delta_{ij}=\frac{p_{ij}}{p_{ij}}$. It then constructs a Picard-Queyranne structure, $G'$, of the graph $G$. Note that this graph preserves all the essential edges and merges strongly connected components into single vertices. The final potentials on vertices in $G$ will be the potentials of corresponding strongly connected component nodes in $G'$. Potential of $t'$, the sink vertex, is initialized to 0 as all dual optimal solutions must satisfy that.
    
    The correctness of the algorithm follows from two simple facts -

    \begin{itemize}
        \item Essential edges always carry full flow, meaning $f_{ij}>0$ for these edges. Strict complementarity conditions ensure that the corresponding dual constraint on these essential edges must always remain tight. So, if we have an essential edge $e=(i,j)$ from \textsc{Fixed} to \textsc{Free}, the potential of the free vertex should be set according to $\delta_{ij} = \pi_i-\pi_i$. Note that, essential edges in $G'$ and $G$ have reverse directionality.  
        \item 
        If the max-flow is non-zero, vertices $s'$ and $t'$ must have essential edges flowing in to and out of them, respectively. And, since flow is conserved at every other node, all of them must have at least one essential edge into and out of them. So, the only way the algorithm doesn't assign potentials on some vertices is if all the essential edges incident on them contains a cycle, but $G'$ is a Directed Acyclic Graph.      
        
    \end{itemize}

    So, potentials to all the vertices will be assigned within $n$ iterations. These potentials will be extended to vertices in $G$ as explained above. Now, since we did not consider all the edges to build $(\pi,\delta)$, we check the feasibility of all dual constraints, including $\pi_s-\pi_t \geq 0$. If the solution is dual feasible, it gives the exact dual solution that produces this Owen set imputation. If it is not, then the edges that we have chosen when building the solution $(\pi,\delta)$ and the dual constraint that failed gives us a no-certificate.

\end{proof}

\section{MST Game and Min-cost Branching Game}
\label{sec:branching_game}
\newcommand{\abs}[1]{\left|#1\right|}

In this section, we will solve the general case of the min-cost branching game. The undirected version, the min-cost spanning tree game is a special case of this and can be solved easily by replacing each undirected edge with two directed edges. Note that, unlike the max-flow game, we deal with costs instead of profits in these games. 
Firstly, we would like to obtain the leximin(/leximax) fair imputations in the core of these games. But, a corollary of the following theorem is that finding these imputations is NP-hard.

\begin{restatable}{theorem}{mstleximinhardness}
\label{thm:mst-leximin-hardness}
${}^\dagger$
    Finding a core imputation that maximizes the minimum cost-share of any vertex in an MST game or a min-cost branching game is NP-hard. Similarly, it is NP-hard to find a core imputation that minimizes the maximum cost-share of any vertex.
\end{restatable}

\subsection{Owen set for the min-cost branching game}
\label{sec:MST_dual_consisitent_imputations}
Let $G=(V, E)$ be a directed graph, $c: E\rightarrow\mathbb{R}_+$ a cost function on the edge and $r\in V$ a root vertex. The value/worth of a set $S\subseteq \Vr$ , $v(S)$, corresponds to the minimum cost branching in the subgraph of $G$ restricted to $(S\cup \{r\})$ and defines the characteristic function of the game. 

The problem of finding a minimum-cost branching can be formulated as an integer program. Let $x(e)=1$ if $e\in E$ is included in the branching and is 0 otherwise. Let $\btd(S)\subseteq E$ be the edges with tail in $S$ and head in $\overline{S}$. For every $S\subseteq\Vr$, at least one edge from the set $\btd(S)$ should be contained in any branching and hence $\sum_{e\in\btd(S)} x(e)\ge 1, \forall S\subseteq\Vr$. The integer program which minimizes $\sum_{e\in E} x(e)c(e)$ subject to the above constraint yields a minimum cost branching in $G$. The LP-relaxation of the integer program replaces the integrality constraint $x(e)\in\set{0,1}$ with $x(e)\ge 0$. 

\begin{mini}
		{} {\sum_{e\in E} x(e)c(e)}
			{\label{eq.mst-org-primal}}
		{}
        \addConstraint{\sum_{e\in\btd(S)} x(e)}{\ge 1}{\quad\forall S\subseteq\Vr}
		\addConstraint{x_e}{\geq 0}{\quad\forall e \in E}
\end{mini} 
The dual of this LP has a non-negative variable $y(S)$ for $S\subseteq\Vr$. 
\begin{maxi}
		{} {\sum_{S\subseteq\Vr} y(S)}
			{\label{eq.mst-org-dual}}
		{}
        \addConstraint{\sum_{S:e\in\btd(S)} y(S)}{\le c(e)}{\quad\forall e\in E}
		\addConstraint{y(S)}{\geq 0}{\quad\forall S\subseteq\Vr}
\end{maxi}

Let $y$ be a feasible dual solution. The function $y':2^V\times V\rightarrow\mathbb{R}_+$ is a {\em split} of $y$ if $\forall S\subseteq\Vr$, $\sum_{v\in S} y'(S,v)=y(S)$ and $y'(S,v) > 0$ implies $v\in S$. A cost-share $s$ is {\em consistent with the dual solution $y$} if there exists a split $y'$ of $y$ such that for all $v\in\Vr$, $s(v)=\sum_{S:v\in S} y'(S,v)$. We say that a cost-share(/imputation), $s$, is in the Owen set of Min-cost branching game if there exists a feasible dual solution $y$ such that $s$ is consistent with $y$. 

\begin{lemma}
    For the min-cost branching game, the Owen set is a subset of the core.
\end{lemma}

\begin{proof}
    Let $s$ be an Owen set cost-share , i.e., there is a split $y'$ of a feasible dual solution $y$ such that $s$ is consistent with $y$. To show that it is also in the core, we need to show that the total cost of any sub-coalition $S$ of vertices is at most the cost of a min-cost branching of $S\cup\{r\}$, say $T$. We will show that
    
    $$\forall S\subseteq V\setminus \{r\}, \sum_{v\in S} s(v) \leq \sum_{X\cap S \neq \phi} y_X \leq \sum_{e\in T} c_e$$

    Both the inequalities follow from the definitions of Owen set cost-share and rearrangement of terms as below. For the first inequality, see that
    $$ \sum_{v\in S} s(v) 
    = \sum_{v\in S} \sum_{X:v\in X} y'(X,v) 
    \leq \sum_{X\cap S \neq \phi} \sum_{v\in X}y'(X,v) 
    = \sum_{X\cap S \neq \phi} y_X $$

    The second inequality can be derived as follows
    $$\sum_{e\in T} c_e 
    \geq \sum_{e\in T} \sum_{X:e\in\btd(X)} y_X 
    = \sum_{X} y_X\cdot |T\cap\btd(X)| 
    \geq \sum_{X: X\cap S \neq \phi} y_X\cdot |T\cap\btd(X)|
    \geq \sum_{X: X\cap S \neq \phi} y_X $$
     
    Thus all Owen set imputations are in the core.   
\end{proof}
However, not all core imputations are in Owen set and the following example (see Figure~\ref{fig:tree}) establishes this. Consider a graph with vertices $V=\set{v_1,v_2,v_3,u_1,u_2,u_3,a,b,r}$ and edges $E=E_1\cup E_2\cup E_3\cup\set{(a,b),(b,r)}$, where $E_1=\{(v_1,u_1),(v_1,u_2),(v_2,u_2),(v_2,u_3),\allowbreak(v_3,u_3),(v_3,u_1)\}$, $E_2=\set{(u_1,a),(u_2,a),(u_3,a)}$ and $E_3=\set{(u_1,b),(u_2,b),(u_3,b)}$. All edges in $E_2$ have cost 0 and edge $(a,b)$ has cost 2. All other edges have unit cost. The root is $r$. 
\begin{figure}
    \centering
    \includegraphics[scale=0.5]{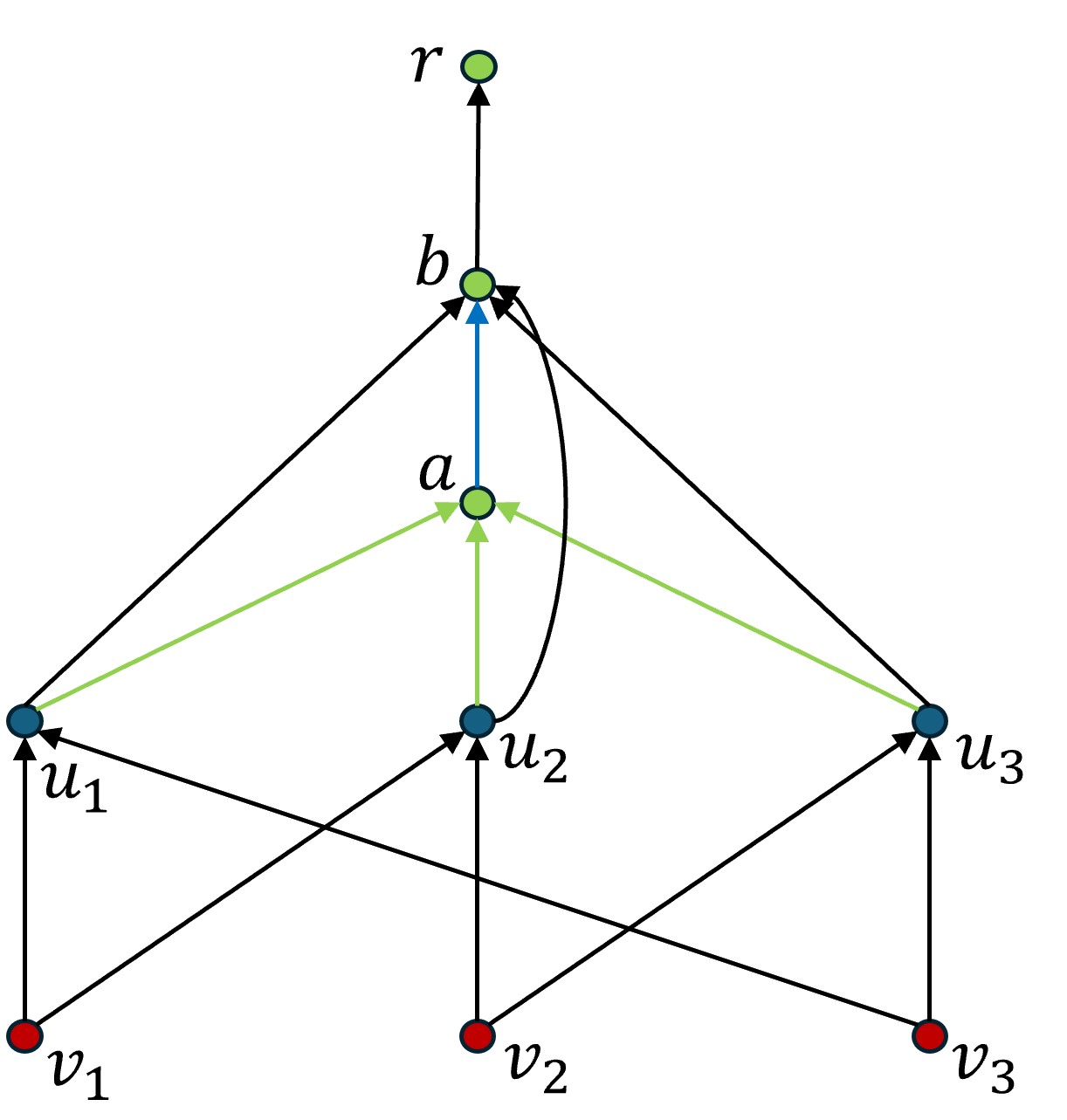}
    \caption{The green edges have cost 0 and the black edges are unit cost. The blue edge has cost 2.}
    \label{fig:tree}
\end{figure}

It is easy to check that a cost-share of 2 to each vertex in the set $X=\set{v_1, v_2, v_3}$ and 0 to all other vertices is a core imputation. However, this cost-share is not in the Owen set. Note that the minimum cost branching in this instance has cost 6 and this is also the total cost-share of vertices in $X$. Hence, if there exists a feasible dual solution $y$ which can be split to obtain this cost-share, then $y(S) > 0$ implies $S\cap X\neq\phi$. 

Note that if $S\cap X\neq\phi$ and $b\not\in S$ then $\abs{\btd(S)\cap (E_1\cup E_3)}\ge 2$. Since total cost of edges in $E_1\cup E_3$ is 9, the total $y$-value of sets $S$ such that $S\cap X\neq\phi, b\not\in S$ is at most 4.5. The total $y$-value of sets $S$ such that $S\cap X\neq\phi, b\in S$ is at most 1 since $(b,r)\in\btd(S)$. Thus the total y-value of sets $S$ such that $S\cap X\neq\phi$ is at most 5.5 which implies that the total cost-share of vertices in $T$ in any Owen set imputation is at most 5.5. Thus the imputation which assigns a cost-share of 2 to each vertex in $T$ is not in Owen set.

\begin{lemma}
    Given an imputation, we can efficiently decide if it is in the Owen set by checking if their exists a corresponding dual optimal solution. 
\end{lemma}

Let $p=\{p_v\}_{v\in \Vr}$ be the imputation. Consider LP~\ref{eq.mst-primal}. The feasibility version of this LP where the constraint on the vertices is replaced with $\sum_{S\subseteq\Vr:v\in S} y(S,v) = p_v, \forall v\in\Vr$. This LP is feasible if and only if the imputation is in the Owen set. While this LP has exponentially many variables, its dual can be solved efficiently using the ellipsoid method. The separation oracle is same for the LP required to obtain the leximin Owen set imputation and is described in the following section.

\subsection{Finding leximin Owen set imputation using linear programming}
\label{section:leximinMST}

We begin by formulating a linear program(LP) for finding an Owen set imputation that maximizes the minimum cost-share of any vertex. Our LP has non-negative variables $y(S,v), v\in S, S\subseteq\Vr$, which is the share of vertex $v$ in the dual variable $y(S)$. Thus $y(S)=\sum_{v\in S} y(S,v)$ and hence the feasibility of the dual solution can be captured by the constraints $\sum_{S:e\in\btd(S)}\sum_{v\in S} y(S,v)\le c(e), \forall e\in E$. The cost-share of a vertex $v$ equals $\sum_{S:v\in S} y(S,v)$ and since we wish to maximize the minimum cost we should include the constraint $\sum_{S:v\in S} y(S,v)\ge \lambda, \forall v\in\Vr$, and maximize $\lambda$.

To ensure that $p(v)=\sum_{S:v\in S} y(S,v)$ is an imputation we require that $$\sum_{v\in\Vr} p(v)=\sum_{v\in\Vr} \sum_{S:v\in S} y(S,v) = \sum_{S\subseteq\Vr} y(S)=\opt$$ where $\opt$ is the cost of the minimum cost branching in $G$.

We now formulate a series of LPs designed to return a leximin Owen set imputation. A corresponding set of LPs for finding a leximax Owen set imputation, along with their respective separation oracle, can be found in Appendix~\ref{app:leximaxMST}.\\

Consider the following LP that returns an Owen set imputation that maximizes the minimum cost-share. 

\begin{maxi}
		{} {\lambda}
			{\label{eq.mst-primal}}
		{}
        \addConstraint{\sum_{S\subseteq\Vr:e\in\btd(S)}\sum_{v\in S} y(S,v)}{\le c(e)}{\quad\forall e\in E}
		\addConstraint{\sum_{S\subseteq\Vr:v\in S} y(S,v)}{\ge \lambda}{\quad\forall v\in\Vr}
        \addConstraint{\sum_{S\subseteq\Vr} \sum_{v\in S} y(S,v)}{=\opt}{}
        \addConstraint{y(S,v)} {\geq 0}{\quad\forall S\subseteq\Vr, \forall v\in S}
\end{maxi} 
Since the number of variables in this LP is exponentially large we consider the dual of this LP that has variables $z(e), e\in E$, $z(v), v\in\Vr$, and $\beta$. 

\begin{mini}
		{} {\sum_{e\in E} z(e)c(e)-\beta\opt}
			{\label{eq.mst-dual}}
		{}
        \addConstraint{\sum_{e\in\btd(S)} z(e)}{\ge z(v)+\beta}{\quad\forall S\subseteq\Vr, \forall v\in S}
        \addConstraint{\sum_{v\in\Vr} z(v)}{=1}
		\addConstraint{z(e)}{\geq 0}{\forall e \in E}
        \addConstraint{z(v)}{\geq 0}{\forall v \in \Vr}
	\end{mini} 

Although LP \ref{eq.mst-dual} has exponentially many constraints, it can be solved in polynomial-time using the ellipsoid method since there is a polynomial-time separation oracle to identify a violating constraint. Given an assignment of values to edges $z(e),e\in E$ and vertices $z(v), v\in V, \sum_{v\in\Vr} z(v)=1$ and a value $\beta$, we should evaluate whether the first constraint from LP~\ref{eq.mst-dual} corresponding to some set $S\subseteq\Vr$ and some $v\in S$ fails. 

\textbf{Seperation Oracle:} Let $\alpha(v)$ be the maximum flow from $v$ to $r$ in the graph $G$ with edge capacities given by $z(e), e\in E$. If $\alpha(v)$ is less than $z(v)+\beta$ then by the max-flow min-cut theorem there exists a set $S\subseteq\Vr, v\in S$ such that $\sum_{e\in\btd(S)} z(e) < z(v)+\beta$. The minimum cut separating $v$ and $r$ then identifies the violated inequality. If for all $v\in\Vr$, $\alpha(v)\ge z(v)+\beta$ then no constraint is being violated.


Given this separation oracle, we can solve the dual LP \ref{eq.mst-dual} in polynomial-time to obtain the optimal dual solution $ z^* $. By strong duality, the optimum primal objective, denoted by $ \alpha^* $, equates to the optimum dual objective. Let $ y^* $ denote any optimum solution of the primal LP \ref{eq.mst-primal}. By the complementary slackness condition, for all vertices $ v \in V $ such that $ z_v^* > 0 $, it holds that $ \sum_{S : v \in S} y^*(S, v) = \alpha^* $. This implies that the set of vertices with positive duals, denoted as $ V_{\alpha^*} $, are assigned the cost value $ \alpha^* $ in every feasible max-min solution; hence, these vertices form a subset of the minimum set of vertices that must be fixed at $ \alpha^* $. Given that the dual optimum solutions satisfy the feasibility constraint $ \sum_{v \in V \setminus \{r\}} z^*(v) = 1 $, there exists at least one vertex with a positive dual, implying that $ V_{\alpha^*} $ is non-empty. 

Let $V_F$ denote the set of vertices whose values have been fixed so far. This set may include all vertices except $r$. The mapping $m: V \setminus \{r\} \to \mathbb{R}_+$ assigns the fixed value to vertices in $V_F$ and is zero for all others.
Next, the vertices in $ V_{\alpha^*} $ are fixed, and the set of fixed variables $V_F$ along with the mapping $m(\cdot)$ are updated accordingly. 
The procedure iterates as long as there is a vertex in $V\setminus\{r\}$ not in $V_F$.  Since in every iteration at least the cost-share of one vertex is fixed, there are at most $ n $ iterations, ensuring that the problem is solvable in polynomial-time. Note that only vertices are paying costs, and each vertex pays a total of $\sum_{S:v\in S} y(S,v)$. Therefore at every iteration we find the max-min cost that should be assigned, and we assign it to a minimum set of vertices. 

The following set of LPs return the leximin Owen set imputation in the min-cost branching game. The procedure stops when all vertex costs are fixed. 
For the specific problem of finding the leximin Owen set imputation, the constraint $\sum_{S \subseteq \Vr} \sum_{v \in S} y(S,v) = \opt$ and the variables $y(S,v)$ for the exponentially many non-contiguous sets $S$ can be eliminated to further simplify the series of s. However, the resulting set of LPs necessitates a more complicated separation oracle, which is discussed in Appendix~\ref{app:MST_complex_LPs}.

\begin{maxi}
		{} {\alpha}
			{\label{eq.mst-primal-}}
		{}
        \addConstraint{\sum_{S \subseteq\Vr:e\in\btd(S)}\sum_{v\in S} y(S,v)}{\le c(e)}{\quad\forall e\in E}
		\addConstraint{\sum_{S \subseteq\Vr:v\in S} y(S,v)}{\ge \alpha}{\quad\forall v\notin V_F, v\neq {r}}
  		\addConstraint{\sum_{S \subseteq\Vr:v\in S} y(S,v)}{=  m(v)}{\quad\forall v\in V_F}
        \addConstraint{\sum_{S \subseteq \Vr} \sum_{v\in S} y(S,v)}{=\opt}{}
        \addConstraint{y(S,v)}{\geq 0}{\quad\forall S\subseteq\Vr,\forall v\in S}
	\end{maxi} 

The dual program is the following:
\begin{mini}
		{} {\sum_{e\in E} z(e)c(e) - \sum_{v \in V_F } z(v) \text{m}(v)-\beta\opt}
			{\label{eq.mst-dual-}}
		{}
        \addConstraint{\sum_{e\in\btd(S)} z(e)}{\ge z(v)+\beta}{\quad\forall S\subseteq\Vr, \forall v\in S}
        \addConstraint{\sum_{v\notin V_F, v\neq r} z(v)}{=1}
		\addConstraint{z(e)}{\geq 0}{\quad\forall e \in E}
        \addConstraint{z(v)}{\geq 0}{\quad\forall v \notin V_F, v\neq r }
\end{mini} 

\begin{theorem}
\label{thm:leximinminbranching}
    The leximin and leximax Owen set imputations for the min-cost branching game can be found efficiently.
\end{theorem}

While a combinatorial algorithm to find the leximin Owen set imputation might exist, we prove that such an algorithm a slightly generalized problem - given a set $T$ of vertices, find an Owen set imputation which maximizes the minimum cost-share of a vertex in $T$ - will yield a combinatorial algorithm for finding the optimum fractional setcover. The latter problem is well-researched and while there are combinatorial algorithms known to compute a $(1+\epsilon)$ approximation to the optimum value, the only algorithm known for determining the exact value is to solve the LP. 

\begin{restatable}{theorem}{mstleximincombi}
\label{thm:mst_leximin_combi}
${}^\dagger$
A combinatorial algorithm that maximises the minimum cost-share of a vertex in $T$ in an Owen set imputation would yield a combinatorial algorithm for finding the optimum fractional set cover.  
\end{restatable}

\section{\texorpdfstring{$b$}{b}-Matching Game}

\label{sec:b_matching_game}
\renewcommand{\deg}[2]{\mbox{deg}_{#1}(#2)}

We would like to obtain the leximin(/leximax) fair imputations in the core of these games. But the following theorem states that finding a sub-coalition failing core constraint for an ``approximate'' profit share is NP-hard.
We expect this NP-hardness result to extend to normal profit shares and to computing leximin(/leximax) core imputations in the same way as with the max-flow and MST games\footnote{A very recent paper \cite{GTV} has established the NP-hardness of these problems.}.

\begin{restatable}{lemma}{bmatchhardness}
\label{thm:bmatch-hardness}
${}^\dagger$
    Given a $b$-matching game on a bipartite graph $G=(U,V,E), w: E\rightarrow\mathbb{R}_+, b: U \cup V \rightarrow \mathbb{Z}_+ $ and an ``approximate'' profit-share $p:U\cup V\rightarrow\mathbb{R}_+$ such that $p(U\cup V)=\alpha\cdot v(U\cup V)$, deciding if there exists a coalition $S\subseteq U\cup V$ such that $p(S)<v(S)$ is NP-complete for any $\alpha>1$.
    
\end{restatable}

And so, we focus our attention on the Owen set of the $b$-Matching game and we show how to efficiently compute leximin Owen set imputations for this game. 

\subsection{Owen set imputations for the \texorpdfstring{$b$}{b}-matching game}

The linear program, LP~\ref{eq.b-uncon-core-primal-bipartite}, depicts the LP relaxation for finding a max-weight bipartite $ b $-matching. In this formulation, the variable $ x_{ij} $ indicates the degree to which edge $ (i, j) $ is included in the solution. Note that the variables $ x_{ij} $ are not bounded above as an edge may be matched more than once. For ease of notation, we will use $b_i$ instead of $b(i)$ for the capacity of vertices.

	\begin{maxi}
		{} {\sum_{(i, j) \in E}  {w_{ij} x_{ij}}}
			{\label{eq.b-uncon-core-primal-bipartite}}
		{}
		\addConstraint{\sum_{(i, j) \in E} {x_{ij}}}{\leq b_i \quad}{\forall i \in U}
		\addConstraint{\sum_{(i, j) \in E} {x_{ij}}}{\leq b_j }{\forall j \in V}
		\addConstraint{x_{ij}}{\geq 0}{\forall (i, j) \in E}
	\end{maxi}

Taking $u_i$ and $v_j$ to be the dual variables for the first and second constraints of (\ref{eq.b-uncon-core-primal-bipartite}), we obtain the dual LP: 

 	\begin{mini}
		{} {\sum_{i \in U}  {b_i u_{i}} + \sum_{j \in V} {b_j v_j}} 
			{\label{eq.b-uncon-core-dual-bipartite}}
		{}
		\addConstraint{ u_i + v_j}{ \geq w_{ij} \quad }{\forall (i, j) \in E}
		\addConstraint{u_{i}}{\geq 0}{\forall i \in U}
		\addConstraint{v_{j}}{\geq 0}{\forall j \in V}
	\end{mini}


\begin{definition}
    A bipartite $b$-matching core imputation $(p_U, p_V)$ is in the \textbf{Owen set} if there exists an optimal solution $(u, v)$ for the dual LP \ref{eq.b-uncon-core-dual-bipartite} and profit of each agent $i \in U$ is $p_U(i) = b_i u_i$ and the profit of each agent $j \in V$ is $p_V(j) = b_j v_j$.
\end{definition}

In \cite{Shapley1971assignment}, authors show that for the uniform bipartite $b$-matching game, the dual completely characterizes the core. However for the non-uniform bipartite $b$-matching game, all Owen set imputations are in the core but not all core imputations are in the Owen set. The following example from \cite{vazirani2023lpduality} provides one such instance.

\begin{figure}[ht]
\begin{center}
\begin{tikzpicture}[scale=1.5]
  \tikzstyle{vertex}=[circle,fill=blue!25,minimum size=20pt,inner sep=0pt]
  
  \node[vertex] (u) at (0,0) {$u$};
  \node[vertex] (v1) at (-2,0) {$v_1$};
  \node[vertex] (v2) at (2,0) {$v_2$};

  \draw (u) -- (v1) node[midway, above] {$1$};
  \draw (u) -- (v2) node[midway, above] {$3$};
\end{tikzpicture}
\caption{The graph for Example \ref{ex.b-arbitrary}.}
\label{fig.bmatchingexample}
\end{center}
\end{figure}


\begin{example}
	\label{ex.b-arbitrary}

 For the bipartite $ b $-matching game outlined by the graph in Figure \ref{fig.bmatchingexample}, assume the $ b $ values are set at $ (2, 2, 1) $ for vertices $ (u, v_1, v_2) $, and the weights of the edges $ (u, v_1) $ and $ (u, v_2) $ are given as $ 1 $ and $ 3 $ respectively.

When both edges are matched once, the value of the game totals 4. The unique optimal dual solution is $ (1, 0, 2) $ corresponding to $ (u, v_1, v_2) $. It can be observed that the profit-share of $ (4, 0, 0) $ amongst $ (u, v_1, v_2) $ is in the core. The corresponding dual solution would be $ (2, 0, 0) $; nonetheless, such a solution is not dual feasible. As such, this particular core imputation does not mirror an optimal dual solution.
\end{example}

\subsection{Combinatorial algorithm to compute the leximin Owen set imputation}
\label{sec:b_matching_comb_algo}
The leximin Owen set imputation for the $b$-matching game can be computed using the LP method (see \Cref{app:bmatching_lp}). We can also reduce the problem to finding leximin core imputation of an assignment game(see \Cref{app:bmatching_comb}) albeit in exponential time,. Below, we provide a high-level overview of an efficient combinatorial algorithm adapted from \cite{Vazirani-leximin} for the assignment game. A more detailed description of the algorithm is given in Appendix~\ref{app:b-matching-algo}.

\cite{Vazirani-leximin} first divides each edge $(i,j)$ into three groups—\emph{essential}, \emph{viable}, and \emph{subpar}— based on whether the edges are matched $min(b_i,b_j)$ times in all, some or no maximum weight $b$-matching. The graph $G=(V,E)$ is then restricted to just the essential and viable edges to give graph $H=(V,T_0)$; note that complementarity slackness conditions ensure that these edges are always tight in LP~\ref{eq.b-uncon-core-dual-bipartite}. Each connected component of $H$ is labeled as a \emph{unique imputation} component if all vertices receive the same profit share in every Owen set imputation, or a \emph{fundamental} component otherwise.

A key feature of the algorithm in \cite{Vazirani-leximin} is that the profit shares of vertices in a fundamental component can be “rotated” without violating core conditions. In the assignment game, this means profits (dual values) on the $U$-side can be increased (or decreased) while the corresponding profits on the $V$-side are decreased (or increased) by the same amount. For the $b$-matching game, this generalizes: the duals can still be adjusted at the same rate on both sides (with opposite signs), but the profit changes are proportional to the $b_i$  values of each vertex. The following lemma proves that such profit adjustments will still yield a valid Owen set imputation.

\begin{restatable}{lemma}{lemcomponentsarebalanced}
\label{lem:components_are_balanced}
    Let $C$ be any fundamental component in $H_0$. Then the sum of $b_i$ values of the $U$ vertices is same as the sum of $b_i$ values of the $V$ vertices. 
\end{restatable}

\begin{proof}

    Let the sum of the $b_i$ values on the $U$ side and $V$ side of $C$ be denoted by $b_U$ and $b_V$, respectively. Let $x = (u, v)$ be an optimal solution satisfying the \emph{strict complementarity conditions} of the dual LP, and let $x'$ be another optimal solution of the dual LP that differs from $x$ only at the vertices within $C$. The existence of $x'$ is guaranteed because $C$ is a fundamental component with no subpar edges, and since $x$ satisfies strict complementarity, all subpar edges outside of $C$ are over-tight, allowing for small changes in the dual variables without losing feasibility of the dual LP.

    Consider the change in dual values between $x$ and $x'$. Since the increment in $u_i$ values from $x$ to $x'$ is equal to the decrement in $v_j$ values, let's denote this value by $\delta$.

    Since both $x$ and $x'$ are optimal solutions of the dual LP and differ only at vertices within $C$, their profits at vertices outside $C$ are identical. Consequently, the total profit-share within $C$ must also be identical in both $x$ and $x'$. By comparing the total profit-shares of the vertices in $C$, we obtain the difference in profits under $x$ and $x'$ is as $(b_U - b_V) \cdot \delta$ which should be $0$ since both $x$ and $x'$ are dual optimal. This implies that $b_U = b_V$, completing the proof. 

\end{proof}

The algorithm to compute the leximin Owen set imputation for the $b$-matching game will differ from the assignment game in just one key aspect. An arbitrary imputation is lexicographically improved, so that the profits of all minimum-profit vertices change at an equal rate. The duals must then move at a rate proportional to $1/b_v$ relative to the profits. Consequently, each fundamental component rotates at a corresponding rate. This modification gives us an efficient algorithm to compute the leximin Owen set imputation for the $b$-matching game.

\begin{theorem}
\label{thm:b_matching_run_time}
    The leximin and leximax Owen set imputations of the $b$-matching game can be computed in time $O(m^{2+o(1)})$. 
\end{theorem}

\bibliographystyle{alpha}
\bibliography{refs}
    
\appendix

\section{Deferred Proofs}
\label{app:proofs}

In this sections we will provide the proofs deferred in the paper. Before that, we describe the EXACT COVER BY 3-SETS(X3C) problem, which is the basis of all NP-hardness reductions in this paper. This NP-complete problem is stated as follows. \\

\begin{table}[!ht]
\centering
\renewcommand{\arraystretch}{1.5}
\begin{tabular}{|l|p{10cm}|}
\hline
\multicolumn{2}{|c|}{\textsc{EXACT COVER BY 3-SETS(X3C):}} \\
\hline
\textbf{Instance:} & A finite set of elements $X = \{x_1,..., x_{3q} \}$ and a collection $F = \{f_1 ..... f_{|F|}\}$ of 3-element subsets of $X$. \\
\hline
\textbf{Question:} & Is there an exact cover of $X$ in $F$, i.e., a subcollection $F'\subseteq F$ such that every element in $X$ occurs exactly once in $F'$. \\
\hline
\end{tabular}
\end{table}

\subsection{Proofs from Section~\ref{sec:flow_game}}

\flowleximinhardness*
\begin{proof}[Proof of Theorem~\ref{thm:flow-leximin-hardness}]

The theorem states that finding the maximin profit-share, which is equal to the minimum profit-share of an agent in the leximin imputation, is in itself an NP-hard problem. We prove this theorem via a sequence of lemmas. We start with the co-NP hardness result of \cite{Fang2002computational}. They show that given an instance of max-flow game and a profit-share, it is NP-hard to find a subcoalition that violates the core constraint, using a reduction from EXACT COVER BY 3-SETS(X3C). 

Given a graph $G$ an imputation $p:E\to\mathbb{R}_+$ and $E'\subseteq E$, let $p(E') = \sum_{e\in E} p(e)$.
\begin{lemma} (\cite{Fang2002computational} Theorem 2.1)\\
\label{lem:flow-hardness}
    Given a max-flow game on a graph $G=(V,E),$ $ c: E\rightarrow\mathbb{R}_+$ and an imputation $p:E\rightarrow\mathbb{R}_+$, deciding if there exists a subcoalition $S\subset E$ such that $p(S)<v(S)$ is NP-complete.
\end{lemma}

In \cite{Fang2002computational}, for a specific max-flow game $(G,c)$ and an imputation $p$ constructed based on an instance of EXACT COVER BY 3-SETS(X3C), it is shown that $p$ is not in the core of $(G,c)$, i.e., there exists a subcoalition $S\subset V$ such that $p(S)<v(S)$ if and only if $X$ has an exact cover in $F$. This shows the hardness of testing whether an imputation is in the core. Given this hardness result, we iteratively transform the graph, the edge capacities and the imputation maintaining the hardness. The final imputation is constructed such that every edge gets the same share, thus equal to maximin profit-share of the edges. We prove that this imputation is in the core if and only if an exact cover exists. This shows the NP-hardness. We state the result below and the complete proof is included in the Appendix~\ref{app:proofs}.

\begin{figure}[H]
    \centering
        \begin{tikzpicture}[scale=1.5]
    \tikzstyle{vertex}=[circle, fill=black, inner sep=1.5pt]

\node (s1) at (2.5,-3) [label=above:{$i=\{1,2,\cdots,|F|\}, j=\{1,2,\cdots,3q\}$}] {};
\node (s2) at (2,-2.3) [label=above:{$b_0$}] {};

\node[vertex] (s) at (-1,0) [label=left:$s$] {};
\node[vertex] (v0) at (0,0) [label=below:$v_0$] {};
\node (a0) at (-0.5,0) [label=below:$a_0$] {};
\node[vertex] (u1) at (2,1) [label=above:$u_i$] {};
\node[vertex] (u2) at (2,0.5) {};
\node (dots1) at (2,0) [label=right:$\vdots$] {};
\node[vertex] (u4) at (2,-0.5) {};
\node[vertex] (u5) at (2,-1) {};
\node[vertex] (v1) at (4,1) [label=above:$v_j$] {};
\node[vertex] (v2) at (4,0.5) {};
\node (dots2) at (4,0) [label=right:$\vdots$] {};
\node[vertex] (v4) at (4,-0.5) {};
\node[vertex] (v5) at (4,-1) {};
\node[vertex] (t) at (5,0) [label=right:$t$] {};
\node[vertex] (v0') at (0,-1.5) [label=below:$v_0'$] {};

\draw[postaction={decorate,decoration={markings,mark=at position 0.5 with {\arrow{stealth}}}}] (s) -- (v0) node[midway, above] {$3|F|$};

\draw[postaction={decorate,decoration={markings,mark=at position 0.5 with {\arrow{stealth}}}}] (v0) -- (u1) node[midway, above] {$3$};

\draw[postaction={decorate,decoration={markings,mark=at position 0.5 with {\arrow{stealth}}}}] (v0) -- (u2);
\draw[postaction={decorate,decoration={markings,mark=at position 0.5 with {\arrow{stealth}}}}] (v0) -- (u4);
\draw[postaction={decorate,decoration={markings,mark=at position 0.5 with {\arrow{stealth}}}}] (v0) -- (u5) node[pos=0.25, below] {$A_1$};

\draw[postaction={decorate,decoration={markings,mark=at position 0.5 with {\arrow{stealth}}}}] (u1) -- (v1) node[midway, above] {$1$};
\draw[postaction={decorate,decoration={markings,mark=at position 0.5 with {\arrow{stealth}}}}] (u1) -- (v2);

\draw[postaction={decorate,decoration={markings,mark=at position 0.5 with {\arrow{stealth}}}}] (u2) -- (v1);
\draw[postaction={decorate,decoration={markings,mark=at position 0.5 with {\arrow{stealth}}}}] (u2) -- (v2);
\draw[postaction={decorate,decoration={markings,mark=at position 0.5 with {\arrow{stealth}}}}] (u2) -- (3,0.25);

\draw[postaction={decorate,decoration={markings,mark=at position 0.5 with {\arrow{stealth}}}}] (u4) -- (3,-0.25);
\draw[postaction={decorate,decoration={markings,mark=at position 0.5 with {\arrow{stealth}}}}] (u4) -- (v4);
\draw[postaction={decorate,decoration={markings,mark=at position 0.5 with {\arrow{stealth}}}}] (u4) -- (v5);

\draw[postaction={decorate,decoration={markings,mark=at position 0.5 with {\arrow{stealth}}}}] (u5) -- (v4);
\draw[postaction={decorate,decoration={markings,mark=at position 0.5 with {\arrow{stealth}}}}] (u5) -- (v5) node[midway, below] {$A_2$};

\draw[postaction={decorate,decoration={markings,mark=at position 0.5 with {\arrow{stealth}}}}] (v1) -- (t) node[midway, above] {$1$};
\draw[postaction={decorate,decoration={markings,mark=at position 0.5 with {\arrow{stealth}}}}] (v2) -- (t);
\draw[postaction={decorate,decoration={markings,mark=at position 0.5 with {\arrow{stealth}}}}] (v4) -- (t);
\draw[postaction={decorate,decoration={markings,mark=at position 0.5 with {\arrow{stealth}}}}] (v5) -- (t) node[midway, below] {$A_3$};
\draw[postaction={decorate,decoration={markings,mark=at position 0.5 with {\arrow{stealth}}}}] (v0') .. controls (2,-2) and (5,-2) .. (t) node[midway, below] {$3|F|-3q$};

\draw[postaction={decorate,decoration={markings,mark=at position 0.5 with {\arrow{stealth}}}}] (u1) -- (v0') node[pos=0.9, above] {$3$};

\draw[postaction={decorate,decoration={markings,mark=at position 0.5 with {\arrow{stealth}}}}] (u2) -- (v0');
\draw[postaction={decorate,decoration={markings,mark=at position 0.5 with {\arrow{stealth}}}}] (u4) -- (v0');
\draw[postaction={decorate,decoration={markings,mark=at position 0.5 with {\arrow{stealth}}}}] (u5) -- (v0') node[midway, below] {$B$};

\end{tikzpicture}

    \caption{Network $G=(V,E)$}
    \label{fig:nphardnessflow}
\end{figure}

Consider the following instance of the EXACT COVER BY 3-SETS(X3C) problem - a finite set with $X = \{x_1,..., x_{3q} \}$ and a collection $F = {f_1 ..... f_{|F|}}$ of 3-element subsets of $X$. The max-flow game considered in \cite{Fang2002computational}(and lemma~\ref{lem:flow-hardness}) is as follows(some notation has been changed to be consistent with the definition above). The graph $G=(V,E)$ has $|X|+|F|+2$ vertices - $V = V_X \cup V_F \cup \{v_0,v_0',s,t\}$ where $ V_X = \{v_1,v_2,\ldots ,v_{3q}\}$ are the element vertices corresponding to $x_1,x_2,\ldots,\text{ and }x_{3q}$ respectively and $V_F = \{ u_1,u_2, \ldots,\text{ and } u_{|F|}\}$ are the set vertices corresponding to $f_1,f_2,\ldots, f_{|F|}$ respectively. $s$ and $t$ are source and sink vertices while $v_0$ and $v_0'$ are two other supporting vertices.
The vertices are arranged as shown in figure~\ref{fig:nphardnessflow}. The edge set, $E$, and their capacities are defined as follows. 
\begin{itemize}
    \item $a_0 = (s,v_0),\text{ } c(a_0) = 3|F|$
    \item $b_0 = (v_0',t),\text{ } c(b_0) = 3|F|-3q$
    \item $A_1 = \{ (v_0,u_i): i \in \{ 1,2,\ldots,|F|\} \},\text{ } \forall e\in A_1,\text{ } c(e)=3$
    \item $A_2 = \{ (u_i,v_j): i \in \{ 1,2,\ldots,|F|\}, j \in \{ 1,2,\ldots,3q\} \},\text{ } \forall e\in A_2,\text{ } c(e)=1$
    \item $A_3 = \{ (v_j,t): j \in \{ 1,2,\ldots,3q\} \},\text{ } \forall e\in A_3,\text{ } c(e)=1$
    \item $B = \{ (u_i,v_0'): i \in \{ 1,2,\ldots,|F|\} \},\text{ } \forall e\in B,\text{ } c(e)=3$
\end{itemize}

The maximum flow in this graph is $3|F|$. Consider imputation $p\in \mathbb{R}_+^{|E|}$ defined as
\begin{itemize}
    \item $p(a_0) = 2q-\frac{1}{3|F|} $
    \item $p(b_0) = 3|F|-3q+(\frac{1}{3|F|}-\frac{1}{3(q+1)})$
    \item $p(e) = \frac{1}{3|F|(q+1)}, \text{ } \forall e\in A_1$
    \item $p(e) = 0, \text{ } \forall e \in A_2$
    \item $p(e) = \frac{1}{3}, \text{ } \forall e \in A_3$
    \item $p(e) = 0, \text{ } \forall e \in B$ 
\end{itemize}

Consider the graph $G$ and imputation $p$ above. We will first modify the edge capacities and profit-shares such that, every edge-player gets a positive profit. For this, let's look at a toy example. Consider max-flow game on a graph which is just a path of edges $e_1,e_2, \ldots, e_k$ between the source and the sink with edge capacities $c_1,c_2,\ldots, c_k$ respectively. The maximum flow,$(f)$, in the graph is equal to the minimum among $c_1,c_2,\ldots, c_k$. Giving all this flow $f$ to any single edge and giving zero profit to the other edges is trivially in the core of the game as no (strict) subcoalition of the grand coalition can get anything more than zero. We will use this idea below in the following way. The max-flow game described in lemma~\ref{lem:flow-hardness}, is giving zero profits to the edges in $A_2$ and $B$. The game also has $4|F|$ unique $s-t$ paths, each of which can be uniquely determined by an edge in either $A_2$ or $B$. We will increase the edge capacities along these $4|F|$ paths by a ``small'' amount $\alpha$($>0$, to be decided later). Note that, if an edge is present in multiple paths, its capacity will be increased multiple times. That should increase the total flow by $4|F|\alpha$. We will distribute this additional profit equally among just the edges in $A_2$ and $B$ and argue that a subcoalition $S$ is violating the core constraint in the new graph if and only if it was violating the core constraint in the previous graph shown in figure \ref{fig:nphardnessflow}.  

Formally, consider the modification of $G=(V,E)$ and $p$, $G_\alpha=(V_\alpha,E_\alpha)$ and $p_\alpha$, where the vertices and edges are the same - $V_\alpha=V,E_\alpha=E$ and the edge capacities and profit-shares are as described below. See figure~\ref{fig:hardnessflowalpha}).

\begin{figure}[H]
    \centering
        \begin{tikzpicture}[scale=2.2, every node/.append style={scale=0.9}]
    \tikzstyle{vertex}=[circle, fill=black, inner sep=1.5pt]

\node (s1) at (2.5,-2.5) [label=above:{$i=\{1,2,\cdots,|F|\}, j=\{1,2,\cdots,3q\}$}] {};
\node (s2) at (2,-2.1) [label=above:{$b_0$}] {};

\node[vertex] (s) at (-1,0) [label=left:$s$] {};
\node[vertex] (v0) at (0,0) [label=below:$v_0$] {};
\node (a0) at (-0.5,0) [label=below:$a_0$] {};
\node[vertex] (u1) at (2,1) [label=above:$u_i$] {};
\node[vertex] (u2) at (2,0.5) {};
\node (dots1) at (2,0) [label=right:$\vdots$] {};
\node[vertex] (u4) at (2,-0.5) {};
\node[vertex] (u5) at (2,-1) {};
\node[vertex] (v1) at (4,1) [label=above:$v_j$] {};
\node[vertex] (v2) at (4,0.5) {};
\node (dots2) at (4,0) [label=right:$\vdots$] {};
\node[vertex] (v4) at (4,-0.5) {};
\node[vertex] (v5) at (4,-1) {};
\node[vertex] (t) at (5,0) [label=right:$t$] {};
\node[vertex] (v0') at (0,-1.5) [label=below:$v_0'$] {};

\draw[postaction={decorate,decoration={markings,mark=at position 0.5 with {\arrow{stealth}}}}] (s) -- (v0) node[midway, above] {$3|F| + 4|F|\alpha$};

\draw[postaction={decorate,decoration={markings,mark=at position 0.5 with {\arrow{stealth}}}}] (v0) -- (u1) node[midway, above, yshift=5pt] {$3+4\alpha$};

\draw[postaction={decorate,decoration={markings,mark=at position 0.5 with {\arrow{stealth}}}}] (v0) -- (u2);
\draw[postaction={decorate,decoration={markings,mark=at position 0.5 with {\arrow{stealth}}}}] (v0) -- (u4);
\draw[postaction={decorate,decoration={markings,mark=at position 0.5 with {\arrow{stealth}}}}] (v0) -- (u5) node[pos=0.25, below] {$A_1$};

\draw[postaction={decorate,decoration={markings,mark=at position 0.5 with {\arrow{stealth}}}}] (u1) -- (v1) node[midway, above] {$1+\alpha$};
\draw[postaction={decorate,decoration={markings,mark=at position 0.5 with {\arrow{stealth}}}}] (u1) -- (v2);

\draw[postaction={decorate,decoration={markings,mark=at position 0.5 with {\arrow{stealth}}}}] (u2) -- (v1);
\draw[postaction={decorate,decoration={markings,mark=at position 0.5 with {\arrow{stealth}}}}] (u2) -- (v2);
\draw[postaction={decorate,decoration={markings,mark=at position 0.5 with {\arrow{stealth}}}}] (u2) -- (3,0.25);

\draw[postaction={decorate,decoration={markings,mark=at position 0.5 with {\arrow{stealth}}}}] (u4) -- (3,-0.25);
\draw[postaction={decorate,decoration={markings,mark=at position 0.5 with {\arrow{stealth}}}}] (u4) -- (v4);
\draw[postaction={decorate,decoration={markings,mark=at position 0.5 with {\arrow{stealth}}}}] (u4) -- (v5);

\draw[postaction={decorate,decoration={markings,mark=at position 0.5 with {\arrow{stealth}}}}] (u5) -- (v4);
\draw[postaction={decorate,decoration={markings,mark=at position 0.5 with {\arrow{stealth}}}}] (u5) -- (v5) node[midway, below] {$A_2$};

\draw[postaction={decorate,decoration={markings,mark=at position 0.5 with {\arrow{stealth}}}}] (v1) -- (t) node[midway, above, yshift=11pt] {$1+a_j\alpha$};
\draw[postaction={decorate,decoration={markings,mark=at position 0.5 with {\arrow{stealth}}}}] (v2) -- (t);
\draw[postaction={decorate,decoration={markings,mark=at position 0.5 with {\arrow{stealth}}}}] (v4) -- (t);
\draw[postaction={decorate,decoration={markings,mark=at position 0.5 with {\arrow{stealth}}}}] (v5) -- (t) node[midway, below] {$A_3$};
\draw[postaction={decorate,decoration={markings,mark=at position 0.5 with {\arrow{stealth}}}}] (v0') .. controls (2,-2) and (5,-2) .. (t) node[midway, below, yshift=-5pt] {$3|F|-3q+\alpha|F|$};

\draw[postaction={decorate,decoration={markings,mark=at position 0.5 with {\arrow{stealth}}}}] (u1) -- (v0') node[pos=0.9, above, yshift=10pt] {$3+\alpha$};

\draw[postaction={decorate,decoration={markings,mark=at position 0.5 with {\arrow{stealth}}}}] (u2) -- (v0');
\draw[postaction={decorate,decoration={markings,mark=at position 0.5 with {\arrow{stealth}}}}] (u4) -- (v0');
\draw[postaction={decorate,decoration={markings,mark=at position 0.5 with {\arrow{stealth}}}}] (u5) -- (v0') node[midway, below] {$B$};

\end{tikzpicture}
    \caption{Network $G_\alpha=(V_\alpha,E_\alpha)$}
    \label{fig:hardnessflowalpha}
\end{figure}

\begin{itemize}
    \item $c_\alpha(a_0) = 3|F|+4|F|\alpha$
    \item $c_\alpha(b_0) = 3|F|-3q + |F|\alpha$
    \item $c_\alpha(e)=3+4\alpha, \forall e\in A_1,$
    \item $c_\alpha(e)=1+\alpha, \forall e\in A_2$
    \item $c_\alpha(e)=1+(a_j)\alpha, \forall e\in A_3,$ where $a_j$ is the number of times $x_j$ occurs in the collection $F$
    \item $c_\alpha(e)=3+\alpha, \forall e\in B$
\end{itemize}

The maximum flow in this graph is $3|F|+4|F|\alpha$. Consider the imputation $p_\alpha\in \mathbb{R}_+^{|E|}$ defined as
\begin{itemize}
    \item $p_\alpha(a_0) = p(a_0) = 2q-\frac{1}{3|F|} $
    \item $p_\alpha(b_0) = p(b_0) = 3|F|-3q+(\frac{1}{3|F|}-\frac{1}{3(q+1)})$
    \item $p_\alpha(e) = p(e) = \frac{1}{3|F|(q+1)}, \text{ } \forall e\in A_1$
    \item $p_\alpha(e) = \alpha, \text{ } \forall e \in A_2$
    \item $p_\alpha(e) = p(e) = \frac{1}{3}, \text{ } \forall e \in A_3$
    \item $p_\alpha(e) = \alpha, \text{ } \forall e \in B$ 
\end{itemize}

\begin{lemma}
    $p_\alpha$ is not in the core of $(G_\alpha,c_\alpha)$ if and only if $p$ is not in the core of $(G,c)$
\end{lemma}

\begin{proof}
    Assume $p$ is in the core of $(G,c)$ but $p_\alpha$ is not in the core of $(G_\alpha,c_\alpha)$ as a subcoalition $S\subset E$ violates core constraint. Let $n_2$ and $n_B$ be the number of edges in $A_2\cap S$ and $B\cap S$ respectively. Note that $p_\alpha(S) = p(S)+ (n_2+n_B)\alpha$ as only edges in $A_2$ and $B$ get additional profit of $\alpha$ each. Also note that, the worth of $S$ in $G_\alpha$, $v_\alpha(S) \leq v(S) + (n_2+n_B)\alpha$ as at most $(n_2+n_B)$ paths are present in $S$ and each can carry an additional flow of $\alpha$. Then, $S$ would also violate core constraint in the original graph as $ p(S) - v(S) \leq (p_\alpha(S) - (n_2+n_B)\alpha) - (v_\alpha(S) - (n_2+n_B)\alpha) = (p_\alpha(S)-v_\alpha(S)) < 0$, giving a contradiction.

    Now assume $p$ was not in the core of $(G,c)$. Lemma~\ref{lem:flow-hardness} implies that there exists an exact cover of $X$ and a set of edges $S$ which fails the core constraint. This set $S$ is exactly the $3q$ paths corresponding to the exact cover in $F$. $p_\alpha$ also violates the core constraint on $S$ as, $p_\alpha(S) - v_\alpha(S) = (p(S)+3q\alpha) - (v(S) + 3q\alpha) = p(S)-v(S) < 0$. This proves the other implication in the lemma.
\end{proof}

Let us choose $\alpha = \frac{1}{3|F|(q+1)}$, making the profits of edges in $A_1,A_2$ and $B$ equal. The reason for this value will be explained later in lemma~\ref{lem:flow-hardness-final}. Consider the new max-flow game $G'=(V',E')$ and imputation $p'$, where the vertices and edges are the same - $V'=V_\alpha,E'=E_\alpha$ and the capacity and profits of an edge $e$ in $G'$ are $(3|F|(q+1))$ times those in $G_\alpha$(see figure~\ref{fig:flow-hardness-multiplied}). 
That is, 

\begin{figure}[H]
    \centering
        \begin{tikzpicture}[scale=2.2, every node/.append style={scale=0.9}]
    \tikzstyle{vertex}=[circle, fill=black, inner sep=1.5pt]

\node (s1) at (2.5,-2.5) [label=above:{$i=\{1,2,\cdots,|F|\}, j=\{1,2,\cdots,3q\}$}] {};
\node (s2) at (2,-2.1) [label=above:{$b_0$}] {};

\node[vertex] (s) at (-1,0) [label=left:$s$] {};
\node[vertex] (v0) at (0,0) [label=below:$v_0$] {};
\node (a0) at (-0.5,0) [label=below:$a_0$] {};
\node[vertex] (u1) at (2,1) [label=above:$u_i$] {};
\node[vertex] (u2) at (2,0.5) {};
\node (dots1) at (2,0) [label=right:$\vdots$] {};
\node[vertex] (u4) at (2,-0.5) {};
\node[vertex] (u5) at (2,-1) {};
\node[vertex] (v1) at (4,1) [label=above:$v_j$] {};
\node[vertex] (v2) at (4,0.5) {};
\node (dots2) at (4,0) [label=right:$\vdots$] {};
\node[vertex] (v4) at (4,-0.5) {};
\node[vertex] (v5) at (4,-1) {};
\node[vertex] (t) at (5,0) [label=right:$t$] {};
\node[vertex] (v0') at (0,-1.5) [label=below:$v_0'$] {};

\draw[postaction={decorate,decoration={markings,mark=at position 0.5 with {\arrow{stealth}}}}] (s) -- (v0) node[midway, above] {$M(3|F| + 4|F|\alpha)$};

\draw[postaction={decorate,decoration={markings,mark=at position 0.5 with {\arrow{stealth}}}}] (v0) -- (u1) node[midway, above, yshift=7pt] {$M(3+4\alpha)$};

\draw[postaction={decorate,decoration={markings,mark=at position 0.5 with {\arrow{stealth}}}}] (v0) -- (u2);
\draw[postaction={decorate,decoration={markings,mark=at position 0.5 with {\arrow{stealth}}}}] (v0) -- (u4);
\draw[postaction={decorate,decoration={markings,mark=at position 0.5 with {\arrow{stealth}}}}] (v0) -- (u5) node[pos=0.25, below] {$A_1$};

\draw[postaction={decorate,decoration={markings,mark=at position 0.5 with {\arrow{stealth}}}}] (u1) -- (v1) node[midway, above] {$M(1+\alpha)$};
\draw[postaction={decorate,decoration={markings,mark=at position 0.5 with {\arrow{stealth}}}}] (u1) -- (v2);

\draw[postaction={decorate,decoration={markings,mark=at position 0.5 with {\arrow{stealth}}}}] (u2) -- (v1);
\draw[postaction={decorate,decoration={markings,mark=at position 0.5 with {\arrow{stealth}}}}] (u2) -- (v2);
\draw[postaction={decorate,decoration={markings,mark=at position 0.5 with {\arrow{stealth}}}}] (u2) -- (3,0.25);

\draw[postaction={decorate,decoration={markings,mark=at position 0.5 with {\arrow{stealth}}}}] (u4) -- (3,-0.25);
\draw[postaction={decorate,decoration={markings,mark=at position 0.5 with {\arrow{stealth}}}}] (u4) -- (v4);
\draw[postaction={decorate,decoration={markings,mark=at position 0.5 with {\arrow{stealth}}}}] (u4) -- (v5);

\draw[postaction={decorate,decoration={markings,mark=at position 0.5 with {\arrow{stealth}}}}] (u5) -- (v4);
\draw[postaction={decorate,decoration={markings,mark=at position 0.5 with {\arrow{stealth}}}}] (u5) -- (v5) node[midway, below] {$A_2$};

\draw[postaction={decorate,decoration={markings,mark=at position 0.5 with {\arrow{stealth}}}}] (v1) -- (t) node[midway, above, yshift=10pt, xshift=10pt] {$M(1+a_j\alpha)$};
\draw[postaction={decorate,decoration={markings,mark=at position 0.5 with {\arrow{stealth}}}}] (v2) -- (t);
\draw[postaction={decorate,decoration={markings,mark=at position 0.5 with {\arrow{stealth}}}}] (v4) -- (t);
\draw[postaction={decorate,decoration={markings,mark=at position 0.5 with {\arrow{stealth}}}}] (v5) -- (t) node[midway, below] {$A_3$};
\draw[postaction={decorate,decoration={markings,mark=at position 0.5 with {\arrow{stealth}}}}] (v0') .. controls (2,-2) and (5,-2) .. (t) node[midway, below, yshift=-5pt] {$M(3|F|-3q+\alpha|F|)$};

\draw[postaction={decorate,decoration={markings,mark=at position 0.5 with {\arrow{stealth}}}}] (u1) -- (v0') node[pos=0.9, above, yshift=10pt, xshift=-12pt] {$M(3+\alpha)$};

\draw[postaction={decorate,decoration={markings,mark=at position 0.5 with {\arrow{stealth}}}}] (u2) -- (v0');
\draw[postaction={decorate,decoration={markings,mark=at position 0.5 with {\arrow{stealth}}}}] (u4) -- (v0');
\draw[postaction={decorate,decoration={markings,mark=at position 0.5 with {\arrow{stealth}}}}] (u5) -- (v0') node[midway, below] {$B$};

\end{tikzpicture}
    \caption{Network $G'=(V',E')$, where $M=3|F|(q+1)$}
    \label{fig:flow-hardness-multiplied}
\end{figure}

\begin{itemize}
    \item $c'(e) = c_\alpha(e) \cdot (3|F|(q+1)), \forall e\in E'$
    \item $p'(e) = p_\alpha(e) \cdot (3|F|(q+1)), \forall e\in E'$
\end{itemize}

\begin{lemma}
    $p'$ is not in the core of $(G',c')$ if and only if $p_\alpha$ is not in the core of $(G_\alpha,c_\alpha)$
\end{lemma}
\begin{proof}
    The worth and the profit-share to any subcoalition are multiplied by the same factor and so, a subcoalition would satisfy the core constraint if and only if it was satisfying the constraint before. 
\end{proof}
Now note that, in $p'$ each vertex gets a positive integer amount of cost-share. Create a new graph $G^*=(V^*,E^*)$ as follows. For each edge $e\in E'$ with a profit-share $p'(e)$, have $p'(e)$ edges, called ``split'' edges of $e$, in $E^*$ along a path between same vertices as before. See figure~\ref{fig:flow-hardness-split}. Note that, since $p'(e) = 1, \forall e \in (A_1\cup A_2\cup B)$, they are still single edges in the graph. The capacity of edges, $c^*$, in $E^*$ will remain the same as the corresponding edge in $E'$. Note that there still are only $4|F|$ paths in $E'$ and the max-flow value has not changed. 

\begin{figure}[H]
    \centering
        \begin{tikzpicture}[scale=2.2, every node/.append style={scale=0.9}]
    \tikzstyle{vertex}=[circle, fill=black, inner sep=1.5pt]

\node (s1) at (2.5,-2.5) [label=above:{$i=\{1,2,\cdots,|F|\}, j=\{1,2,\cdots,3q\}$}] {};

\node[vertex] (s) at (-1,0) [label=left:$s$] {};
\node[vertex] (v0) at (0,0) [label=below:$v_0$] {};
\node (a0) at (-0.5,0) [label=below:$a_0$] {};
\node[vertex] (u1) at (2,1) [label=above:$u_i$] {};
\node[vertex] (u2) at (2,0.5) {};
\node (dots1) at (2,0) [label=right:$\vdots$] {};
\node[vertex] (u4) at (2,-0.5) {};
\node[vertex] (u5) at (2,-1) {};
\node[vertex] (v1) at (4,1) [label=above:$v_j$] {};
\node[vertex] (v2) at (4,0.5) {};
\node (dots2) at (4,0) [label=right:$\vdots$] {};
\node[vertex] (v4) at (4,-0.5) {};
\node[vertex] (v5) at (4,-1) {};
\node[vertex] (t) at (5,0) [label=right:$t$] {};
\node[vertex] (v0') at (0,-1.5) [label=below:$v_0'$] {};


    \node (i1) at (-0.6,0) {};
    \node (i2) at (-0.4,0) {};

    \draw[postaction={decorate,decoration={markings,mark=at position 0.5 with {\arrow{stealth}}}}] (s) -- (i1) node[midway, above, xshift=15pt] {$M(3|F| + 4|F|\alpha)$};
    \node at (-0.5,0) {$\cdots$};
    \draw[postaction={decorate,decoration={markings,mark=at position 0.5 with {\arrow{stealth}}}}] (i2) -- (v0);

\draw[postaction={decorate,decoration={markings,mark=at position 0.5 with {\arrow{stealth}}}}] (v0) -- (u1) node[midway, above, yshift=7pt] {$M(3+4\alpha)$};

\draw[postaction={decorate,decoration={markings,mark=at position 0.5 with {\arrow{stealth}}}}] (v0) -- (u2);
\draw[postaction={decorate,decoration={markings,mark=at position 0.5 with {\arrow{stealth}}}}] (v0) -- (u4);
\draw[postaction={decorate,decoration={markings,mark=at position 0.5 with {\arrow{stealth}}}}] (v0) -- (u5) node[pos=0.25, below] {$A_1$};

\draw[postaction={decorate,decoration={markings,mark=at position 0.5 with {\arrow{stealth}}}}] (u1) -- (v1) node[midway, above] {$M(1+\alpha)$};
\draw[postaction={decorate,decoration={markings,mark=at position 0.5 with {\arrow{stealth}}}}] (u1) -- (v2);

\draw[postaction={decorate,decoration={markings,mark=at position 0.5 with {\arrow{stealth}}}}] (u2) -- (v1);
\draw[postaction={decorate,decoration={markings,mark=at position 0.5 with {\arrow{stealth}}}}] (u2) -- (v2);
\draw[postaction={decorate,decoration={markings,mark=at position 0.5 with {\arrow{stealth}}}}] (u2) -- (3,0.25);

\draw[postaction={decorate,decoration={markings,mark=at position 0.5 with {\arrow{stealth}}}}] (u4) -- (3,-0.25);
\draw[postaction={decorate,decoration={markings,mark=at position 0.5 with {\arrow{stealth}}}}] (u4) -- (v4);
\draw[postaction={decorate,decoration={markings,mark=at position 0.5 with {\arrow{stealth}}}}] (u4) -- (v5);

\draw[postaction={decorate,decoration={markings,mark=at position 0.5 with {\arrow{stealth}}}}] (u5) -- (v4);
\draw[postaction={decorate,decoration={markings,mark=at position 0.5 with {\arrow{stealth}}}}] (u5) -- (v5) node[midway, below] {$A_2$};

    \node (i3) at (4.3,0.7) {};
    \node (i4) at (4.7,0.3) {};

    \draw[postaction={decorate,decoration={markings,mark=at position 0.5 with {\arrow{stealth}}}}] (v1) -- (i3) node[midway, above, xshift=55pt,yshift=-25pt] {$M(1+a_j\alpha)$};
    \node[rotate=-45] at (4.5,0.5) {$\cdots$};
    \draw[postaction={decorate,decoration={markings,mark=at position 0.5 with {\arrow{stealth}}}}] (i4) -- (t);

    \node (i5) at (4.3,0.35) {};
    \node (i6) at (4.7,0.15) {};

    \draw[postaction={decorate,decoration={markings,mark=at position 0.5 with {\arrow{stealth}}}}] (v2) -- (i5);
    \node[rotate=-35] at (4.5,0.25) {$\cdots$};
    \draw[postaction={decorate,decoration={markings,mark=at position 0.5 with {\arrow{stealth}}}}] (i6) -- (t);

    \node (i7) at (4.3,-0.35) {};
    \node (i8) at (4.7,-0.15) {};

    \draw[postaction={decorate,decoration={markings,mark=at position 0.5 with {\arrow{stealth}}}}] (v4) -- (i7);
    \node[rotate=35] at (4.5,-0.25) {$\cdots$};
    \draw[postaction={decorate,decoration={markings,mark=at position 0.5 with {\arrow{stealth}}}}] (i8) -- (t);

    \node (i9) at (4.3,-0.7) {};
    \node (i10) at (4.7,-0.3) {};

    \draw[postaction={decorate,decoration={markings,mark=at position 0.5 with {\arrow{stealth}}}}] (v5) -- (i9) node[midway, below, xshift=10pt,yshift=10pt] {$A_3$};;
    \node[rotate=45] at (4.5,-0.5) {$\cdots$};
    \draw[postaction={decorate,decoration={markings,mark=at position 0.5 with {\arrow{stealth}}}}] (i10) -- (t);

    \node (i11) at (3.5,-1.5) {};
    \node (i12) at (4,-1.5) {};

    \draw[postaction={decorate,decoration={markings,mark=at position 0.5 with {\arrow{stealth}}}}] (v0') -- (i11) node[midway, below] {$M(3|F|-3q+\alpha|F|)$};
    \node at (3.75,-1.5) {$\cdots$};
    \draw[postaction={decorate,decoration={markings,mark=at position 0.5 with {\arrow{stealth}}}}] (i12) -- (t);

\draw[postaction={decorate,decoration={markings,mark=at position 0.5 with {\arrow{stealth}}}}] (u1) -- (v0') node[pos=0.9, above, yshift=10pt, xshift=-12pt] {$M(3+\alpha)$};

\draw[postaction={decorate,decoration={markings,mark=at position 0.5 with {\arrow{stealth}}}}] (u2) -- (v0');
\draw[postaction={decorate,decoration={markings,mark=at position 0.5 with {\arrow{stealth}}}}] (u4) -- (v0');
\draw[postaction={decorate,decoration={markings,mark=at position 0.5 with {\arrow{stealth}}}}] (u5) -- (v0') node[midway, below] {$B$};

\end{tikzpicture}
    \caption{Network $G^*=(V^*,E^*)$, where $M=3|F|(q+1)$}
    \label{fig:flow-hardness-split}
\end{figure}

Now consider the profit-share where each edge gets unit profit, i.e., $p^*(e)=1, \forall e\in E^*$. Note that this imputation both minimizes the maximum and maximizes the minimum cost-share among vertices and also it is both the leximin and leximax imputation in $(G^*,c^*)$. The following lemma characterizes if it is in the core of the game.

\begin{lemma}
\label{lem:flow-hardness-final}
    $p^*$ is not in the core of $(G^*,c^*)$ if and only if $p'$ is not in the core of $(G',c')$.
\end{lemma}

\begin{proof}

    Assume $p'$ is not in the core of $(G',c')$ as $p'(S)<v'(S)$ for some set $S\subset E$. Let $S^*$ be the union of all corresponding split edges of every edge in $S$. Then $v^*(S^*) = v'(S)$ as capacities of edges in $E^*$ is the same as those in $E'$ and so, the same paths will be chosen in both-max-flows. Also note that $p^*(S^*) = p'(S)$ as every edge $e$ in $S$ has exactly $p'(v)$ split edges in $S^*$. Therefore $p'(S') < v'(S') \implies p^*(S^*) < v^*(S^*)$. This shows that if $p'$ is not in the core of $(G',c')$ then $p^*$ is not in the core of $(G^*,c^*)$.

    Now assume $p^*$ is not in the core of $(G^*,c^*)$ as $p^*(S)<v^*(S)$ for some set $S\subset E$. Note that, if $S$ has one split edge of some edge $e$ in $E'$, it can be assumed that it contains all the split edges of $e$. Else, removing such an edge would not change the value of the subcoalition  while only decreasing( or not increasing) the profit-share of the sub-coalition. And so, the new subcoalition would still have $p^*(S)<v^*(S)$. So, consider such a subcoalition $S^*$ and in the graph $G'$, consider the set of edges $S'$, for which all their split edges are included in $S^*$. 
    Note that the both the value and profit-shares of both coalitions are the same meaning $ p^*(S^*) < v^*(S^*) \implies p'(S') < v'(S') $. This proves the other implication of the lemma.
    
\end{proof}

    Note that the construction of this graph would take polynomial time as every edge had a profit that is polynomial in terms of $q$ and $|F|$ because of the chosen value of $\alpha$ - any other value that would make the profits integral and keep the graph size polynomial in $q$ and $|F|$ would work. Together with the lemmas stated above, this completes the proof of theorem~\ref{thm:flow-leximin-hardness}. 
    
    Since the final imputation $p^*$ assigns equal profits to every agent, this imputation has the highest maximin value, lowest minimax value, and it is the leximin and leximax among all possible imputations. And since testing its validity in the core is NP-hard, we get the following corollary. 

\begin{corollary}
    Finding the leximin/leximax core imputation in a max-flow game is NP-hard.
\end{corollary}

\end{proof}

\subsection{Proofs from Section~\ref{sec:branching_game}}

\mstleximinhardness*
\begin{proof}[Proof of Theorem~\ref{thm:mst-leximin-hardness}]

The theorem states that finding the maximin profit-share, which is equal to the minimum profit-share of an agent in the leximin imputation, is in itself an NP-hard problem. We prove this theorem via a sequence of lemmas. We will show that this is true for the MST game which implies the hardness for the min-cost branching game. We use the co-NP hardness result of \cite{faigle1997complexity}. They show that given an instance of MST game and a cost-share, it is NP-hard to find a subcoalition that violates the core constraint, using a reduction from EXACT COVER BY 3-SETS(X3C). 

Given a graph $G$ with root vertex $r$ and an imputation $s$, let us extend the definitions of a cost-share $s$ of vertices to sets of vertices as $s(S) = \sum_{v\in S} s(v)$, the total cost-share of vertices. The characteristic function $v(S) = \sum_{\text{edge }e \text{ in } \text{MST of } S\cup \{r\}} c(e)$, i.e., as the worth of a set of vertices given by the minimum cost spanning tree of that vertices with the root.

\begin{lemma} (\cite{faigle1997complexity} Theorem 2.1)\\
\label{thm:MST-hardness}
    Given an MST game on a graph $G=(V,E), c: E\rightarrow\mathbb{R}_+$, a root vertex $r$ and an imputation $s:V\rightarrow\mathbb{R}_+$, deciding if there exists a subcoalition $S\subset V$ such that $s(S)<v(S)$ is NP-complete.
\end{lemma}

In \cite{faigle1997complexity}, for a specific minimum spanning tree game $(G,c)$ and an imputation $s$ constructed based on an instance of EXACT COVER BY 3-SETS(X3C), it is shown that $s$ is not in the core of $(G,c)$, i.e., there exists a subcoalition $S\subset V$ such that $s(S)<v(S)$ if and only if $X$ has an exact cover in $F$. This shows the hardness of testing whether an imputation is in the core. Like in the case of max-flow game, we extend this hardness result by iteratively transforming the graph, the edge costs and the imputation maintaining the hardness. The final imputation is such that every vertex gets the same share, thus equal to maximin profit share of the vertices. We will show that this imputation is in the core if and only if an exact cover exists. The result is  stared below and the complete proof is presented in the Appendix~\ref{app:proofs}.

\begin{figure}[H]
    \centering
    \begin{tikzpicture}
\node (L4-1) at (-4,4) [circle, fill, inner sep=2pt] {};
\node (L4-2) at (-3,4) [circle, fill, inner sep=2pt] {};
\node (L4-3) at (-2,4) [circle, fill, inner sep=2pt] {};
\node (L4-4) at (-1,4) [circle, fill, inner sep=2pt] {};
\node (L4-5) at (1,4) [circle, fill, inner sep=2pt] {};
\node (L4-6) at (2,4) [circle, fill, inner sep=2pt] {};
\node (L4-7) at (3,4) [circle, fill, inner sep=2pt] {};
\node (L4-8) at (4,4) [circle, fill, inner sep=2pt] {};
\node (L4-9) at (5,4) [label=right:$3q \text{ elements}$] {};

\node at (0,4) {$\cdots$};

\node (L3-1) at (-3.5,2) [circle, fill, inner sep=2pt] {};
\node (L3-2) at (-2.3,2) [circle, fill, inner sep=2pt] {};
\node (L3-3) at (-1.5,2) [circle, fill, inner sep=2pt] {};
\node (L3-4) at (1.5,2) [circle, fill, inner sep=2pt] {};
\node (L3-5) at (2.5,2) [circle, fill, inner sep=2pt] {};
\node (L3-6) at (3.5,2) [circle, fill, inner sep=2pt] {};
\node (L3-7) at (4.5,2) [label=right:$|F| sets$] {};

\node at (0,2) {$\cdots$};

\node (st) at (0,0) [circle, fill, inner sep=2pt, label=right:$St$] {};

\node (g) at (0,-2) [circle, fill, inner sep=2pt, label=right:$g$] {};

\node (r) at (0,-3) [circle, fill, inner sep=2pt, label=right:$r$] {};

\draw (L4-1) -- (L3-1);
\draw (L4-1) -- (L3-2);
\draw (L4-2) -- (L3-1);
\draw (L4-2) -- (L3-2);
\draw (L4-2) -- (L3-3);
\draw (L4-3) -- (L3-2);
\draw (L4-4) -- (L3-1);
\draw (-0.75,3) -- (L3-3);

\draw (L4-5) -- (L3-4);
\draw (L4-5) -- (L3-5);
\draw (L3-4) -- (0.75,3);
\draw (L4-6) -- (L3-5);
\draw (L4-6) -- (L3-6);
\draw (L4-7) -- (L3-4);
\draw (L4-7) -- (L3-6);
\draw (L4-8) -- (L3-5);
\draw (L4-8) -- (L3-6) node[midway, right] {$q+1$};

\foreach \i in {1,2,3,4,5} {
    \draw (L3-\i) -- (st);
}
\draw (L3-6) -- (st) node[midway,right, xshift = 9pt,yshift = 3pt] {$q$};

\foreach \i in {1,2,3,4,5} {
    \draw (L3-\i) -- (g);
}
\draw (L3-6) -- (g) node[midway, right] {$q+1$};

\draw (st) -- (g) node[midway,yshift = 13pt, xshift = 3pt] {$q+1$};
\draw (g) -- (r) node[midway,right] {$2q-1$};

    \end{tikzpicture}
    \caption{Graph $G=(V,E)$}
    \label{fig:mst-hardness-orig}
\end{figure}

Consider the following instance of the EXACT COVER BY 3-SETS(X3C) problem - a finite set with $X = \{x_1,..., x_{3q} \}$ and a collection $F = {f_1 ..... f_{|F|}}$ of 3-element subsets of $X$. The MST game they considered in \cite{faigle1997complexity}(and lemma~\ref{lem:flow-hardness}) is as follows(some notation has been changed to be consistent with the definition above). The graph $G=(V,E)$ has $|X|+|F|+3$ vertices - $V = \{1,2,\ldots ,3q\} \cup \{ 3q+1,3q+2, \ldots, 3q+|F|\} \cup \{0,g,St\}$. The first $3q$ correspond to the element-players, the next $F$ to the set-players in the collection $F$ and the last three correspond to the root vertex(0), a guardian-player($g$) and a Steiner-vertex-player$(St)$ respectively. The vertices are arranged in layers as shown in figure~\ref{fig:mst-hardness-orig}. The edges and their costs are defined as follows. For each set $f_i=\{j,k,l\}, i\in \{1,2,\ldots,|F|\}$,
\begin{itemize}
    \item $c(3q+i,j) = c(3q+i,k) = c(3q+i,l) = q+1$
    \item $c(3q+i,St) = q$
    \item $c(3q+i,g) = q+1$
    \item $c(St,g) = q+1$
    \item $c(g,0) = 2q-1$
\end{itemize}

Given this, consider the imputation $s\in \mathbb{R}_+^{|X|+|F|+2}$ defined as
\begin{itemize}
    \item $s_i = q+2, \forall i\in {1,2, \ldots, 3q}$
    \item $s_i = q, \forall i \in {3q+1,3q+2, \ldots, 3q+|F|}$
    \item $s_{St} = s_{g} = 0$
\end{itemize}

Consider the modification of $G=(V,E)$ and $s$, $G'=(V',E')$ and $s'$, where the vertices and edges are the same - $V'=V,E'=E$ - but every edge in $G'$ has a cost $c'$ that is one unit more than in $G$ and $s'$ gives every node one unit more cost than in $s$ (see figure~\ref{fig:mst-hardness-one-added}). That is, 

\begin{itemize}
    \item $c'(e) = c(e)+1, \forall e\in E$
    \item $s'_v = s_v + 1, \forall v\in V$
\end{itemize}

\begin{figure}[H]
    \centering
    \begin{tikzpicture}
\node (L4-1) at (-4,4) [circle, fill, inner sep=2pt] {};
\node (L4-2) at (-3,4) [circle, fill, inner sep=2pt] {};
\node (L4-3) at (-2,4) [circle, fill, inner sep=2pt] {};
\node (L4-4) at (-1,4) [circle, fill, inner sep=2pt] {};
\node (L4-5) at (1,4) [circle, fill, inner sep=2pt] {};
\node (L4-6) at (2,4) [circle, fill, inner sep=2pt] {};
\node (L4-7) at (3,4) [circle, fill, inner sep=2pt] {};
\node (L4-8) at (4,4) [circle, fill, inner sep=2pt] {};
\node (L4-9) at (5,4) [label=right:$3q \text{ elements}$] {};

\node at (0,4) {$\cdots$};

\node (L3-1) at (-3.5,2) [circle, fill, inner sep=2pt] {};
\node (L3-2) at (-2.3,2) [circle, fill, inner sep=2pt] {};
\node (L3-3) at (-1.5,2) [circle, fill, inner sep=2pt] {};
\node (L3-4) at (1.5,2) [circle, fill, inner sep=2pt] {};
\node (L3-5) at (2.5,2) [circle, fill, inner sep=2pt] {};
\node (L3-6) at (3.5,2) [circle, fill, inner sep=2pt] {};
\node (L3-7) at (4.5,2) [label=right:$|F| sets$] {};

\node at (0,2) {$\cdots$};

\node (st) at (0,0) [circle, fill, inner sep=2pt, label=right:$St$] {};

\node (g) at (0,-2) [circle, fill, inner sep=2pt, label=right:$g$] {};

\node (r) at (0,-3) [circle, fill, inner sep=2pt, label=right:$r$] {};

\draw (L4-1) -- (L3-1);
\draw (L4-1) -- (L3-2);
\draw (L4-2) -- (L3-1);
\draw (L4-2) -- (L3-2);
\draw (L4-2) -- (L3-3);
\draw (L4-3) -- (L3-2);
\draw (L4-4) -- (L3-1);
\draw (-0.75,3) -- (L3-3);

\draw (L4-5) -- (L3-4);
\draw (L4-5) -- (L3-5);
\draw (L3-4) -- (0.75,3);
\draw (L4-6) -- (L3-5);
\draw (L4-6) -- (L3-6);
\draw (L4-7) -- (L3-4);
\draw (L4-7) -- (L3-6);
\draw (L4-8) -- (L3-5);
\draw (L4-8) -- (L3-6) node[midway, right] {$q+2$};

\foreach \i in {1,2,3,4,5} {
    \draw (L3-\i) -- (st);
}
\draw (L3-6) -- (st) node[midway,right, xshift = 9pt,yshift = 3pt] {$q+1$};

\foreach \i in {1,2,3,4,5} {
    \draw (L3-\i) -- (g);
}
\draw (L3-6) -- (g) node[midway, right] {$q+2$};

\draw (st) -- (g) node[midway,yshift = 13pt, xshift = 3pt] {$q+2$};
\draw (g) -- (r) node[midway,right] {$2q$};

    \end{tikzpicture}
    \caption{Graph $G'=(V',E')$}
    \label{fig:mst-hardness-one-added}
\end{figure}

Let $v \text{ and }v'$ be the corresponding characteristic functions of the MST games on $G$ and $G'$. Then

\begin{lemma}
    $s'$ is not in the core of $(G',v')$ if and only if $s$ is not in the core of $(G,v)$
\end{lemma}

\begin{proof}
    For any subcoalition $S\subset V\setminus \{0\}$, note that $s'(S) = s(S)+ |S|$ as every node in $S$ gets one more unit of cost. Also note that, $v'(S) = v(S)+|S|$ because an MST connecting all nodes in $S$ to node 0 would need $|S|$ edges and each edges costs just a unit more. And so, $s'(S) > v'(S) \iff s(S) > v(S)$, proving the above statement. 
\end{proof}

Now note that, in $s'$ each vertex gets a positive integer amount of cost-share. Create a new graph $G^*=(V^*,E^*)$ as follows. For each vertex $v\in V'$ with a cost-share $s'(v)$, have $s'(v)$ vertices in $V^*$. Of these, one will be called a representative vertex and the remaining will be called duplicates of the corresponding vertex. The edges, $E^*$, in $G^*$ will be of two kinds. One set, $E^*_1$, will only contain edges between the representative vertices and they will look the same as in $G$ or $G'$. The other set of edges, $E^*_2$, will only exist between the representative vertex and the corresponding duplicates - that is, zoomed in on these vertices, the subgraph will look like a tree of height one with representative vertex as the root. See figure~\ref{fig:mst-hardness-split}. Note that, since $s'(g)=s'(St) = 1$, they do not have any duplicates in the graph. The costs of edges, $c^*$, in $E^*_1$ remain the same as in $E'$ but the costs of edges in $E^*_2$ is zero as shown below.

\begin{itemize}
    \item $c^*(e) = c'(e), \forall e\in E^*_1$
    \item $c^*(e) = 0, \forall e\in E^*_2$
\end{itemize}

Let $v^*$ be the corresponding characteristic functions of the MST game on $G^*$. Now consider the cost-share where each vertex gets unit cost, i.e., $s^*(v)=1, \forall v\in V^*$. Note that this imputation both minimizes the maximum and maximizes the minimum cost-share among vertices as it is both the leximin and leximax imputation in $(G^*,v^*)$. The following lemma characterizes if it is in the core of the game.

\begin{figure}[H]
    \centering
    \begin{tikzpicture}
\node (L4-1) at (-4,4) [circle, fill, inner sep=2pt] {};
\node (L4-2) at (-3,4) [circle, fill, inner sep=2pt] {};
\node (L4-3) at (-2,4) [circle, fill, inner sep=2pt] {};
\node (L4-4) at (-1,4) [circle, fill, inner sep=2pt] {};
\node (L4-5) at (1,4) [circle, fill, inner sep=2pt] {};
\node (L4-6) at (2,4) [circle, fill, inner sep=2pt] {};
\node (L4-7) at (3,4) [circle, fill, inner sep=2pt] {};
\node (L4-8) at (4,4) [circle, fill, inner sep=2pt] {};
\node (L4-9) at (5,4) [label=right:$3q \text{ elements}$] {};
\node (L4-9) at (5,4.5) [label=right:$3q(q+2) \text{ copies of elements}$] {};

\foreach \i in {1,2,3,4,5,6,7,8} {
    \node (A\i-1) at ($(L4-\i) + (-0.35,0.5)$) [circle, fill, draw=blue, inner sep=0.5pt] {};
    \node (A\i-3) at ($(L4-\i) + (0,0.5)$) {$\textcolor{blue}{..}$};
    \node (A\i-2) at ($(L4-\i) + (0.35,0.5)$) [circle, fill, draw=blue, inner sep=0.5pt] {};
    \node (A\i-4) at ($(L4-\i) + (0.2,0.5)$) [circle, fill, draw=blue, inner sep=0.5pt] {};
    \node (A\i-5) at ($(L4-\i) + (-0.2,0.5)$) [circle, fill, draw=blue, inner sep=0.5pt] {};
    \draw[draw=blue] (A\i-1) -- (L4-\i);
    \draw[draw=blue] (A\i-2) -- (L4-\i);
    \draw[draw=blue] (A\i-4) -- (L4-\i);
    \draw[draw=blue] (A\i-5) -- (L4-\i);
}

\foreach \i in {1,2,3,4,5,6} {
    \node (B\i-1) at ($(L3-\i) + (-0.35,0.5)$) [circle, fill, draw=red, inner sep=0.5pt] {};
    \node (B\i-3) at ($(L3-\i) + (0,0.5)$) {$\textcolor{red}{..}$};
    \node (B\i-2) at ($(L3-\i) + (0.35,0.5)$) [circle, fill, draw=red, inner sep=0.5pt] {};
    \node (B\i-4) at ($(L3-\i) + (0.2,0.5)$) [circle, fill, draw=red, inner sep=0.5pt] {};
    \node (B\i-5) at ($(L3-\i) + (-0.2,0.5)$) [circle, fill, draw=red, inner sep=0.5pt] {};
    \draw[red] (B\i-1) -- (L3-\i);
    \draw[draw=red] (B\i-2) -- (L3-\i);
    \draw[draw=red] (B\i-4) -- (L3-\i);
    \draw[draw=red] (B\i-5) -- (L3-\i);
}

\node at (0,4) {$\cdots$};

\node (L3-1) at (-3.5,2) [circle, fill, inner sep=2pt] {};
\node (L3-2) at (-2.3,2) [circle, fill, inner sep=2pt] {};
\node (L3-3) at (-1.5,2) [circle, fill, inner sep=2pt] {};
\node (L3-4) at (1.5,2) [circle, fill, inner sep=2pt] {};
\node (L3-5) at (2.5,2) [circle, fill, inner sep=2pt] {};
\node (L3-6) at (3.5,2) [circle, fill, inner sep=2pt] {};
\node (L3-7) at (4.5,2) [label=right:$|F| sets$] {};
\node (L3-7) at (4.5,2.5) [label=right:$q|F| \text{ copies of sets}$] {};

\node at (0,2) {$\cdots$};

\node (st) at (0,0) [circle, fill, inner sep=2pt, label=right:$St$] {};

\node (g) at (0,-2) [circle, fill, inner sep=2pt, label=right:$g$] {};

\node (r) at (0,-3) [circle, fill, inner sep=2pt, label=right:$r$] {};

\draw (L4-1) -- (L3-1);
\draw (L4-1) -- (L3-2);
\draw (L4-2) -- (L3-1);
\draw (L4-2) -- (L3-2);
\draw (L4-2) -- (L3-3);
\draw (L4-3) -- (L3-2);
\draw (L4-4) -- (L3-1);
\draw (-0.75,3) -- (L3-3);

\draw (L4-5) -- (L3-4);
\draw (L4-5) -- (L3-5);
\draw (L3-4) -- (0.75,3);
\draw (L4-6) -- (L3-5);
\draw (L4-6) -- (L3-6);
\draw (L4-7) -- (L3-4);
\draw (L4-7) -- (L3-6);
\draw (L4-8) -- (L3-5);
\draw (L4-8) -- (L3-6) node[midway, right] {$q+2$};

\foreach \i in {1,2,3,4,5} {
    \draw (L3-\i) -- (st);
}
\draw (L3-6) -- (st) node[midway,right, xshift = 9pt,yshift = 3pt] {$q+1$};

\foreach \i in {1,2,3,4,5} {
    \draw (L3-\i) -- (g);
}
\draw (L3-6) -- (g) node[midway, right] {$q+2$};

\draw (st) -- (g) node[midway,yshift = 13pt, xshift = 3pt] {$q+2$};
\draw (g) -- (r) node[midway,right] {$2q$};

    \end{tikzpicture}
    \caption{Graph $G^*=(V^*,E^*)$. Blue and red edges have zero cost.}
    \label{fig:mst-hardness-split}
\end{figure}

\begin{lemma}
    $s^*$ is not in the core of $(G^*,v^*)$ if and only if $s'$ is not in the core of $(G',v')$.
\end{lemma}

\begin{proof}
    Assume $s'$ is not in the core of $(G',v')$ as $s'(S)>v'(S)$ for some set $S\subset V\setminus\{0\}$. Let $S^*$ be the union of all representative vertices and duplicate vertices of every vertex in $S$. Then $v^*(S^*) = v'(S)$ as costs of edges in $E^*_1$ is the same as those in $E'$ and cost of edges in $E^*_2$ is zero. Also note that $s^*(S^*) = s'(S)$ as every vertex $v$ in $S$ has exactly $s'(v)$ copies in $S^*$. Therefore $s'(S') > v'(S') \implies s^*(S^*) > v^*(S^*)$. This shows that if $s'$ is not in the core of $(G',v')$ then $s^*$ is not in the core of $(G^*,v^*)$.

    Now assume $s^*$ is not in the core of $(G^*,v^*)$ as $s^*(S)>v^*(S)$ for some set $S\subset V^* \setminus \{0\}$. Note that, if $S$ has a duplicate vertex of some vertex $v\in V$, then it must also contain the representative vertex of $v$, else, $v^*(S)=\inf$ as there's no other edge that can link the vertex to node 0, giving a contradiction. Now consider the set $S^*$ such that, if a representative vertex $v_0$ is in $S$, then all the corresponding duplicate vertices are also chosen in $S^*$. Note that, even with this change, $v^*(S^*)=v^*(S)$ as we are only adding edges in $E^*_2$ which have zero cost. Also, $s^*(S^*)\geq s^*(S)$ as each new vertex will only increase the cost-share. Now, in the graph $G'$, consider the set of vertices, $S'$, whose representative vertices belong to $S^*$. For this set, it is easy to see that $s'(S') = s^*(S^*) \geq s^*(S) > v^*(S) = v^*(S^*) = v'(S')$. This shows the other direction of the lemma and completes the proof.
\end{proof}

    Combining all the lemmas above and the easily-observed fact that each graph construction described above can be done in polynomial time, finding leximin core imputations for the MST game is NP-hard. Since MST games are a special subclass of min-cost branching games, this completes the proof of theorem~\ref{thm:mst-leximin-hardness}.

\begin{corollary}
    Finding the leximin/leximax core imputation in an MST game or a min-cost branching game is NP-hard.
\end{corollary}

\end{proof}

\mstleximincombi*
\begin{proof}[Proof of ~\Cref{thm:mst_leximin_combi}]

Let $U=\set{e_1,e_2,\ldots, e_n}$ be the elements and $S_i, 1\le i\le m$ the sets in a set-cover instance. Let $c_i$ be the cost of set $S_i$. We construct a graph $G=(V,E)$ which has a vertex $u_i$ corresponding to set $S_i, 1\le i\le m$ and a vertex $v_j$ for each element $e_j, 1\le j\le n$. If $e_j\in S_i$ we add edge $(v_j,u_i)$ of zero cost to $E$. 

$G$ has 3 other vertices $a,b,r$ and edges $(u_i,a), (u_i,b), 1\le i\le m$, $(a,b)$ and $(b,r)$. 
Edges $(u_i,a)$ have zero costs, and edges $(u_i,b)$ have cost $c_i$, $1\le i\le m$. Edges $(a,b)$ and $(b,r)$ have cost $c_0=n\cdot\max_i c_i$. The vertices $v_j, 1\le j\le n$ form the set $T$. The following figure illustrates an example where $S_1=\{v_1,v_3\}$, $S_2 = \{v_1,v_2\}$ and $S_3=\{v_2,v_3\}$.

\begin{figure}[H]
    \centering
    \begin{tikzpicture}

\coordinate (u1) at (-3,-1.5);
\coordinate (u2) at (0,-1.5);
\coordinate (u3) at (3,-1.5);
\coordinate (v1) at (-3,-3);
\coordinate (v2) at (0,-3);
\coordinate (v3) at (3,-3);
\coordinate (a) at (0,0);
\coordinate (b) at (0,1.5);
\coordinate (r) at (0,3);

\draw[->, >=stealth, line width=1pt] (v1) -- (u1) node[midway, left] {0};
\draw[->, >=stealth, line width=1pt] (v1) -- (u2) node[midway, left] {0};
\draw[->, >=stealth, line width=1pt] (v2) -- (u2) node[midway, left] {0};
\draw[->, >=stealth, line width=1pt] (v2) -- (u3) node[midway, left] {0};
\draw[->, >=stealth, line width=1pt] (v3) -- (u3) node[midway, left] {0};
\draw[->, >=stealth, line width=1pt] (v3) -- (u1) node[midway, left] {0};
\draw[->, >=stealth, line width=1pt] (u1) -- (a) node[midway, left] {0};
\draw[->, >=stealth, line width=1pt] (u2) -- (a) node[midway, left] {0};
\draw[->, >=stealth, line width=1pt] (u3) -- (a) node[midway, left] {0};

\draw[->, >=stealth, line width=1pt, blue] (u1) -- (b) node[midway, left] {$c_1$};
\draw[->, >=stealth, line width=1pt, blue] (u2) .. controls (1,1) and (0,1) .. (b) node[pos=0.25, right] {$c_2$};

\draw[->, >=stealth, line width=1pt, blue] (u3) -- (b) node[midway, right] {$c_3$};

\draw[->, >=stealth, line width=1pt, green] (a) -- (b) node[pos=0.25, left] {$c_{ab}$};

\draw[->, >=stealth, line width=1pt, red] (b) -- (r) node[midway, left] {$c_0$};

\filldraw [black] (r) circle (2pt) node[above] {$r$};
\filldraw [green] (a) circle (2pt) node[above right] {$a$};
\filldraw [blue] (b) circle (2pt) node[above right] {$b$};
\filldraw [blue] (u1) circle (2pt) node[below left] {$u_1$};
\filldraw [blue] (u2) circle (2pt) node[below right] {$u_2$};
\filldraw [blue] (u3) circle (2pt) node[below right] {$u_3$};
\filldraw [red] (v1) circle (2pt) node[below left] {$v_1$};
\filldraw [red] (v2) circle (2pt) node[below right] {$v_2$};
\filldraw [red] (v3) circle (2pt) node[below right] {$v_3$};

\end{tikzpicture}
  \label{fig:setcover}
  \caption{Example for equivalence between max min MST core imputation and optimum fractional set cover}
\end{figure}

\begin{lemma}
\label{lemma:setcover}
Let $\lambda^*$ be the maximin cost-share of a vertex in $T$ and $C^*$ the optimum fractional setcover. Then $C^*=n\lambda^*-c_0$.

\end{lemma}

\begin{proof}
Recall the LP for the fractional set cover problem
\begin{mini}
		{} {\sum_{1\leq i \leq m} c_i x_i}
			{\label{eq.setcover}}
		{}
        \addConstraint{\sum_{i:e_j\in S_i} x_i}{\ge 1}{\quad \forall e_j, 1\leq j\leq n}
        \addConstraint{x_i}{\geq 0}
       {\quad \forall x_i, 1\leq i\leq m} 
	\end{mini} 

and its dual,
\begin{maxi}
		{} {\sum_{1\leq j\leq n} y_j}
			{\label{eq.setcover-dual}}
		{}
        \addConstraint{\sum_{j:e_j \in S_i} y_j}{\le c_i}{\quad\forall S_i, 1\leq i \leq m}
		\addConstraint{y_j}{\geq 0}{\quad\forall e_j, 1\leq j \leq n}
    \end{maxi} 

We first prove that $C^* \leq n\lambda^* - c_0$. Let $y^*$ be an optimum solution to LP (\ref{eq.setcover-dual}). By strong duality $C^*=\sum_{1\le j\le n} y^*_j$. We use $y^*$ to assign a cost-share of at least $(C^*+c_0)/n$ to every vertex in $T$, establishing the inequality.

For vertex $v_j$, let $S'_j$ be the vertices reachable from $v_j$ by a path of cost 0 in the graph $G$. Note that $a\in S'_j$ and if $e_j\in S_i$ then $u_i\in S'_j$. 

Consider assigning values to variables $y(S,v), v\in S\subseteq\Vr$ as follows. For all $1\leq j \leq n$, let $y(S'_j, v_j)= y^*_j$ and let $y(\Vr,v_j)= (C^*+c_0)/n-y^*_j$. Note that for all $1\le j\le n$, $c_0= n\max_i c_i\ge ny^*_j$ and hence $y(\Vr,v_j) \ge 0$. Thus the total cost-share assigned to $v_j\in T$ is $(C^*+c_0)/n$. This assignment of variables satisfies the packing constraints on the edges since
\begin{enumerate}
\item for any edge $e=(u_i,b)$, $\sum_{j:e\in\btd(S'_j)} y(S'_j,v_j) = \sum_{j:e_j\in S_i} y^*_j \le c_i$,
\item for edge $e=(a,b)$, $\sum_{j:e\in\btd(S'_j)} y(S'_j,v_j) = \sum_{j=1}^n y^*_j \le n\max_i c_i = c_0$, and
\item for edge $e=(b,r)$, $\sum_{j=1}^n y(\Vr,v_j)= \sum_{j=1}^n ((C^*+c_0)/n-y^*_j) = C^*+c_0-\sum_{j=1}^n y^*_j = c_0$.
\end{enumerate}

For the converse we need to show that $C^* \geq n\lambda^* - c_0$. Let $y(S,v_j), v_j\in S\subseteq \Vr$ be an assignment to variables that gives every vertex $v_j\in T$ a cost-share of at least $\lambda^*$. Note that if $b\in S$ then $(b,r)\in \btd(S)$ which implies that $\sum_{j=1}^n \sum_{S:b\in S} y(S,v_j) \le c_0$. Since $\sum_{j=1}^n \sum_{S\subseteq\Vr} y(S,v_j) \ge n\lambda^*$ it follows that $\sum_{j=1}^n \sum_{S\subseteq V\setminus\set{r,b}} y(S,v_j) \ge  n\lambda^* - c_0$. 

 We construct a solution, $y^*$, to the dual setcover LP by assigning $y^*_j=\sum_{S\subseteq V\setminus\set{r,b}} y(S,v_j)$. To argue that this is a feasible solution we need to show that for all $i$, $1\le i\le m$, $\sum_{j:e_j\in S_i} y^*_j \le c_i$. Note that, for $S\subseteq V\setminus\set{r,b}$ if $y(S,v_j)> 0$ then $S'_j\subseteq S$. This implies that for all $i: e_j\in S_i$, $(u_i,b)\in\btd(S)$ if $y(S,v_j)> 0$. Thus for $1\le i\le m$,
 \begin{align*}
    \sum_{j:e_j\in S_i} y^*_j = \sum_{j:e_j\in S_i}\sum_{S\subseteq V\setminus\set{r,b}} y(S,v_j) \le c_i
 \end{align*}
where the last inequality follows from the fact that for every set $S$ that contributes a non-zero value to the sum, edge $(u_i,b)\in\btd(S)$ and the variables $y(S,v_j)$ satisfy the packing constraints on the edges. To complete the proof note that $$C^*\ge \sum_{j=1}^n y^*_j = \sum_{j=1}^n \sum_{S\subseteq V\setminus\set{r,b}} y(S,v_j) \ge  n\lambda^* - c_0.$$
\end{proof}

\end{proof}

\subsection{Proofs from Section~\ref{sec:b_matching_game}}

\bmatchhardness*
\begin{proof}[Proof of Theorem~\ref{thm:bmatch-hardness}]

Lemma~\ref{thm:bmatch-hardness} states that given an instance of the $b$-matching game  and an ``approximate'' profit-share, it is NP-hard to find a subcoalition that violates the core constraint, using a reduction from EXACT COVER BY 3-SETS(X3C). Note that in the lemma, $p$ is not an imputation as it distributes strictly more value than worth of the game.

Recall that an instance of the exact cover by 3-sets problem is specified by a set of elements $X=\set{x_1,\ldots,x_{3q}}$ and a collection $F=\set{f_1,\ldots f_k}$ of 3-element subsets of $X$. The problem remains NP-hard for $k=2q$ and given such an instance we construct a bipartite graph $G=(U,V,E)$ and the function $b$ as follows.
 \begin{enumerate}
\item For each element $x_i, 1\le i\le 3q$ we include a vertex $u_i$ in $U$ and for each set $f_j, 1\le j\le k$ we add a vertex $v_j$ to $V$. Two root vertices $r_1,r_2$ are also added to $U$. 
\item If $x_i\in f_j$ we add an edge $(u_i,v_j)$ of unit weight to $E$. We also add edges $(v_j,r_1), (v_j,r_2), 1\le j\le k$. Edges incident to $r_1$ have weight $(1-\epsilon)$ while edges incident to $r_2$ have weight $(1-\epsilon/4)$, where $\epsilon>0$ is a constant to be determined later.
\item Let $b(u_i)=1, 1\le i\le 3q$, $b(v_j)=4, 1\le j\le k$, $b(r_1)=q$, and $b(r_2)=4q$.
\end{enumerate}
\begin{figure}
    \centering
    \includegraphics[scale=0.5]{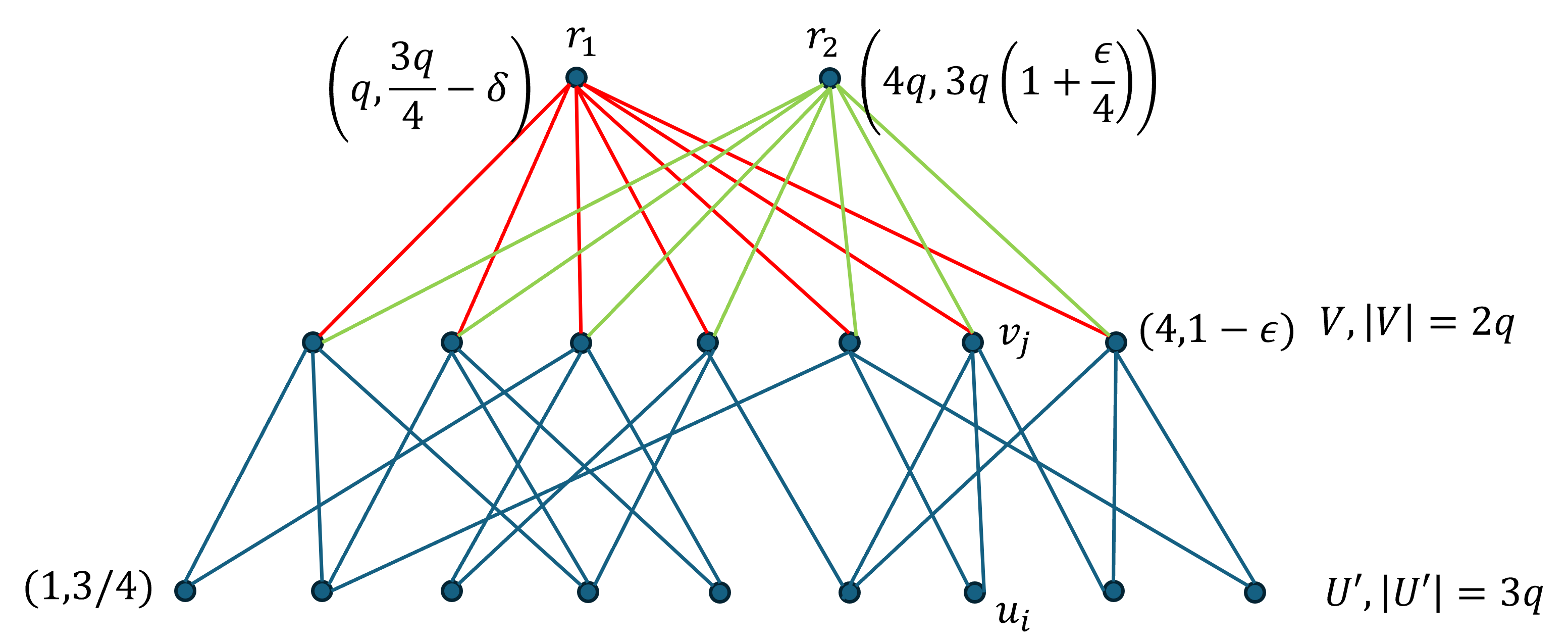}
    \caption{The tuple besides a vertex $v$ is $(b(v),p(v))$. The blue edges have unit weight, red edges have weight $(1-\epsilon)$ and green edges have weight $(1-\epsilon/4)$.}
    \label{fig:b_match_tree}
\end{figure}
Consider the profit-share $p$ defined as $p(u_i)=3/4, 1\le i\le 3q$, $p(v_j)=1-\epsilon, 1\le j\le k$, $p(r_1)=3q/4-\delta$ and $p(r_2)=3q(1+\epsilon/4)$ where $\delta > 0$ is a constant to be determined later. See figure~\ref{fig:b_match_tree} for reference. We will argue that there exists $S\subseteq U\cup V$ such that $p(S)< v(S)$ iff there exists a collection of $q$ sets in $F$ which cover all elements in $X$. 


Let $U'=\set{u_i, 1\le i\le 3q}$.
For the easy direction let $F'\subset F$, $\abs{F'}=q$, be an exact cover of $X$. Let $S=U\setminus\set{r_2}\cup\set{v_j|f_j\in F'}$ and $E'\subseteq E$ be the edges induced over $S$. Define $\deg{E'}{v}$ to be the degree of vertex $v$ in $E'$. Note that for all $v\in U\cup V, \deg{E'}{v}=b(v)$ and hence edges of $E'$ form a maximum weight $b$-matching on $S$. Thus $v(S)=4q - \epsilon$. However, $p(S)=4q-\epsilon-\delta$ and hence $p(S)< v(S)$.

For the converse, we need to show that if there is a set $S\subseteq U\cup V$ such that $p(S)< v(S)$ then the instance has an exact cover. We consider 4 cases determined by which of $r_1,r_2$ are in $S$.
\begin{enumerate}
\item $r_1, r_2\notin S$. Let $\abs{S\cap V}=t, \abs{S\cap U'}=s$. Choose a minimal set $S$ such that $p(S)<v(S)$. Then, every vertex in $|S\cap U'|$ is adjacent to some vertex in $|S\cap V|$ - else it can be removed to get $S'$ satisfying $p(S')<v(S')$. Hence, $s\leq 3t$. 

Note that $v(S)\leq s$ and $p(S)=t(1-\epsilon)+3s/4= t+3s/4-3\epsilon s/4$. For a small $\epsilon$, $p(S)< v(S)$ implies $t+3s/4\le s$. This implies $s\ge 4t$ which is a contradiction since $s\le 3t$.

\item $r_1,r_2\in S$. Let $\abs{S\cap V}=t, \abs{S\cap U'}=s$. Then 
\begin{align*}
    v(S) & =s+4q(1-\epsilon/4)+(4t-s-4q)(1-\epsilon) = 4t-\epsilon(4t-s-3q)\\
    p(S) & =3s/4+t(1-\epsilon)+3q(1+\epsilon/4)+(3q/4-\delta)  = 15q/4+3s/4+t-\epsilon(t-3q/4)-\delta\\
    p(S)-v(S) & = 15q/4+3s/4-3t + \epsilon(3t-9q/4-s)-\delta
\end{align*}
If $p(S)-v(S) <0$, then for a small choice of $\epsilon,\delta$, we must have $15q/4+3s/4-3t\le 0$ which implies $4t\ge (s+5q)$. The number of edges incident to vertices in $S\cap U$ is at most $4t$ and this cannot exceed the $b$ values of vertices in $S\cap U$ which is $s+5q$. This implies that $t=(s+5q)/4$ and $15q/4+3s/4-3t = 0$.

Hence if $p(S)<v(S)$ then for a sufficiently small choice of $\delta$, we must have $3t-9q/4-s\le 0$. Substituting the value of $t$ obtained above we get $3q/2-s/4\le 0$ which is a contradiction
since $s \le 3q$.

\item $r_1\notin S, r_2\in S$. Let $\abs{S\cap V}=t, \abs{S\cap U'}=s$. Then 
\begin{align*}
    v(S) & =s+(4t-s)(1-\epsilon/4) = 4t-\epsilon(t-s/4)\\
    p(S) & =3s/4+t(1-\epsilon)+3q(1+\epsilon/4) = 3q+3s/4+t-\epsilon(t-3q/4)\\
    p(S)-v(S) & = 3q+3s/4-3t + \epsilon(3q/4-s/4)
\end{align*}
If $p(S)-v(S) <0$, then for a small choice of $\epsilon$, we must have $3q+3s/4-3t\le 0$ which implies $4t\ge (s+4q)$. The number of edges incident to vertices in $S\cap U$ is at most $4t$ and this cannot exceed the $b$ values of vertices in $S\cap U$ which is $s+4q$. This implies that $3q+3s/4-3t = 0$.

Hence if $p(S)<v(S)$ then we must have $3q/4-s/4 < 0$. This implies $s> 3q$ which is a contradiction.
\item $r_1\in S, r_2\notin S$. Let $\abs{S\cap V}=t, \abs{S\cap U'}=s$. Then 
\begin{align*}
    v(S) & =s+(4t-s)(1-\epsilon) = 4t-\epsilon(4t-s)\\
    p(S) & =3s/4+t(1-\epsilon)+3q/4-\delta = 3q/4+3s/4+t-\epsilon t -\delta\\
    p(S)-v(S) & = 3q/4+3s/4-3t + \epsilon(3t-s)-\delta
\end{align*}
If $p(S)-v(S) <0$, then for a small choice of $\epsilon$ and $\delta$, we must have $3q/4+3s/4-3t\le 0$ which implies $4t\ge (s+q)$. The number of edges incident to vertices in $S\cap U$ is at most $4t$ and this cannot exceed the $b$ values of vertices in $S\cap U$ which is $s+q$. This implies that $4t=s+q$ and $3q/4+3s/4-3t = 0$.

Hence if $p(S)<v(S)$ then for a sufficiently small $\delta$ we must have $3t-s \le 0$. By substituting the value of $t$ obtained above, we get $s\ge 3q$. Since $s\le 3q$ it must be the case that $s=3q$ and hence $t=q$. Hence vertices in $S\cap V$ are an exact cover of the vertices in $S\cap U'$. 
\end{enumerate}

Note that the total profit-share of all vertices $p(U\cup V)=3q(3/4)+2q(1-\epsilon)+3q(1+\epsilon/4)+3q/4-\delta$. The maximum $b$-matching in this instance matches all vertices maximally and has size $8q$. The weight of the maximum $b$-matching is $v(U\cup V)=3q+(1-\epsilon/4)4q+(1-\epsilon)q = 8q-2q\epsilon$. Thus, $p(U\cup V)-v(U\cup V)= 3q\epsilon/4-\delta > 0$ since $\delta$ is much smaller than $\epsilon$. Hence the profit-shares we have assigned do not constitute an imputation.

Since $\alpha\ge p(U\cup V)/v(U\cup V) = 1+ 3\epsilon/(32-8\epsilon)$, we conclude that $\epsilon$ should be at most $32(\alpha-1)/(8\alpha-5)\ge 4(1-1/\alpha)$. Another bound on $\epsilon$ is obtained from the above case analysis and it is easy to check that $\epsilon=1/6q$ suffices to carry the argument through. Hence we set $\epsilon=\min\{4(1-1/\alpha),1/6q\}$.    
\end{proof}

\section{Finding leximin Owen set imputation of the max-flow game using linear programming}
\label{app:flow_leximin_lp}

The linear program~\ref{eq.flow-dual} has polynomially many variables and constraints, it can be solved efficiently. In this section, we will show how to modify this LP so that we can use the techniques of section~\ref{sec:lp_leximin_solution} and compute the leximin Owen set imputation. \\
Note that in this setting, the profits are only assigned to the edges, and each edge $(i,j)$ receives a profit of $p(i,j) = c_{ij} \delta_{ij}$. Therefore at every iteration we find the max-min profit that should be assigned, and we assign it to a set of edges. The following set of LPs return the leximin Owen set imputations for profits in max-flow game. $V_F$ can contain only the edges, and $V_F = \emptyset$ initially. The mapping $m: E \to \mathcal{R_+}$ is defined over the set of all edges. The procedure stops when all edge profits are fixed. 

 	\begin{maxi}
		{} {\alpha}  
			{\label{eq.flow-dual-method}}
		{}
        \addConstraint{\sum_{(i, j) \in E}  {c_{ij} \delta_{ij}}}{= \opt}{}
		\addConstraint{\delta_{ij} - \pi_i + \pi_j}{ \geq 0 \quad }{\forall (i, j) \in E}
		\addConstraint{\pi_s - \pi_t}{\geq 1}{}
		\addConstraint{\delta_{ij}}{\geq 0}{\forall (i, j) \in E}
		\addConstraint{\pi_i}{\geq 0}{\forall i \in V}
        \addConstraint{c_{ij} \delta_{ij}}{\geq \alpha}{\forall(i,j) \notin V_F}
        \addConstraint{c_{ij} \delta_{ij}}{= m((i,j))}{\forall(i,j) \in V_F}
	\end{maxi}

While this procedure yields a polynomial-time algorithm for computing the required imputation, it may require solving up to $|E|+1$ s, which is a notable drawback. In contrast, the combinatorial algorithm presented in Section~\ref{subsec:flow-leximin} typically offers a significant runtime advantage.

\section{Min-cost branching}\label{app:leximaxMST}


\subsection{Finding leximax Owen set imputation}
Similar to the procedure described in section \ref{section:leximinMST}, the following set of LPs can be defined for finding the leximax Owen set imputation. Let $\opt$ be the cost of minimum branching.
\begin{mini}
		{} {\alpha}
			{\label{eq.mst-primal-leximax}}
		{}
        \addConstraint{\sum_{S\in\Vr:e\in\btd(S)}\sum_{v\in S} y(S,v)}{\le c(e)}{\quad\forall e\in E}
		\addConstraint{\sum_{S\in\Vr:v\in S} y(S,v)}{\leq \alpha}{\quad\forall v\notin V_F, v\neq {r}}
  		\addConstraint{\sum_{S\in\Vr:v\in S} y(S,v)}{=  m(v)}{\quad\forall v\in V_F}
        \addConstraint{\sum_{S\in \Vr} \sum_{v\in S} y(S,v)}{=\opt}{}
        \addConstraint{y(S,v)}{\geq 0}{\quad\forall S\in\Vr,\forall v\in S}
\end{mini} 

The dual program is the following:
\begin{maxi}
		{} {\beta\opt -\sum_{e\in E} z(e)c(e) - \sum_{v \in V_F } z(v) \text{m}(v)}
			{\label{eq.mst-dual-leximax}}
		{}
        \addConstraint{\sum_{e\in\btd(S)} z(e)}{\ge -z(v)+\beta}{\quad\forall S\in\Vr, \forall v\in S}
        \addConstraint{\sum_{v\notin V_F, v\neq r} z(v)}{=1}
		\addConstraint{z(e)}{\geq 0}{\quad\forall e \in E}
        \addConstraint{z(v)}{\geq 0}{\quad\forall v \notin V_F, v\neq r }
\end{maxi} 

Note that the dual LP \ref{eq.mst-dual-leximax} contains exponentially many constraints, but it can be solved in polynomial-time using the ellipsoid method, provided we have a polynomial-time separation oracle to identify a violated constraint. The separation oracle in this context is highly similar to the one used for the leximin Owen set version of the problem discussed in Section \ref{section:leximinMST}. The only difference is that, here, the maximum flow from any vertex $v \in \Vr$ to $r$ is compared to $-z(v) + \beta$.

\subsection{More concise series of LPs for leximin Owen set imputation}
\label{app:MST_complex_LPs}
For the specific problem of finding the leximin Owen set imputation, the constraint $\sum_{S \in \Vr} \sum_{v \in S} y(S,v) = \opt$ can be removed from LPs \ref{eq.mst-primal-}, leading to the following more concise series of LPs, where $\mathcal{S}$ is the set of all contiguous sets, i.e. $\mathcal{S} = \{S \in \Vr, |\btd(S)\cap T|=1\}$ for any minimum branching tree $T$.

\begin{maxi}
		{} {\alpha}
			{\label{eq.mst-primal-concise}}
		{}
        \addConstraint{\sum_{S:e\in\btd(S)}\sum_{v\in S} y(S,v)}{\le c(e)}{\quad\forall e\in E}
		\addConstraint{\sum_{S:v\in S} y(S,v)}{\ge \alpha}{\quad\forall v\notin V_F, v\neq {r}}
  		\addConstraint{\sum_{S:v\in S} y(S,v)}{ =  m(v)}{\quad\forall v\in V_F}
        \addConstraint{y(S,v)}{\geq 0}{\quad\forall S\in\mathcal{S},\forall v\in S}
	\end{maxi} 
The dual program is the following:
\begin{mini}
		{} {\sum_{e\in E} z(e)c(e) - \sum_{v \in V_F } z(v) \text{m}(v)}
			{\label{eq.mst-dual-concise}}
		{}
        \addConstraint{\sum_{e\in\btd(S)} z(e)}{\ge z(v)}{\quad\forall S\in\mathcal{S},\forall v\in S}
        \addConstraint{\sum_{v\notin V_F, v\neq r} z(v)}{=1}
		\addConstraint{z(e)}{\geq 0}{\quad\forall e \in E}
        \addConstraint{z(v)}{\geq 0}{\quad \forall v \notin V_F, v\neq r }
\end{mini}

Although LP \ref{eq.mst-dual-concise} has exponentially many constraints, it can be solved in polynomial time using the ellipsoid method, if there is a polynomial-time separation oracle for identifying violated constraints. The relevant constraints are those corresponding to sets $S \in \mathcal{S} $ and $v \in S$. Given values assigned to edges $z(e), e \in E$ and vertices $z(v), v \in V$, we provide a polynomial-time separation oracle to solve this LP efficiently. 

\begin{lemma}\label{lem:oracle}
Given an assignment of values to edges $z(e),e\in E$ and vertices $z(v), v\in V$, there exists a polynomial time algorithm to determine if there exists a contiguous set of vertices, $S$ in $T$ and a vertex $v\in S$ such that $\sum_{e\in\btd(S)} z(e) < z(v)$.
\end{lemma}
\begin{proof}
To prove the lemma it suffices to show a polynomial time algorithm which, given an edge $e=(u,w)\in T$, checks if there exists a set $S$ and a vertex $v\in S$ such that $\btd(S)\cap T=\{e\}$ and $\sum_{e\in\btd(S)} z(e) < z(v)$.

Define $ V' $ as the set of all descendants of vertex $ u $ within the tree $ T $. 
Consider $ N(V') \subseteq V \setminus V' $ to represent the subset of vertices that 
are adjacent to at least one vertex in $ V' $. We denote by $ H' = (V' \cup N(V'), E') $ the graph induced by the vertex set $ V' \cup N(V') $, with edges removed that have both endpoints residing in $ N(V') $. We assign an infinite capacity to 
every edge in $ E' \setminus \{(u,w)\} $ that also exists as an edge in $ T $, while all 
other edges in $ E' $ are given a capacity that matches the $ z $ value of the corresponding 
edge in $ E $. Additionally, for each vertex in $ N(V') \setminus \{w\} $, we introduce an 
edge leading to $ w $ with infinite capacity. 
Figure~\ref{fig:seporacle} demonstrates this configuration with an example where all black edges, and the edge $ (u, w) $, belong to a given tree $ T $. The pink subtree encompasses all descendants of $ u $ within tree $ T $. The blue edges represent those edges that connect at least one endpoint in $ V' $, yet are not part of tree $ T $. Finally, the red dashed edges represent added connections with infinite capacity from each vertex in $ N(V') \setminus \{w\} $ to $ w $.

We next argue that for $v\in V'$ if the max-flow from $v$ to $w$ in $H'$ is less than $z(v)$ then there exists a set $S, v\in S$ such that $\btd(S)\cap T=\{e\}$ and $\sum_{e\in\btd(S)} z(e) < z(v)$. 

Consider such a vertex $v$ and let $S$ be the side of the minimum capacity cut separating $v$ and $w$ that contains $v$. By the max-flow min-cut theorem, the total capacity of edges in $\btd(S)$ is less than $z(v)$. Since $v$ is a descendant of $w$ in the tree $T$, at least one edge on the path from $v$ to $w$ in $T$ should be in $\btd(S)$. All these edges, except $(u,w)$, have infinite capacity and since $\btd(S)$ has a finite capacity, $(u,w)\in\btd(S)$. Again, since $\btd(S)$ has a finite capacity, no other edge of $\btd(S)$ is in $T$ and hence $S$ is a contiguous set of vertices in $T$. Thus the total capacity of edges in $\btd(S)$ equals $\sum_{e\in\btd(S)} z(e)$ which, since it is less than $z(v)$, gives us a violated inequality.

For the converse, note that if there were a set of vertices in the tree, $S$, and $v\in S$, such that $\btd(S)\cap T=\{e\}$ and $\sum_{E\in\btd(S)} z(e) < z(v)$ then the cut $\btd(S)$ would have capacity $\sum_{E\in\btd(S)} z(e)$ in $H'$ and since this is less than $z(v)$, the maximum flow from $v$ to $w$ would be less than $z(v)$.
\end{proof}

\begin{figure}[!ht]
  \centering
  \begin{tikzpicture}
    \tikzset{vertex/.style={circle, draw, fill=black, inner sep=0pt, minimum width=4pt}}
    \tikzset{infinite/.style={green!50!black, densely dashed, line width=1pt, -{Latex[bend]}}}
    \tikzset{finite/.style={blue, line width=1pt}}

    \filldraw[fill=purple!10, draw=black] (0,0.5) -- (-1.5, -2.5) -- (1.5, -2.5) -- cycle;
    \node[vertex] (u) at (0,0) [label=right:$u$] {};
    \node[vertex] (w) at (0,1) [label=above:$w$] {};
    \node[vertex] (v) at (0,-1) [label=below:$v$] {};
    \node[vertex] (k) at (-1,-2) [] {};
    \node[vertex] (z) at (1,-2) [] {};

    \node[vertex] () at (-0.35,0.35) [] {};
    \node[vertex] () at (0.35,0.35) [] {};
    \node[vertex] () at (1.2,-0.5) [] {};
    \node[vertex] () at (-1.2,-0.5) [] {};
    \node[vertex] () at (-2.8,-1.7) [] {};
    \node[vertex] () at (-3.5,-1.9) [] {};
    \node[vertex] () at (2.8,-1.7) [] {};
    \node[vertex] () at (3.5,-1.9) [] {};

    \draw[thick,color=blue] (u) -- (w); 
    \draw[thick, color=black] (u) -- (v) node[midway, right] {$\infty$};
    \draw[thick, color=black] (u) -- (v);
    \draw[thick,color=black] (v) -- (k) node[midway, left] {$\infty$};
    \draw[thick,color=black] (v) -- (z) node[midway, right] {$\infty$};

    \draw[thick, color = blue] (u) -- ++(-0.35,0.35);
    \draw[dashed, thick, color = red] (w) -- (-0.35,0.35);
    
    \draw[thick, color = blue] (u) -- ++(0.35,0.35);
    \draw[dashed, thick, color = red] (w) -- (0.35,0.35);

    \draw[thick, color = blue] (v) -- ++(1.2,0.5);
    \draw[dashed, thick, color = red] (w) -- (1.2,-0.5);

    \draw[thick, color = blue] (v) -- ++(-1.2,0.5);
    \draw[dashed, thick, color = red] (w) -- (-1.2,-0.5);

    \draw[thick, color = blue] (k) -- ++(-1.8,0.3);
    \draw[dashed, thick, color = red] (w) -- (-2.8,-1.7);

    \draw[thick, color = blue] (k) -- ++(-2.5,0.1);
    \draw[dashed, thick, color = red] (w) -- (-3.5,-1.9)node[midway, left] {$\infty$};

    \draw[thick, color = blue] (z) -- ++(1.8,0.3);
    \draw[dashed, thick, color = red ] (w) -- (2.8,-1.7);

    \draw[thick, color = blue] (z) -- ++(2.5,0.1);
    \draw[dashed, thick, color = red] (w) -- (3.5,-1.9) node[midway, right] {$\infty$};
    
    \draw[thick, color = blue] (k) -- (z);

  \end{tikzpicture}
  \caption{Example for separation oracle of MST LP}
  \label{fig:seporacle}
\end{figure}

Let $ y $ be a feasible solution to the LP \ref{eq.mst-org-dual}, and let $ y' $ be a \emph{split} of $ y $. Since $ y(S) = \sum_{v \in S} y'(S, v) $, by a change of variables in the LP \ref{eq.mst-org-dual}, we get the following equivalent LP. Note that the optimal objective of this LP is the cost of the min-cost branching, $ \opt $.

\begin{maxi}
		{} {\sum_{S\subseteq\Vr} \sum_{v\in S} y'(S,v)}
			{\label{eq.mst-split}}
		{}
        \addConstraint{\sum_{S:e\in\btd(S)} \sum_{v\in S} y'(S,v)}{\le c(e)}{\quad\forall e\in E}
		\addConstraint{\sum_{v\in S} y'(S,v)}{\geq 0}{\quad\forall S\subseteq\Vr}
\end{maxi} 
 
\begin{lemma}
    The set of LPs \ref{eq.mst-primal-concise} returns the leximin Owen set imputation. 
\end{lemma}
\begin{proof}
    The set of LPs \ref{eq.mst-primal-concise} finds a leximin solution among the feasible solutions of the LP \ref{eq.mst-split}. However, we argue that a feasible solution that is lexicographically better than all other feasible solutions should also be optimal. Let $ \alpha^*, y^*(S, v) $ be the optimal solution of the last iteration of the LP \ref{eq.mst-primal-concise}, and let $ T $ be any min-cost branching.
    \begin{equation}
    \begin{aligned}
        \opt &= \sum_{e \in T} c(e) \ge \sum_{e \in T} \sum_{S : e \in \btd(S)} \sum_{v \in S} y^*(S, v) \\
        &= \sum_{S : e \in \btd(S)} \left[ \sum_{v \in S} y^*(S, v) \cdot \abs{\btd(S) \cap T} \right] = \sum_{S : e \in \btd(S)} \sum_{v \in S} y^*(S, v)
    \end{aligned}
    \label{eq.imp}
    \end{equation}
    The last equality holds because we have ensured that $ y^*(S, v) > 0 $ only if $ \abs{\btd(S) \cap T} = 1 $. To show that $ y^*(S, v) $ is an optimal solution to LP \ref{eq.mst-split}, we need to show that the inequality is always an equality, i.e., for all $ e \in T $, $ c(e) = \sum_{S : e \in \btd(S)} \sum_{v \in S} y^*(S, v) $. If an edge $ (u, w) $ is not tight, by raising the dual of $ y^*(\{u\}, u) $, we get a lexicographically better feasible solution, which contradicts the fact that we started with the leximin feasible solution.
\end{proof}

\section{Omitted details from \texorpdfstring{$b$}{b}-matching game}
\label{app:b-matching-omitted}

\subsection{Finding leximin Owen set imputation using linear programming}
\label{app:bmatching_lp}
Since LP~\ref{eq.b-uncon-core-dual-bipartite} has polynomially many variables and constraints, it can be solved efficiently. Therefore we need to show how this LP should be modified to give us a series of LPs that determine the leximin Owen set imputation. \\
Note that in this setting, the profits are only assigned to the vertices, each vertex $i \in U$ receives a profit of $p_U(i) = b_i u_i$, and each vertex $j \in V$ receives a profit of $p_V(j) = b_j v_j$. Therefore at every iteration we find the maxi-min profit that should be assigned, and we assign it to a particular set of vertices. The following set of LPs return the leximin Owen set imputations for profits in $b$-matching game. $V_F$ can contain only the vertices in $U \cup V$. The mapping $m: U \cup V \to \mathbb{R}_+$ is defined over the set of all vertices. The procedure stops when profits of all vertices are fixed. 

 	\begin{maxi}
		{} {\alpha} 
			{\label{eq.b-uncon-core-dual-bipartite_leximin}}
		{}
        \addConstraint{\sum_{i \in U}  {b_i u_{i}} + \sum_{j \in V} {b_j v_j}}{ = \opt}{}
		\addConstraint{ u_i + v_j}{ \geq w_{ij} \quad }{\forall (i, j) \in E}
		\addConstraint{b_i u_{i}}{\geq \alpha}{\forall i \in U, i \notin V_F}
		\addConstraint{b_j v_{j}}{\geq \alpha}{\forall j \in V, j \notin V_F}
        \addConstraint{b_i u_{i}}{= m(i), }{\forall i \in U, i \in V_F}
        \addConstraint{b_j v_{j}}{= m(j), }{\forall j \in V, j \in V_F}
	\end{maxi}

\subsection{Omitted definitions}
\label{app:b-matching-definitions}

The algorithmic approach for the $b$-matching setting presented in Section~\ref{sec:b_matching_comb_algo} follows the same principles as in \cite{Vazirani-leximin}. In this section, we will detail on the relevant definitions that were skipped in the section.

First, we will characterize the vertices and edges based on the number of times they are matched in some or all maximum weight $b$-matchings.

\begin{definition} 
    Let $q$ be a vertex in $U \cup V$ and $b_q$ its associated capacity. We classify vertex $q$ as follows:
    \begin{enumerate}
        \item \textit{Essential}: Vertex $q$ is matched $b_q$ times in every maximum weight $b$-matching. 
        \item \textit{Viable}: Vertex $q$ is matched $b_q$ times in some, but not all, maximum weight $b$-matchings. 
        \item \textit{Subpar}: Vertex $q$ is never matched $b_q$ times in any maximum weight $b$-matching. 
    \end{enumerate}
\end{definition}

\begin{definition} 
    Let $y$ be an Owen set imputation. An agent $q$ is said to \textit{get paid} in $y$ if $y_q > 0$, and \textit{does not get paid} if $y_q = 0$. Furthermore, $q$ is \textit{paid sometimes} if there exists at least one imputation in the Owen set under which $q$ gets paid, and \textit{never paid} if it is not paid under any imputation in the Owen set. 
\end{definition}

\begin{definition}
    We will say that an edge $(i, j)$ is \textit{tight} if $u_i + v_j = w_{ij}$, it is \textit{over-tight} if $u_i + v_j > w_{ij}$ and it is \textit{under-tight} if $u_i + v_j < w_{ij}$.
\end{definition}

\begin{definition} 
    Let $e = (i, j)$ be an edge in the graph $G$. We classify $e$ as follows: 
    \begin{enumerate} 
        \item \textit{Essential}: Edge $e$ is matched $ \min(b_i, b_j)$ times in every maximum weight $b$-matching in $G$. 
        \item \textit{Viable}: There exist maximum weight $b$-matchings $M$ and $M'$ such that $e$ is matched $\min(b_i, b_j)$ times in $M$ but not in $M'$.
        \item \textit{Subpar}: Edge $e$ is never matched $\min(b_i, b_j)$ times in any maximum weight $b$-matching in $G$. 
    \end{enumerate} 
\end{definition}

An essential edge is always connected to at least one essential vertex, whereas a viable or subpar edge can be adjacent to any type of vertex. Furthermore, all vertices adjacent to a subpar vertex are essential vertices. These facts, combined with the complementary slackness theorem and strict complementarity conditions, lead to the following lemmas.

\begin{lemma}(\cite{vazirani2023lpduality})
     For every vertex $q \in (U \cup V)$: 
     $$
     q \text{ is paid sometimes } \iff q \text{ is essential}
     $$
\end{lemma}

\begin{lemma} (\cite{vazirani2023lpduality})
    For every edge $e \in E$:
    $$
        e \text{ is always tight }\iff e \text{ is viable or essential}
    $$
\end{lemma}

After construction of the tight subgraph $H=(V,T_0)$, its connected components are split into two categories, unique imputation components and fundamental components, as defined below.

\begin{definition}
\label{def:unique_imputation_component}
    If a connected component of $H_0$ is such that all its vertices get the same profit-shares across all Owen set imputations, it will be called a \textit{unique imputation component}.
\end{definition}

\begin{definition}
\label{def:fundamental_component}
    A connected component of $H_0$ which is not a \textit{unique imputation component} is called a fundamental component.
\end{definition}

As defined above, only vertices in fundamental components can exhibit multiple profit-shares across different Owen set imputations. These fundamental components are pivotal in the algorithm, as the profit-shares of vertices within these components will be updated simultaneously.

\subsection{Combinatorial algorithms to compute leximin Owen set imputation}
\label{app:bmatching_comb}

The leximin Owen set imputation of the $b$-matching game can be derived through a fully combinatorial method. One approach involves reducing the problem to an assignment game on a related graph as follows. Create a new graph, $ G'=(U'\cup V',E') $, from $G=(U\cup V,E)$, where each vertex $ v $ of the original graph is replaced by $ b_v $ copies. If $u\in U \text{ and }v \in V$ are connected by an edge, every copy of $u$ is connected to every copy of $v$ using an edge with the same weight as the original edge.

It can be easily shown that any core imputation on the assignment game on $G'$ corresponds to an Owen set imputation of the $b$-matching game. Specifically, in any core imputation of the new assignment game, every copy of a vertex will have the same profit-share. Therefore, by giving each vertex in the $b$-matching game $b_v$ times the profit-share of one of its copies, we obtain an Owen set imputation.

Furthermore, it can be easily demonstrated that, by starting with any arbitrary imputation of the assignment game on $G'$ and iteratively updating the profits so that the corresponding imputation in the $b$-matching game on $G$ improves lexicographically-- following the algorithm outlined in \cite{Vazirani-leximin}-- we obtain the leximin Owen set imputation.

This straightforward algorithm, however, is pseudo-polynomial because the construction of the new graph $G'$ depends on the $b_v$ values, which can be large. The complexity of the algorithm is primarily determined by the time required to find a maximum weight matching in the associated assignment game, which also depends on the given $b_v$ values. By using the Hungarian algorithm, which runs in $O(mn + n^2 \log n)$ time on a graph with $n$ vertices and $m$ edges, we get the following theorem.

\begin{theorem}
\label{thm:bmatching_DCC_theorem}
    There exists a $O(B^3mn + B^2n^2\log (Bn))$ time combinatorial algorithm to find the leximin or leximax Owen set imputation of the max-weight bipartite $b$-matching game on a graph, $G$ with $n$ vertices and $m$ edges, where $B= max_{x\in U\cup V} b_x$.
\end{theorem}

The above result indicates that the Owen set of the $b$-matching game is closely related to the core of the assignment game. However, to develop an efficient algorithm, we must delve deeper into the combinatorial structure of the problem.

Below, we will first provide a concise overview of the ``primal-dual type'' algorithm from \cite{Vazirani-leximin}, emphasizing the key aspects and the properties that core imputations of the assignment game satisfy. We will demonstrate that these properties extend to the Owen set imputations of the $b$-matching game and argue that, with small modifications, the algorithm can be adapted to obtain the leximin Owen set imputation for the $b$-matching game.

\subsubsection{Overview of the algorithm for assignment game}

At a high level, the approach in \cite{Vazirani-leximin} can be summarized as follows: The method begins by categorizing edges and vertices into three groups—\emph{essential}, \emph{viable}, and \emph{subpar}—based on insights from the complementary slackness conditions of the associated dual. The graph, when restricted to essential and viable edges, is then divided into two types of components: \emph{unique imputation} components and \emph{fundamental} components.

The process starts with an arbitrary dual optimal core imputation, and the dual variables are updated iteratively to ultimately reach the leximin imputation. For unique imputation components, where there is already a unique dual feasible core imputation, the leximin imputation is achieved immediately. For the fundamental components, the focus is on identifying the vertices with the minimum profit-share and attempting to increase their profits. This is done by simultaneously adjusting the profits of all vertices within the same fundamental component, ensuring that the imputation remains dual optimal.

Throughout this process, several events may occur: new vertices may become the ones with the minimum profit-share, subpar edges might get tight, and components may merge. These events are handled appropriately, and at some point, certain vertices' profit-shares can no longer be increased, indicating that the leximin condition has been achieved for those vertices. This iterative process continues until no further improvements can be made, at which point the final imputation is guaranteed to be leximin within the core. 

\subsubsection{Overview of the algorithm for \texorpdfstring{$b$}{b}-matching game}
\label{app:b-matching-algo}

Firstly, let us review the definitions from Appendix~\ref{app:b-matching-definitions}. In the context of $b$-matching, a vertex $q \in U \cup V$ is classified as essential, viable, or subpar if it is matched $b_q$ times in every, some, or no maximum weight $b$-matching, respectively. Similarly, an edge $e = (i, j)$ is classified as essential, viable, or subpar if it is matched $\min(b_i, b_j)$ times in every, some, or no maximum weight $b$-matching, respectively.

As observed by \cite{vazirani2023lpduality}, the complementary slackness theorem implies that Owen set imputations only pay essential vertices. Additionally, across all Owen set imputations, every non-subpar edge is always tight, whereas a subpar edge can be either tight or over-tight. Notably, according to LP~\ref{eq.b-uncon-core-dual-bipartite}, no edge is under-tight in any dual feasible solution. These properties are generalizations from the core imputations in the assignment game.

Next, consider the graph $H = (V, T_0)$, formed by restricting $G$ to include only tight edges, i.e., essential and viable edges. As the algorithm progresses, profits are updated, causing additional edges to become tight. The updated set of tight edges and the resulting tight subgraph are denoted by $T$ and $H = (V, T)$, respectively. The connected components of $H$ are then classified into unique imputation components—components with unique Owen set imputations—and fundamental components—components with non-unique Owen set imputations.

Note that each fundamental component consists solely of viable and essential edges, which must remain tight to ensure the solution's dual optimality. Consider a fundamental component where one of its vertices, say $x \in U$, has the minimum profit-share among all vertices in the initial imputation. To increase this profit-share and thereby lexicographically improve the imputation, we must simultaneously increase the $u_i$ values of all $U$ vertices in the connected component and decrease the $v_j$ values of all $V$ vertices in the connected component at an equal rate.

This synchronized adjustment is crucial because increasing only the $u_i$ value of $x$ would cause the edges incident on it to become over-tight, which is not permissible for essential and tight edges. To maintain the tightness of these edges, we must simultaneously decrease the $v_j$ values of the neighboring vertices. To preserve the tightness of edges incident on these neighbors, we further increase the $u_i$ values of their neighbors, and so on. Consequently, the $u_i$ values of all $U$ vertices within the component are increased, while the $v_j$ values of all $V$ vertices are decreased at an equal rate.

Thus, when the $u_i$ values of all vertices in a fundamental component are increased by a certain amount from one imputation to another, the $v_j$ values must be decreased by the same amount. However, since altering the $u_i$ and $v_j$ values changes their contributions to profits by $b_{u_i}$ and $b_{v_j}$ times, respectively, there is a concern that the weight of the $b$-matching itself might change, potentially rendering it suboptimal. The following lemma will demonstrate that this is not the case.

\lemcomponentsarebalanced*

This lemma demonstrates that, similar to the assignment game, in any fundamental component containing a minimum-profit vertex, we can raise the duals on the side with the minimum-profit vertex and simultaneously decrease the duals on the opposite side at the same rate. This ensures that all non-subpar edges remain tight, and the total profits assigned remain unchanged, thus maintaining both the feasibility and optimality of the dual solution.

Accommodating the generalized definitons, Lemma~\ref{lem:components_are_balanced} demonstrates that an Owen set imputation will remain in Owen set after synchronous changes of dual values in a fundamental component, much like core imputations in assignment game. Consequently, we can apply the same algorithm used to find the leximin core imputation in the assignment game to find the leximin Owen set imputation for the bipartite $b$-matching game.

For the assignment game, we were finding the leximin for the dual vector (which also represented the payoffs for agents), but here, we are finding the leximin for the payoff vector, which moves at a different rate compared to the duals.

The process begins with finding an arbitrary imputation, classifying the edges, vertices, and components, and then organizing each fundamental component into a ``bin.'' The vertices with the minimum profit-share are identified, and their components are removed from the bin and made ``active.'' 

We then increase the $u_i$ (or $v_j$) values of all these minimum-profit vertices so that their profit-shares $b_i u_i$ (or $b_j v_j$) increase at the same rate. When adjusting the dual values of these vertices, duals of all the vertices in its fundamental component update simultaneously, at the same rate. Note that while the profits may change at different rate, Lemma~\ref{lem:components_are_balanced} ensures the total change of profit within every fundamental component is zero. During this process, one of three different events can occur.

\begin{itemize} 
    \item If a vertex from another component also becomes a minimum profit-share vertex, move this component to the active set and begin updating its duals as well. 
    \item If vertices on opposite sides of a component (which may have merged, as described below) become minimum profit-share vertices, no further dual changes are allowed for the vertices and their original fundamental components along the path connecting these minimum profit-share vertices. These fundamental components are considered ``fully repaired''. Initially, all unique imputation components are fully repaired. The fundamental components of this merged component that are not on the path are moved back to the bin and will re-update their duals when they once again have a minimum profit-share vertex. 
    \item If a subpar edge is at risk of becoming under-tight due to changing dual values, add this edge to the set of tight edges and merge the two components it connects—say an active component $C$ and another component $C'$—into a new component $D$. If $C'$ is in the bin, then $D$ will be made active. If $C'$ is in an active state then $D$ will contain two vertices with minimum profit-share on opposite sides. In this case, the fundamental components of $D$ that are on the path between these minimum-profit vertices are marked 'fully repaired' and the rest are moved back to the bin. If $C'$ is fully repaired, then $D$ is moved completely into the fully repaired set. 
\end{itemize}

We do not delve into further details here, as it would largely repeat the discussion from the \cite{Vazirani-leximin}. The same algorithm can be adapted for leximax Owen set imputation by starting with vertices that have the maximum profit-share and then reducing their profits. The process for handling events is analogous to the leximin case. The run time of the algorithms also will remain the same. For complete details of the algorithm, please refer to \cite{Vazirani-leximin}.

\end{document}